\documentclass[pra,aps,floatfix,amsmath,superscriptaddress,twocolumn,longbibliography,nofootinbib]{revtex4-2}
\usepackage{wrapfig, blindtext}
\usepackage[most]{tcolorbox}
\tcbset{colback=yellow!10!white, colframe=red!50!black, 
  highlight math style= {enhanced, 
    colframe=red,colback=red!10!white,boxsep=0pt}
}
\usepackage{amsmath}
\usepackage{tikz-cd}
\usepackage{empheq}
\usepackage{graphicx}

\usepackage{anyfontsize}
\usepackage{hyperref}
\usepackage[capitalise]{cleveref}
\usepackage{nameref}
\usepackage{xpatch}

\usepackage{wrapfig, blindtext}
\usepackage{wrapfig, blindtext}
\usepackage[most]{tcolorbox}
\tcbset{colback=yellow!10!white, colframe=red!50!black, 
  highlight math style= {enhanced, 
    colframe=red,colback=red!10!white,boxsep=0pt}
}
\usepackage{amsmath}
\usepackage{tikz-cd}
\usepackage{empheq}
\usepackage{graphicx}

\usepackage{etoolbox}

\usepackage{anyfontsize}
\usepackage{hyperref}
\usepackage[capitalise]{cleveref}
\usepackage{nameref}
\usepackage{xpatch}

\usepackage{wrapfig, blindtext}
\usepackage[most]{tcolorbox}
\tcbset{colback=yellow!10!white, colframe=red!50!black, 
        highlight math style= {enhanced, 
            colframe=red,colback=red!10!white,boxsep=0pt}
        }

\usepackage{amsmath}
\usepackage{empheq}
 \newlength\dlf  

\usepackage{amsmath,amssymb}
\usepackage{amssymb,enumerate}
\usepackage{graphicx}
\usepackage{graphics}
\usepackage{amsmath}
\usepackage{amsthm,bbm}
\usepackage{color}
\usepackage{dsfont}
\usepackage{hyperref}
\usepackage{bbm}

\usepackage{lipsum}
\usepackage{lmodern}
\usepackage{tcolorbox}

\usepackage{amsfonts}
\usepackage{graphicx,graphics,epsfig,times,bm,bbm,amssymb,amsmath,amsfonts,mathrsfs}
\usepackage[normalem]{ulem}
\usepackage{setspace}
\usepackage{subfigure}
\usepackage{dsfont}
\usepackage{braket}
\usepackage{upgreek }
\usepackage{tikz}
\usepackage{natbib}
\usepackage{chngcntr}

\newtheorem{theorem}{Theorem}

\newtheorem{corollary}[theorem]{Corollary}

\newtheorem{lemma}{Lemma}

\newtheorem{proposition}{Proposition}
\newtheorem{remark}[theorem]{Remark}

\makeatletter
\renewcommand\bibsection{%
  \section*{{\refname}}%
}%
\makeatother

\newcommand{\bes} {\begin{subequations}}
\newcommand{\ees} {\end{subequations}}
\newcommand{\bea} {\begin{eqnarray}}
\newcommand{\eea} {\end{eqnarray}}
\newcommand{\be} {\begin{equation}}
\newcommand{\ee} {\end{equation}}

\def\>{\rangle}
\def\<{\langle}
\def\Tr{\textrm{Tr}}

\newcommand{\ignore}[1]{}

\usepackage{amsmath,amssymb}

\definecolor{dukeblue}{rgb}{0.0, 0.0, 0.61}

\begin{document}	
	\title{Restrictions on realizable unitary operations imposed by symmetry and locality}

	\author{Iman Marvian}
\affiliation{Departments of Physics \& Electrical and Computer Engineering, Duke University, Durham, North Carolina 27708, USA}

	\begin{abstract}
	
	According to a fundamental result in quantum computing, any unitary transformation on a composite system can be generated using so-called 2-local unitaries that act only on two subsystems. Beyond its importance in quantum computing, this result can also be regarded as a statement about the dynamics of systems with local Hamiltonians: although locality puts various constraints on the short-term dynamics, it does not restrict the possible unitary evolutions that a composite system with a general local Hamiltonian can experience after a sufficiently long time. Here we show that this universality does not remain valid in the presence of conservation laws and global continuous symmetries such as U(1) and SU(2). In particular, we show that generic symmetric unitaries cannot be implemented, even approximately, using local symmetric unitaries. Based on this no-go theorem, we propose a method for experimentally probing the locality of interactions in nature. In the context of quantum thermodynamics our results mean that generic energy-conserving unitary transformations on a composite system cannot be realized solely by combining local energy-conserving unitaries on the components.  We  show how  this can be circumvented via catalysis.
		         
	\end{abstract}
	\maketitle

Locality and  symmetry are fundamental and  ubiquitous  properties of physical  systems and their interplay leads to diverse  emergent phenomena, such as spontaneous symmetry breaking. They also put various constraints on both equilibrium and dynamical properties of physical systems. For instance,  symmetry implies  conservation laws, as highlighted by the Noether's theorem \cite{noether1918nachrichten, noether1971invariant}, and locality of interactions implies finite speed of propagation of   information, as highlighted by the Lieb-Robinson bound \cite{lieb1972finite}. Nevertheless, in spite of the restrictions imposed by locality on the short-term dynamics, it turns out that   after a sufficiently long time and in the absence of symmetries,  a composite system with a general local  (time-dependent) Hamiltonian can experience any arbitrary unitary time evolution.  This is related  to a fundamental  result in quantum  computing: any unitary transformation  on a composite system can be generated by a sequence of 2-local unitary transformations, i.e., those  that couple, at most, two subsystems \cite{divincenzo1995two,  lloyd1995almost, deutsch1995universality}.  

In this Letter, we study this phenomenon in the presence of conservation laws and global symmetries. In particular, we
ask whether this universality remains valid in the presence of symmetries, or whether locality puts additional constraints on the possible unitary evolutions of a composite system. Clearly, if all the local unitaries obey a certain symmetry,  then the overall unitary evolution also obeys the same symmetry.  The question is if  \emph{all}  symmetric unitaries on a composite system can be generated using \emph{local} symmetric unitaries on the system. Surprisingly, it turns out that the answer is negative in the case of continuous symmetries, such as SU(2) and U(1).  In fact, we show that  generic  symmetric unitaries cannot be implemented, even approximately, using  {local} symmetric unitaries.  Furthermore, the difference between the dimensions of the manifold of all symmetric unitaries and the sub-manifold of unitaries generated by $k$-local symmetric unitaries with a fixed $k$, constantly increases with the system size.

This result implies  that, in the presence of locality,  symmetries of Hamiltonian   impose extra  constraints on the time evolution of the system, which are not captured by the Noether's theorem. We  show how the violation of these 
constraints can be observed experimentally and, in fact,   can be used as a new method for probing the locality of interactions in nature. 
These additional constraints can also have interesting implications in the context of quantum chaos and thermalization of many-body systems \cite{khemani2018operator}.  We also explain how  in the case of U(1) symmetry, the no-go theorem can be circumvented using ancillary qubits  and discuss the implications of these results in the contexts of the resource theory of 
quantum thermodynamics  \cite{FundLimitsNature, brandao2013resource, janzing2000thermodynamic,  lostaglio2015quantumPRX, halpern2016microcanonical, halpern2016beyond, guryanova2016thermodynamics, chitambar2019quantum}, quantum reference frames \cite{QRF_BRS_07}
 and quantum circuit synthesis.

\section{Preliminaries}

\subsection{Local Symmetric Quantum Circuits (LSQC)}

Consider an arbitrary composite system formed from local subsystems or \emph{sites} (e.g., qubits or spins). In this paper we focus on systems with finite-dimensional Hilbert spaces.   An operator is called $k$-local if it acts non-trivially on the Hilbert spaces of, at most, $k$ sites.    Consider a symmetry described by a general group $G$. To simplify the following discussion, unless otherwise stated,  we assume all sites in the system have identical Hilbert spaces and carry the same unitary representation of group $G$ (In Supplementary Note 1 we consider a more general case). In particular, on a system with $n$ sites, assume each group element $g\in G$ is represented by the unitary $U(g)=u(g)^{\otimes n}$.  An operator $A$ acting on the total system is called  $G$-invariant, or \emph{symmetric}, if  satisfies $U(g) A U^\dag(g)=A$, for any group element $g\in G$.  The set of symmetric unitaries  itself forms a group, denoted by 
\be
\mathcal{V}^{\text{G}}\equiv\{V: VV^\dag=I ,  [V,U(g)]=0, \forall g\in G\}\ ,   
\ee
 where $I$ is the identity operator.

\begin{figure}\label{fig:circuit}
  \includegraphics[scale=.6]{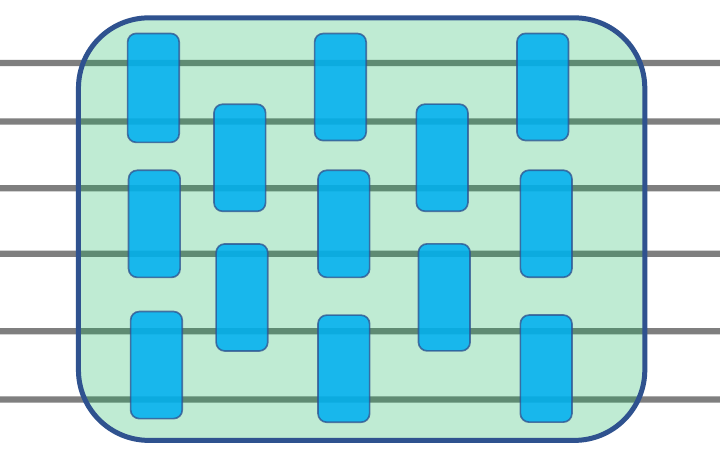}
   \caption{\small{\textbf{Local Symmetric Quantum Circuits.} A quantum circuit with 2-local unitaries on 6 subsystems (e.g., qubits). A Local Symmetric Quantum Circuit (LSQC) only contains local unitaries that respect a certain symmetry. For instance, they are all invariant under rotations around z axis.   Such circuits can model the time evolution of systems with local symmetric Hamiltonians. Conversely, any LSQC corresponds to the time evolution generated by a local symmetric (time-dependent) Hamiltonian. Therefore, by studying LSQC, we can also characterize general features of time evolution under local symmetric Hamiltonians.}}\label{fig90}\end{figure}

 As an example, we consider a system with $n$ qubits and the U(1) symmetry corresponding to global rotations around the $z$  axis. Then, an operator $A$ is symmetric  if  $(e^{-i\theta Z})^{\otimes n} A (e^{i\theta Z})^{\otimes n}=A$, for $\theta\in[0,2\pi)$, or, equivalently, if it commutes with  $\sum_{j=1}^n Z_j$, where  $X_j, Y_j,  Z_j$  denote Pauli operators on qubit $j$ tensor product with the identity operators on the rest of qubits.   Depending on the  context, this symmetry can have different physical interpretations. For instance, if each qubit has Hamiltonian $\frac{\Delta E}{2}  Z$, then $\frac{\Delta E}{2} \sum_{j=1}^n Z_j$ is the total Hamiltonian of the system. Then, unitaries that satisfy this symmetry are the  {energy-conserving} unitaries.

We define $\mathcal{V}_{k}^{G}$  to be the set of all unitary transformations that can be implemented with Local Symmetric Quantum Circuits (LSQC)  with $k$-local unitaries (See Fig.\ref{fig90}). More formally, $\mathcal{V}_{k}^{G}$ is the set of unitaries $V=\prod_{i=1}^{m} V_{i}$,  generated by composing symmetric $k$-local unitaries $V_i: i=1\cdots m$, for a finite $m$.  It can be easily seen that $\mathcal{V}_{k}^{G}$ is a  subgroup of  $\mathcal{V}^{G}=\mathcal{V}_n^{G}$, the group of all symmetric unitaries. More generally, for $k\le l\le n$, we have $\mathcal{V}_{k}^{G}\subseteq \mathcal{V}_{l}^{G}\subseteq \mathcal{V}^G$.  We are interested in characterizing each subgroup 
$\mathcal{V}_k^{G}$ and, in particular,  to determine if there exists $k<n$, such that $k$-local symmetric unitaries become \emph{universal}, that is  $\mathcal{V}_k^G=\mathcal{V}_n^{G}=\mathcal{V}^G$. As we discussed before, in the absence of symmetries, i.e., when $G$ is the trivial group, this holds for $k=2$.  To study these questions we use the Lie algebraic methods of quantum control theory \cite{d2007introduction, jurdjevic1972control}, which  have also   been previously  used  to study the universality of 2-local gates in the absence of symmetries \cite{divincenzo1995two, lloyd1995almost, brylinski2002universal, childs2010characterization,  zanardi2004universal, giorda2003universal,  Bacon:2000qf, lidar1998decoherence}.

It is worth noting that for composite  systems with a given geometry, one can consider the stronger constraint of geometric locality in the above definitions:  the $k$-local symmetric unitaries should act on local neighborhoods, e.g., only on $k$ nearest-neighbor sites.     
However, provided that the sites lie on a connected graph, e.g., on a connected 1D chain, adding this additional constraint does not change the generated group  $\mathcal{V}_{k}^{G}$. This is true because  the swap unitary that exchanges the states of two nearest-neighbor sites is 2-local and respects the symmetry, for all symmetry groups. If the graph  is connected, by combining these 2-local permutations on pairs of neighboring sites, we can generate all permutations and hence change the order of sites arbitrarily. Therefore, any $k$-local symmetric unitary can be realized by a sequence of $k$-local symmetric unitaries on $k$ nearest-neighbor sites. 

\subsection{Time evolution under local symmetric Hamiltonians}\label{Sec:132} 
  
 Next, we consider a slightly different formulation of this problem in terms of the notion of local symmetric Hamiltonians.    
A generic \emph{local} Hamiltonian $H(t)$ acts non-trivially on all subsystems in the system, but,  it has a decomposition as $H(t)=\sum_j h_j(t)$,    where  each term $h_j(t)$ is $k$-local for a fixed $k$, which is often much smaller than the total number of subsystems in the system.   The unitary evolution generated by this Hamiltonian is determined by the  Schr\"{o}dinger equation
\be\label{eq:sch2020}
\frac{d V(t)}{dt}=-i H(t) V(t)=-i\  \big[\sum_j h_j(t)\big]\  V(t)\ ,
\ee
with the initial condition $V(0)=I$. 
 Suppose, in addition to the above locality constraint, the Hamiltonian $H(t)$ also respects the symmetry described by the group $G$, such that $[U(g), H(t)]=0$, for all $g\in G$, and all $t\ge 0$.
Then, it can be shown that the family of unitaries $\{V(t): t\ge 0\}$  generated by any such  Hamiltonian belongs to $\mathcal{V}_{k}^{G}$, i.e., the group of symmetric unitaries that can be implemented by $k$-local symmetric unitaries    (See Supplementary Note 1). Conversely, any unitary in this group  is generated by a Hamiltonian $H(t)$  satisfying the above locality and symmetry constraints (any quantum circuit can be thought of as the time evolution generated by a time-dependent local Hamiltonian). Therefore, by characterizing $\mathcal{V}_{k}^{G}$  and studying its relation with the group of all symmetric unitaries $\mathcal{V}^{G}$, we can also unveil possible constraints on the time evolution  under local symmetric Hamiltonians, which are not captured by the standard conservation laws imposed by the Noether's theorem.

 \section{Main Results}

 \subsection{A no-go theorem:  Non-universality of local unitaries in the presence of symmetries}\label{Sec:I}

We show that in the case of continuous symmetries such as  U(1) and SU(2), most symmetric unitaries cannot be implemented, even approximately, using local symmetric unitaries: First, as we prove in Supplementary Note 1, for any group $G$, the set of symmetric  unitaries $\mathcal{V}^{G}=\mathcal{V}_n^{G}$ and its subgroup $\mathcal{V}^{G}_k$  generated by $k$-local symmetric unitaries,  are both connected  compact Lie groups, and hence closed manifolds (See Fig.\ref{Torus2}).  This means that if a unitary $V$ is not in $\mathcal{V}^{G}_k$, then there is  a  neighborhood of  symmetric unitaries around $V$, none of which can be implemented using $k$-local symmetric unitaries. On the other hand, if $V$ belongs to  $\mathcal{V}^{G}_k$, then it can be  implemented with a uniformly finite number of such unitaries,  that is upper bounded by a fixed number  independent of $V$ \cite{d2007introduction}.

Secondly, we prove that for any  finite or compact Lie group $G$,  the difference between the dimensions of the manifolds associated to all symmetric unitaries $\mathcal{V}^{G}=\mathcal{V}^{G}_n$ and its sub-manifold $\mathcal{V}^{G}_k$   is lower bounded by 
\be\label{Eq:dim3}
\text{dim}(\mathcal{V}^{G})-\text{dim}(\mathcal{V}^{G}_k)\ge |\text{Irreps}_G(n)|-|\text{Irreps}_G(k)|\ ,
\ee
where for any integer $l$, $|\text{Irreps}_G(l)|$ is the number of inequivalent irreducible representations (irreps) of group $G$, appearing in the representation  $\{u(g)^{\otimes l}: g\in G\}$, i.e., in   the action of symmetry on $l$ subsystems.  We conclude that, unless $|\text{Irreps}_G(n)|=|\text{Irreps}_G(k)|$, there is a family of symmetric unitaries on $n$ subsystems that cannot be implemented with $k$-local symmetric unitaries. In the case of continuous symmetries such as U(1) and SU(2), $|\text{Irreps}_G(n)|$ grows unboundedly with $n$. This means that there is no fixed integer  $k$, such that $k$-local symmetric unitaries become universal for all  system size  $n$. This is in a sharp contrast with the universality of $2$-local unitaries in the absence of symmetries.      
In Methods we provide a simple proof of the non-universality of local unitaries in the case of continuous symmetries using a technique called  \emph{charge vectors}.  
 In Supplementary Note 2 we 
prove Eq.(\ref{Eq:dim3}) and present  a more refined version of this inequality in the case of connected Lie groups, such as U(1) and SU(2), as well as an extension of the no-go theorem to the case where the subsystems can have different representations of the symmetry. We also discuss more about the nature of the constraints imposed by locality that lead to the bound in Eq.(\ref{Eq:dim3}) (Namely,  we argue that certain elements of the center of the Lie algebra of symmetric Hamiltonians cannot be generated using local symmetric  Hamiltonians).

\begin{figure}
{\includegraphics[scale=.6]{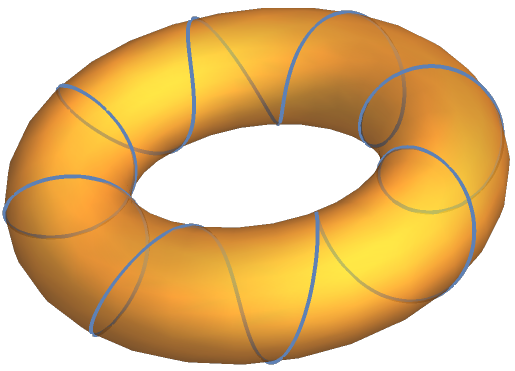}}
\caption{\textbf{The schematic relation between the group of all symmetric unitaries (the torus) and the subgroup generated by Local Symmetric Quantum Circuits (the blue curve)}. They are both  compact connected Lie groups and hence closed manifolds. Unitary evolution under any  local symmetric Hamiltonian is restricted to the submanifold corresponding to LSQC. In other words, adding a perturbation to the Hamiltonian can bring the evolution outside this submanifold, only if it is non-local or symmetry-breaking. In the example of U(1) symmetry, we discuss a more explicit  interpretation of this schematic figure. 
}\label{Torus2}
\end{figure}

 
 \subsection{Example: U(1) symmetry for systems of qubits}

 Recall the example of the U(1) symmetry for a system  of $n$ qubits. In this case,   
the representation of symmetry on $n$ sites is  $(e^{i \theta Z})^{\otimes n}=\exp(i\theta [n I-2 N])$ for $\theta\in [0,2\pi)$, where $N=\sum_j ({I}-Z_j)/{2}$ determines the total \emph{charge} (or, excitations) in the system. It follows that the irreps of U(1) can be labeled by distinct eigenvalues of $N$, which take integer values $m=0,\cdots, n$.  Then,   Eq.(\ref{Eq:dim3}) 
implies that  for a system with $n$ qubits the 
 difference between the dimensions of the manifold of all symmetric unitaries and those generated by $k$-local symmetric unitaries   is, at least,  $n-k$. 
Remarkably, it turns out that in this case this  bound holds as equality. In Methods  we present a full characterization of   Hamiltonians that can be generated using $k$-local U(1)-invariant Hamiltonians.   This result, for instance, implies that 
 even if one can implement all U(1)-invariant unitaries that act on $n-1$ qubits, still the unitary $\exp({i\phi  Z^{\otimes n}})$ cannot be implemented for generic values of $\phi$.  
 
It is useful to express the  constraints  imposed by the locality of interactions  in terms of experimentally observable quantities. Consider a general U(1)-invariant unitary $V$ on $n$ qubits.  For instance, $V$ can be the unitary generated by U(1)-invariant Hamiltonian $H(t)$, from time $t=0$ to $T$ under the Schr\"{o}dinger equation.   Any such unitary has a decomposition as $V=\bigoplus_{m=0}^n V_m$, where
$V_m$ is the component of $V$ in the charge sector $m$, i.e., the eigen-subspace of operator $N=\sum_j ({I}-Z_j)/{2}$ with eigenvalue $m$. For any integer $l=0,\cdots, n$, define   the $l$-body phase $\Phi_l\in(-\pi,\pi]$  of $V$ as  
\be\label{ary4}
 \Phi_l\equiv    \sum_{m=0}^nc_l(m) \theta_m=-\hspace{-1mm}\int_0^T\hspace{-2mm} dt  \hspace{-2mm}\sum_{\substack{{\bf{b}}: w({\bf{b}})=l}} \hspace{-
 2mm} \Tr(H(t) {\bf{Z}}^{\bf{b}})\ \ \  \ \text{: mod} \ 2\pi\ ,
\ee
where  $\theta_m=\text{arg}(\text{det}(V_m))\in(-\pi,\pi]$ is the phase of the determinant of $V_m$,  $c_l(m)=\sum_{s=0}^m (-1)^s  {{m}\choose{s}} {{n-m}\choose{l-s}}$ is an integer coefficient, and we use the convention that  for integers $a$ and $b$, the binomial coefficient ${{a}\choose{b}}=0$ if $b>a$.  
   In the second equality the summation is 
 over all bit strings $\textbf{b}=b_1\cdots b_n\in\{0,1\}^n$ with Hamming weight $w(\textbf{b})\equiv \sum_{j=1}^n b_j$ equal to $l$, and  we have defined 
 ${\bf{Z}}^{\bf{b}}\equiv Z_1^{b_1}\cdots Z_n^{b_n}$.  Note that this  equality is satisfied for any U(1)-invariant Hamiltonian $H(t)$ that realizes unitary $V$.  Using this equality, for instance, we can see that for unitary $V=\exp({i\phi {\bf{Z}}^{\bf{b}}})$, 
 all $l$-body phases vanish, except for $l=w(\textbf{b})$, where  $\Phi_{w(\textbf{b})}=2^n \phi \text{: mod} \ 2\pi$.  In Supplementary Note 3  we prove Eq.(\ref{ary4})  and  present  coefficients $c_l(m)$ for a system with $n=5$ qubits.

The notion of $l$-body phases provides a useful characterization  of  the constraints imposed by the locality of interactions. In Supplementary Note 3 we show that:  $\textbf{(i)}$ for $l\ge 1$, the $l$-body phases $\{\Phi_l\}$ of a U(1)-invariant unitary time evolution can be  measured experimentally. On the other hand, the phases $\{\theta_m\}$ are not physically observable, because they 
 transform non-trivially under the global phase transformation $V\rightarrow e^{i\alpha} V$. Similarly,  $\Phi_0=\sum_m\theta_m=\text{arg}(\text{det}(V))$ is not  observable.  $\textbf{(ii)}$ If a unitary is realizable by $k$-local U(1)-invariant unitaries, then its $l$-body phases are zero for $l>k$, which can be seen using the second equality in Eq.(\ref{ary4}).  This, for instance, implies  that unless  $\phi$ is an integer multiple of ${\pi}/{2^{n-1}}$,  unitary $\exp({i\phi {\bf{Z}}^{\bf{b}}})$  cannot  be  implemented using $k$-local U(1)-invariant unitaries with $k<w(\textbf{b})$. $\textbf{(iii)}$  Conversely, for a general  U(1)-invariant unitary  $V$, if all $l$-body phases vanish for $l>k$, then  $V$ is realizable using $k$-local U(1)-invariant unitaries, up to a unitary in a fixed finite subgroup of  U(1)-invariant unitaries. Finally, it is  worth mentioning that from a geometrical point of view, the transformation $\{\theta_m\} \rightarrow \{\Phi_l\}$ in Eq.(\ref{ary4})  describes  a change of the coordinate system on the  $(n+1)$-torus corresponding to phases $\theta_m=\text{arg}(\text{det}(V_m))$, for  charges $m=0,\cdots, n$. For instance, when the system evolves under the Hamiltonian $H=\gamma   \textbf{Z}^{\textbf{b}}$, its  trajectory on this torus is a helix,  described by equation $\Phi_l(t)=- 2^n \gamma t\times  \delta_{l,w(\textbf{b})}$, where $\delta$ denotes the Kronecker delta (See Fig.\ref{Torus2}). 

In Sec.\ref{Sec:phase}, we discuss  an application of this framework for  synthesizing phase-insensitive quantum circuits. But, first we start with a rather surprising implication of these ideas. 
  \begin{figure}{\includegraphics[scale=.29]{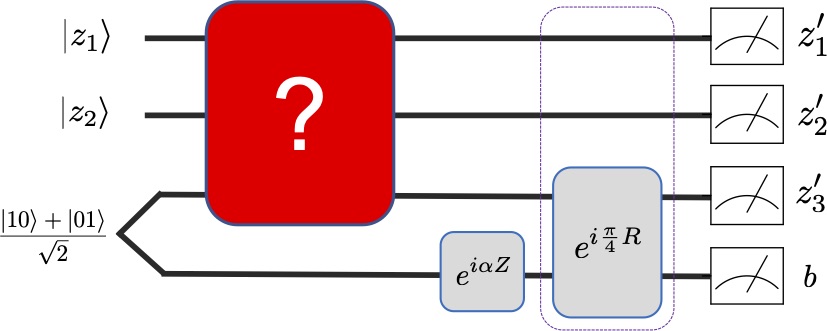}}\caption{
  \textbf{A scheme for local symmetric process tomography and  measurement of $l$-body phases.}  
  The no-go theorem found in this paper has an immediate useful  implication: it gives a new method for detecting the locality of the underlying interactions that govern a charge-conserving unitary process. 
  Specifically, by measuring the  $l$-body phase of the unitary, as defined in Eq.(\ref{ary4}), we can detect $l$-body interactions.   This figure presents a schematic experimental setup that fully characterizes an unknown U(1)-invariant unitary and   its $l$-body  phases, using  initial  states, single-qubit measurements and  2-local unitaries, which all respect the symmetry.     
  In this example, 
the red  box corresponds to an unknown 3-qubit charge-conserving unitary  $V$.  The      
 goal is to measure  the 3-body phase  $\Phi_3\in (-\pi,\pi]$. Observing $\Phi_3\neq 0$ indicates the presence of the 3-body interaction $Z^{\otimes 3}$.   At the input of  $V$ all the qubits  
 are  prepared in unentangled symmetric states $|z\rangle$ with $z=0,1$, except one of them,  which is entangled  with an ancillary  qubit, in the joint state $(|01\rangle+|10\rangle)/\sqrt{2}$.   This ancillary  qubit plays the role of an internal quantum reference frame \cite{QRF_BRS_07} and allows us to probe the relative phases between sectors with different charges through an interference experiment. 
After the unknown unitary $V$,  we apply the single-qubit unitary $\exp({i\alpha Z})$ on the ancillary qubit, then interact it with one of the three qubits in the system via 2-local unitary $\exp(i\frac{\pi R}{4})$, where $R=(XX+YY)/2$, and finally measure all qubits in $\{|0\rangle, |1\rangle\}$ basis. As we discuss further in Supplementary Note 5, using this scheme we can fully characterize  the unknown unitary $V$, up to a global phase and, in particular,   determine the 3-body phase $\Phi_3$.  }\label{Fig:gha}
\end{figure}

\subsection{Application: Probing the locality of  interactions in nature}

Our no-go theorem  leads us to a new method for experimentally probing the locality of interactions: According to this theorem, in the presence of symmetries,   interactions that couple more subsystems can imprint certain observable effects on the time evolution of the system that cannot be reproduced by those that act on smaller number of subsystems.   Therefore, by probing these effects, we can directly obtain information  about the locality of the underlying interactions that govern the process.  This is analogous to the fact that in the presence of symmetries  we can detect a hypothetical symmetry-breaking interaction, just by observing  the violation of Noether's conservation law for the input and output of the process, without knowing the details of the underlying interactions (In our case,  the hypothetical term is not symmetry-breaking; rather it couples multiple subsystems  together).  
 
As a simple example, consider  a system of $n$ qubits evolving  for the total time $T$  under an unknown Hamiltonian $H(t)$ that preserves  $\sum_j Z_j$.  To have a concrete example, one can assume $H(t)$ models the interactions in  a complex scattering  process with $n$ particles, and  states  $\{|0\rangle, |1\rangle\}$ of  each qubit corresponds to an internal degree of freedom of a particle, e.g., its electric charge, whose total value remains conserved  in the process.  Suppose we want to characterize  the locality of interactions that govern this process. For instance,  we start with the hypothesis that $H(t)=H_0(t)+\gamma(t) Z^{\otimes n}$, where $H_0$ only contains $k$-local  terms with $k<n$ and  $\gamma Z^{\otimes n}$ corresponds to a hypothetical $n$-body interaction, e.g., a correction to the Coulomb law.  The goal is to test the hypothesis that the $n$-body term $\gamma Z^{\otimes n}$ is zero, by probing the output of this process for different input states.     
 Note that in the absence of symmetries, unless there are further assumptions about the form of $H_0$, it is impossible to obtain information about the strength of $\gamma$: 
 the universality of 2-local unitaries means that  even if  $\gamma=0$,  the Hamiltonian $H_0$ with 2-local  interactions can generate any arbitrary unitary transformation. Therefore, by probing the outputs 
 of this process for different inputs, we cannot distinguish  the cases of  $\gamma=0$ and $\gamma\neq 0$.

While this is impossible in the absence of symmetries,  our result reveals that symmetries allow us to directly probe the locality of interactions that govern a process, just by observing the inputs and outputs of the process.  This can be achieved systematically by measuring the $l$-body phases of the unitary process for $l\ge 1$.  For instance, in the above example,  
 by measuring the $n$-body phase $\Phi_n\in(-\pi,\pi]$ of the unitary $V$ that describes the overall process,  we obtain a lower bound on  $\gamma_{\text{max}}=\max_{t\in[0,T]} |\gamma(t)|$, that determines 
the maximum strength of the $n$-body interaction; namely, 
\be
\gamma_{\text{max}} \ge \frac{\big|\Phi_n \big|}{   2^n \times  T}\ ,
\ee
 where we have applied the second equality in Eq.(\ref{ary4}).  Note that according to the first equality in Eq.(\ref{ary4}), $\Phi_n=\sum_{m=0}^n (-1)^m \theta_m\ (\text{mod } 2\pi)$.  
 
How can we  measure $l$-body phases of a unitary? More generally, is it possible  to characterize a U(1)-invariant unitary transformation and perform  process tomography \cite{NielsenAndChuang}, using only local symmetric  operations? We find that, despite our no-go theorem on realizable unitaries,  the answer is affirmative. A general U(1)-invariant unitary can be fully characterized, up to a global phase,  using symmetric initial states, symmetric single-qubit measurements, and 2-local symmetric unitaries, provided that one can use a single ancillary qubit, which is initially entangled with one of the qubits in the system.  In particular, the scheme presented in Fig.\ref{Fig:gha} does not require preparing superpositions of states with different charges, which might be impractical due to the superselection rules (See Supplementary Note 5 for further discussion).

  \begin{figure}
{\includegraphics[scale=1.05]{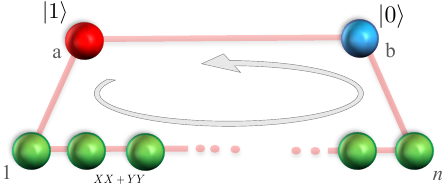}}\caption{
\textbf{Circumventing the no-go theorem with ancillary qubits.}  Our no-go theorem implies that the family of unitaries generated by the $n$-qubit Hamiltonian $Z^{\otimes n} $ cannot be implemented using  local U(1)-invariant unitaries, even if they act on $n-1$ qubits. In this figure, we describe a scheme for circumventing this no-go result, using two ancillary qubits.  This scheme uses     
interaction $R=(XX+YY)/2$  between  nearest-neighbor qubits on a closed loop. 
The two ancillary qubits, denoted by a and b are initially prepared in states $|1\rangle$ and $|0\rangle$, respectively.  First, we show that 
 it is possible to realize the Hamiltonian $K=Z^{\otimes n}\otimes R_{\text{a},\text{b}} $  without any direct interaction between the ancillary qubits. This only requires  coupling qubit a to qubit $j=1$ in the chain, coupling between nearest-neighbor qubits in the chain, and coupling between qubit $j=n$ and ancilla b.   
   This Hamiltonian describes the  process in which a charge is transported through the chain from one ancillary qubit to the other and obtains a phase  depending on the parity of the total charge in the system. As we explain in Supplementary Note 6,  this has an intuitive interpretation in the fermionic description of this system, obtained by applying the Jordan-Wigner transform.   
After evolving the entire system for a short time interval  $\delta t$ under Hamiltonian $K$, we obtain the joint state  $|\psi\rangle|1\rangle_{\text{a}}|0\rangle_{\text{b}}- i\delta t Z^{\otimes n} |\psi\rangle|0\rangle_{\text{a}}|1\rangle_{\text{b}}+\mathcal{O}(\delta t^2)$, where $|\psi\rangle$ is the initial state of $n$ qubits. Next, we directly couple $a$ to $b$ and close the loop,  using the 2-local unitary $\exp{(i\pi R_{\text{a},\text{b}}/4)} \exp{(i\pi Z_b/4)}$ that allows 
the charge to move back and forth between the ancillary qubits, without going through the chain. Finally,  we measure one of the ancillary qubits in $\{|0\rangle,|1\rangle\}$ basis. This  determines the final location of the charge  initially located in qubit $a$.      
The  final state of $n$ qubits is $\exp({\pm i \delta t Z^{\otimes n}})|\psi\rangle+\mathcal{O}(\delta t^2)$, where the sign depends on whether the final location of charge is qubit a or b. Therefore, this process stochastically implements the Hamiltonian $\pm Z^{\otimes n}$. In principle, by 
choosing infinitesimal time step $\delta t$ and repeating this scheme many times, we can implement the desired unitary $\exp(i\phi Z^{\otimes n})$  for arbitrary angle $\phi$, with an error approaching zero and probability of success approaching one. We show that a slightly more complicated version of this scheme  can be realized deterministically. \color{black}   
}\label{ancilla:newFig}
\end{figure}

\subsection{Circumventing the no-go theorem with ancillary systems}
\vspace{-3mm}
Interestingly, it turns out that in the case of U(1) symmetry our no-go theorem can be circumvented,  provided that one is allowed to interact with an ancillary qubit:  for any $n$-qubit U(1)-invariant unitary $V$, there exists  $(n+1)$-qubit unitary $\tilde{V}$ that can be implemented  using  2-local U(1)-invariant Hamiltonians $XX+YY$ and local $Z$, and satisfies  
\be
\tilde{V}\big(|\psi\rangle\otimes |0\rangle_{\text{a}}\big)=(V|\psi\rangle)\otimes |0\rangle_{\text{a}}\ ,
\ee
for all $n$-qubit states $|\psi\rangle$.  
This means that, while by applying local symmetric unitaries  
the ancillary qubit becomes entangled with the qubits in the system, at the end of the process  it returns back to its initial state $|0\rangle$, whereas  the state of system transforms as the desired unitary $V$.

  Fig.\ref{ancilla:newFig} demonstrates a   variant of this result that requires 2 ancillary qubits.   In this example the goal is to implement the unitaries generated by Hamiltonian $Z^{\otimes n}$. Roughly speaking, in this scheme  a charge  is transported 
  through a closed loop that starts from an ancillary qubit, goes through the entire system and finally returns back to the  ancilla. As a result, the joint state obtains a phase depending  on the parity of the total charge in  the system, which  corresponds to  the  observable $Z^{\otimes n}$. The overall effect is equivalent to applying the desired Hamiltonian $Z^{\otimes n}$ on the system. Here, the ancillary qubits can be interpreted as an internal quantum reference frame \cite{QRF_BRS_07}, relative to which the phase-shift generated by observable $Z^{\otimes n}$ is measured in a coherent fashion.   As we further explain  in Supplementary Note 6, this process has also a nice interpretation in the fermionic description of the system, obtained by applying the Jordan-Wigner transform \cite{jordan1928pauli, fradkin1989jordan, nielsen2005fermionic}. 
  
  \color{black}
  

\subsection{Application: Quantum thermodynamics with local interactions}\label{Sec:III}

Our surprising no-go theorem has also interesting    implications in the context of  quantum thermodynamics and, specifically, the operational  approach to thermodynamics, which is often called the \emph{resource theory} of quantum thermodynamics   \cite{FundLimitsNature, brandao2013resource, janzing2000thermodynamic,  lostaglio2015quantumPRX, halpern2016microcanonical, halpern2016beyond, guryanova2016thermodynamics, chitambar2019quantum}. 
  A fundamental assumption in this framework is that all energy-conserving unitaries, i.e., those commuting with the   intrinsic Hamiltonian of the system, are \emph{free}, that is, they  can be implemented with negligible thermodynamic costs. This is assumed even for composite systems with arbitrarily large number of  subsystems. 
 However, our result  implies 
that general energy-conserving  unitaries 
on a composite system cannot be implemented by applying \emph{local} energy-conserving unitaries on the subsystems.    In fact,  even by composing energy-conserving unitaries that act on $n-1$ subsystems, one still cannot generate all energy-conserving unitaries on $n$ subsystems. Note that energy-conserving unitaries are  those  that are invariant under the time-translation symmetry $\{e^{-i H_0 t}: t\in \mathbb{R}\}$ generated by the  intrinsic Hamiltonian $H_0$; a \emph{continuous}  symmetry, which is isomorphic to the group U(1) in the case of periodic systems.


    Therefore, this no-go theorem  suggests that there might be some \emph{hidden} thermodynamic costs for implementing general energy-conserving unitaries, using  \emph{local} energy-conserving unitaries and, in principle, this additional cost can increase with the system size.   The following theorem addresses this concern 
 (See Supplementary Note 7  for a more precise  statement).\\

\noindent\textbf{Theorem:} Consider a finite set of closed systems  with the property that for each system the  gap between any consecutive pairs of energy levels is $\Delta E$. Then, any global energy-conserving unitary transformation on these systems can be implemented by a finite sequence of  2-local energy-conserving unitaries, provided that the systems can interact with a single ancillary qubit  with the energy gap $\Delta E$ between its two levels.  \\

To establish this result, we introduce a generalization of the scheme introduced in the previous section for qubit systems with U(1) symmetry.  We conclude that  the assumption of the resource theory of quantum thermodynamics  \cite{janzing2000thermodynamic, FundLimitsNature, brandao2013resource, guryanova2016thermodynamics, lostaglio2015quantumPRX, halpern2016microcanonical, halpern2016beyond} that all energy-conserving unitaries (and, hence all thermal operations) are \emph{free}, is consistent with the locality of interactions,  provided that one allows  the use of ancillary systems. 
 In the context of quantum thermodynamics such systems can be interpreted as  catalysts  \cite{chitambar2019quantum, jonathan1999entanglement}. It is  worth mentioning that the  assumption in this theorem on the energy gap $\Delta E$ between consecutive levels can be relaxed, provided that one can use larger catalysts with more energy levels.  

 \subsection{Application: Synthesizing noise-resilient    quantum circuits }\label{Sec:phase}

Another motivation to study local symmetric quantum circuits comes from the field of quantum computing and, specifically, the desire to design fault-tolerant quantum circuits.  
In both prominent implementations of quantum computers, namely superconducting and trapped-ion computers, the instability of the master clock that determines the timing of the control pulses is a major source of noise  \cite{ball2016role, bermudez2017assessing}. Each qubit in these systems has a non-zero intrinsic Hamiltonian, which corresponds to an energy difference between states $|0\rangle$ and  $|1\rangle$. Hence, the state of qubit is constantly evolving in time. Ideally, using a stable clock one can keep track of this intrinsic time evolution. 
In other words, one can assume quantum computation is performed in a co-rotating frame, where there is no energy difference between $|0\rangle$ and  $|1\rangle$.  
In practice, however, due to the instabilities of the clock,  this intrinsic time evolution of qubits causes error and destroys coherence between states with different energies. 
For instance, if there is a random time delay $\delta t$ in applying the control pulses that implement a desired unitary transformation $V$, then the actual  implemented unitary in the co-rotating frame will be  $\exp(i\delta t H_0)V\exp(-i\delta t H_0)$, where $H_0=-\Delta E\sum_j Z_j/2$  is the total intrinsic Hamiltonian of the qubits.   In principle,  this effect can be suppressed by restricting the state of qubits to an energy eigen-subspace,  which is  a decoherence-free subspace \cite{lidar1998decoherence, Bacon:2000qf, Zanardi:97c}. But, this amounts to  sacrificing a fraction of physical qubits.  Given the limited number of available qubits on near-term quantum computers,  it is  crucial to explore  other complementary  techniques.

One approach for suppressing this type of noise is to minimize the use of  non-energy-conserving unitaries in the circuit. That is the circuit should be mostly formed from local energy-conserving unitaries.  This includes energy-conserving elementary  gates, such as single-qubit rotations around z, as well as energy-conserving  multi-qubit modules, which may contain  non-energy-conserving  elementary gates.  As long as  the entire module can be  executed in a sufficiently short time during which the clock fluctuations are negligible,  then the energy conservation  of the module guarantees its  resilience against this type of noise. For example, while the 
 standard M{\o}lmer-S{\o}rensen gate~\cite{Molmer1999} 
 $\exp({i\theta X X})$ on trapped-ion quantum computers, is not energy-conserving and hence is sensitive to these  fluctuations,  when it is sandwiched between Hadamards on both qubits, it transforms to $\exp({i \theta Z Z})$, which is energy-conserving.  Similarly, by combining two M{\o}lmer-S{\o}rensen gates  with single-qubit phase gates, we obtain   $\exp({i\theta (X  X+Y Y)})$, which is again energy-conserving.

The tools and ideas introduced in this paper provide a foundation for the   systematic synthesis of  quantum circuits that are resilient against this type of noise. To minimize the number of non-energy-conserving unitaries, the first step is to determine which unitaries can be efficiently realized using local energy-conserving modules. 
 As an example, consider the family of unitaries  generated by the multi-qubit swap Hamiltonian:  Suppose a system with $2r$ qubits is partitioned to two subsystems $\text{A}$ and $\text{B}$, each with $r$ qubits. Let $S_{\text{AB}}$ be the  multi-qubit swap operator
 that exchanges the states of A and B. The family of unitaries $\exp(i \phi S_{\text{AB}})$ for $\phi\in[0,2\pi)$ appears 
 as a subroutine  in various quantum algorithms  (See, e.g.,\cite{lloyd2014quantum,  marvian2016universal, kimmel2017hamiltonian, pichler2016measurement}). It has also found applications in the study of quantum reference frames and quantum thermodynamics \cite{marvian2016universal, popescu2018quantum}.     The multi-qubit swap Hamiltonian $S_{\text{AB}}$ is not only energy-conserving, but in fact it respects the stronger  SU(2) symmetry, i.e.,  $[S_{\text{AB}}, U^{\otimes 2r}]=0$ for all single-qubit unitaries $U$. Therefore, one may  expect that this family of unitaries  should be realizable using a sequence of local SU(2)-invariant  unitaries or, at least, using local energy-conserving  unitaries, which may break the SU(2) symmetry. However, our results refute this conjecture: For generic values of $\phi$, all the $l$-body phases of the unitary $\exp(i \phi S_{AB})$  are non-zero (e.g., $\Phi_{2r}=2^{r} \phi: \text{mod } 2\pi$), which means this unitary is not realizable using local energy-conserving unitaries.   On the other hand, 
 if one is allowed to use a single ancillary qubit, then this family is realizable using   single-qubit rotations around z together with unitaries $\exp(i\theta(XX+YY))$, which, as discussed above, can be obtained from two M{\o}lmer-S{\o}rensen gates. Therefore, to implement a quantum algorithm that employs this subroutine, this part of the circuit can be  realized using only  energy-conserving  modules. This makes the entire circuit more resilient  against clock fluctuations.

\section{Discussion}

Universality of local unitaries in the absence of symmetries is a profound fact about composite  quantum-mechanical systems,   with  vast applications and implications in different areas of physics. Hence, the failure of universality in the presence of symmetries can also have interesting and unexpected implications in different areas.  Here, we saw an example of such surprising   implications, namely the  possibility of  probing the locality of interactions. We end with a brief discussion about  other examples of applications of these results and the related open questions:\\

 \noindent \textbf{Quantum Reference Frames and Covariant Codes:}  Symmetric unitaries also naturally appear in the study of quantum reference frames  \cite{QRF_BRS_07}.   For instance,  it is often assumed that in the absence of a Cartesian reference frame, it is still possible to perform any unitary that respects SO(3) symmetry group corresponding to rotations in 3D space \cite{QRF_BRS_07}.  The no-go theorem found in this paper implies that if one takes into account the locality of interactions, then there can be further restrictions on the realizable unitaries.
It will be interesting to study possible implications of these additional constraints in the context of quantum reference frames.  As an example, Ref. \cite{marvian2008building} shows that  arbitrary symmetry-breaking Hamiltonians on a system  can be simulated  by coupling the  system via rotationally-invariant Hamiltonians 
to  $n\gg 1$ spin-half systems aligned in x and z directions. Therefore, in the limit of large $n$, this quantum reference frame fully lifts the constraint of symmetry.  
It is interesting to further study the efficiency and complexity of such schemes when the Hamiltonians are restricted to be local.

Similar question also arises in the context of covariant error correction, which has recently gained attention in quantum information community (See, e.g.,  \cite{faist2020continuous, hayden2021error, kong2021charge}). Here, the goal is to understand the limitations and capabilities of quantum error-correcting codes that can be realized by symmetric operations.   Then, again it is crucial to understand whether those codes can be  realized via local symmetric unitaries.\\

 \noindent \textbf{Symmetry-Protected Complexity:} Another  interesting open question in this area is to understand  how the notion of \emph{circuit complexity}  changes under the constraint of symmetry.   Recall that the circuit complexity of a unitary transformation or a state 
is the minimum number of local gates needed to implement the unitary or to 
prepare the state from a fixed (product) state \cite{aaronson2016complexity}. For a symmetric unitary or a symmetric state, we can define a modified notion of complexity, which can be called 
\emph{Symmetry-Protected Complexity} (SPC) and is defined as the minimum number \emph{symmetric} local unitaries  needed to implement a symmetric unitary  or to prepare a symmetric state. 
Certain aspects of this notion of complexity has been studied in the context of SPT phases \cite{chen2010local, chen2011classification}. In particular, it is known that for certain family of states the SPC grow linearly with the number of subsystems, whereas the regular complexity remains constant.    Given the conjectured roles of complexity   in the context of holography and AdS/CFT correspondence  \cite{susskind2016computational, brown2016holographic, stanford2014complexity}, it is interesting to further study the notion of SPC and compare it with the regular complexity. \\

\vspace{-2mm}
 
 \noindent \textbf{Analog  Quantum Simulation:}  Understanding the constraints imposed by the  locality of interactions   is also  crucial  in the context of analog quantum simulation, which is  one of the main applications of the near-term quantum technology. In this approach to quantum simulation,
 the degrees of freedom and the dynamics  of the target system are directly mapped to those of the simulator, which is a well-controlled quantum system with a tunable Hamiltonian (See e.g.,\cite{banuls2020simulating, altman2021quantum, yang2020observation}). 
  As we saw in this work,  in the presence of symmetries, the locality  of the simulator Hamiltonian  severely restrict the set of realizable Hamiltonians.   It is interesting to further explore how these  restrictions limit the power of analog quantum simulators in the presence of symmetries, and, in particular, to investigate if  they can be efficiently  circumvented.  \\


\section{Methods}\label{Sec:Methods}

  \subsection{Preliminaries: The Lie algebra generated by local symmetric Hamiltonians}\label{Sec:Lie}

We start by a quick review of a standard result in quantum control theory (See Supplementary Note 1 for more details).   Suppose one can implement the unitary time evolutions generated by Hamiltonians  $\pm A$ and $\pm B$ for an arbitrary amount of time $t\ge 0$; that is one can turn on and off these Hamiltonians at will. Then, combining these time evolutions one  
can obtain  unitaries
\bes\label{Lie}
\begin{align}
&e^{-i B (c_2 \delta t) }e^{-i A (c_1 \delta t) }=e^{-i (c_1 A+ c_2 B) \delta t } +\mathcal{O}(\delta t^2)\\
&e^{-i A \delta t}e^{-i B \delta t} e^{i A \delta t }e^{i B \delta t}= e^{-[A, B] \delta t^2}+\mathcal{O}(\delta t^3)\ ,
\end{align}
\ees
for arbitrary coefficients $c_1,c_2\in\mathbb{R}$, and for sufficiently small $\delta t$. This means that using  Hamiltonians  $\pm A$ and $\pm B$, one can approximately simulate the time evolutions generated by  any Hamiltonian in the linear span of $A$ and $B$ as well as the Hamiltonian $i[A,B]$. Furthermore, by  repeating such combinations of unitaries, one can obtain  a larger class of unitaries.   In fact, it can be proven   that using finite sequences of unitaries generated by Hamiltonians $\pm A$ and $\pm B$, one obtains  all unitary transformations $\{e^{-i H t}: t\in\mathbb{R}\}$ generated by any Hermitian operator $H$ if, and only if, $H$  belongs to the real Lie algebra generated by $A$ and $B$, i.e., it can be written as a linear combination of $A$, $B$, and their (nested) commutators, $i[A, B]$, $[[A, B], A]$, $[[A, B], B],...$, with real coefficients.  As we explain more in Supplementary Note 1, this result means that to characterize the group $\mathcal{V}^{G}_k$ generated by $k$-local symmetric unitaries, it suffices to characterize the Lie algebra generated by $k$-local symmetric skew-Hermitian operators. In particular, the dimension of this Lie algebra, as a vector space over $\mathbb{R}$, is equal to $\text{dim}(\mathcal{V}^{G}_k)$, the dimension of the manifold associated to $\mathcal{V}^{G}_k$, which is also equal to the number of real parameters needed to specify a general element of  $\mathcal{V}^{G}_k$. Using this relation, we establish an upper bound on $\text{dim}(\mathcal{V}^{G}_k)$, which is discussed next.

\color{black}
  
  \subsection{Charge vectors}
  
Next, we introduce the idea  of \emph{charge vectors}, which is our main tool for deriving constraints on the unitary evolutions generated by local symmetric Hamiltonians.  Recall that  $\text{Irreps}_G(n)$ denotes the set of inequivalent irreps of group $G$
 that appear in the representation $\{U(g)=u(g)^{\otimes n}: g\in G\}$ and 
 $|\text{Irreps}_G(n)|$ is the number of  these inequivalent irreps. Let $\Pi_\mu$ be the projector to the subspace corresponding  to irrep  $\mu\in \text{Irreps}_G(n)$, also known as the \emph{isotypic} component of $\mu$.
     For any operator $A$ define the \emph{charge vector} of $A$ as
\be\label{map-center}
|\chi_A\rangle\equiv\sum_{\mu\in\text{Irreps}_G(n)} \Tr(\Pi_\mu A)\ |\mu\rangle\ ,
\ee
 where  $\{|\mu\rangle:\mu\in\text{Irreps}_G(n)\}$ is a set of orthonormal vectors in an abstract vector space with dimension  $|\text{Irreps}_G(n)|$.   
 A general $G$-invariant Hamiltonian can have any charge vector with real coefficients. In particular,  for any set of real numbers $\{h_\mu\in\mathbb{R} :\mu\in\text{Irreps}_G(n)\}$, the Hermitian operator $\sum_{\mu\in\text{Irreps}_G(n)}  \frac{h_\mu}{\Tr(\Pi_\mu)} \Pi_\mu$ is  $G$-invariant and has the charge vector $\sum_{\mu\in\text{Irreps}_G(n)}  h_\mu \ |\mu\rangle$.   In other words, under the linear map $A\rightarrow |\chi_A\rangle$,   the image of the linear space of Hermitian $G$-invariant operators has dimension   $|\text{Irreps}_G(n)|$.

On the other hand, it turns out that if the the unitary evolutions  generated by Hamiltonian $H$ can be simulated by $k$-local $G$-invariant unitaries, i.e., if $\forall t\in \mathbb{R}: e^{-i H t}\in \mathcal{V}_k^G$, then  the charge of vector of $H$ should satisfy certain constraints. Let  $\mathcal{S}_k$ be the set of charge vectors for all such Hamiltonians, i.e.,
\begin{align}
\mathcal{S}_k\equiv \{|\chi_H\rangle: e^{-i H t }\in \mathcal{V}_k^G\ ,\ \forall t\in\mathbb{R}\}\ .
\end{align}
We prove that $\mathcal{S}_k$ is a linear subspace (over the field $\mathbb{R}$) with dimension
\be\label{dims}
\text{dim}(\mathcal{S}_k)\le |\text{Irreps}_G(k)|\ ,
\ee 
and the equality holds if $G$ is a connected Lie group, such as U(1) and SU(2). Therefore, if $|\text{Irreps}_G(k)|< |\text{Irreps}_G(n)|$, then $\text{dim}(\mathcal{S}_k)< |\text{Irreps}_G(n)|$, which means there are $G$-invariant Hamiltonians whose charge vectors do not belong to $\mathcal{S}_k$, which in turn implies they  cannot be simulated using $k$-local symmetric unitaries. For continuous groups,  such as U(1),  $|\text{Irreps}_G(n)|$ grows unboundedly with $n$ and, therefore, universality cannot be achieved with $k$-local symmetric unitaries with a fixed $k$. 

Below we present a simple argument  that explains  why the dimension of $\mathcal{S}_k$ cannot grow 
unboundedly with the system size. 
The specific  bound on $\text{dim}(\mathcal{S}_k)$ in Eq.(\ref{dims}) is proven in Supplementary Note 2,  using the Fourier transform of charge vectors. 
    Furthermore,
in Supplementary Note 2 we 
discuss more about  charge vectors and their Lie-algebraic interpretation. Briefly,   
 charge vector $|\chi_A\rangle$ of an operator $A$ determines its component in the center of  the Lie algebra of all $G$-invariant Hamiltonians, i.e., the Lie algebra corresponding  to the Lie group $ \mathcal{V}^G$. Then, the subspace $\mathcal{S}_k$ determines which part of the center can be generated by $k$-local $G$-invariant Hamiltonians. In particular, if $\text{dim}(\mathcal{S}_k)$ is less than $\text{dim}(\mathcal{S}_n)=|\text{Irreps}_G(n)|$, then local symmetric Hamiltonians cannot generate 
certain elements of the center, which means such Hamiltonians  are not universal, and results in the  bound in Eq.(\ref{Eq:dim3}).

Next, we explain why $\text{dim}(\mathcal{S}_k)$ can not grow unboundedly  with $n$.  To determine $\mathcal{S}_k$
we use the fact that if $e^{-i H t}\in\mathcal{V}^G_k$ for all $t\in\mathbb{R}$, then $ H$ should be in the Lie algebra generated by   $k$-local  
$G$-invariant operators, i.e.,  $H=\sum_j c_j A_j +\sum_{j_1,j_2}  c_{j_1,j_2}\ i [A_{j_1}, A_{j_2}]+\cdots $, where $A_j$  are Hermitian $k$-local $G$-invariant operators and coefficients $c_j, c_{j_1,j_2}, \cdots$ are real numbers.  It can be shown that the commutators appearing in this expansion, do not contribute in the charge vector of $H$, i.e., $|\chi_H\rangle=\sum_j c_j |\chi_j\rangle$, where $ |\chi_j\rangle\equiv \sum_{\mu\in\text{Irreps}_G(n)} \Tr(\Pi_\mu A_j)\ |\mu\rangle$ is the charge vector of $A_j$. To see this note that for any irrep $\mu\in \text{Irreps}_G(n)$,  $\Tr([A_{j_1}, A_{j_2}] \Pi_\mu)=\Tr(A_{j_1} [A_{j_2}, \Pi_\mu])=0$, where the first equality follows from the cyclic property of trace and  the second equality follows from  the assumption that $A_{j_2}$ is $G$-invariant, and therefore commutes with $\Pi_\mu$. It follows that  the commutator $ [A_{j_1}, A_{j_2}]$ and other nested commutators do not contribute in $|\chi_H\rangle$. This implies that  $\mathcal{S}_k$ is spanned by the charge vectors of $k$-local $G$-invariant Hermitian operators, i.e., $\mathcal{S}_k$ is equal to
\begin{align}\label{efef}
\text{Span}_\mathbb{R}\big\{|\chi_A\rangle: A=A^\dag, \text{$A$ is $k$-local}, [A, U(g)]=0: \forall g\in G\big\}\ .
\end{align}
Next, note that for any $k$-local operator $A$,  by applying a properly chosen permutation operator $S$ which changes the order of  sites, we can obtain an operator in the form $SAS^\dag=\tilde{A}\otimes I_\text{rest}$ with the property that $\tilde{A}$ acts  on a fixed set of $k$ sites (e.g., the first $k$ sites according to a certain ordering) and $I_\text{rest}$ is the identity operator on the remaining $n-k$  sites. Since charge vectors remain invariant under permutations, operators $A$ and $SAS^\dag=\tilde{A}\otimes I_\text{rest}$ have the same charge vectors. It follows that the subspace in Eq.(\ref{efef}) is equal to the set of the charge vectors of $G$-invariant Hermitian operators that act non-trivially only on a fixed set of $k$ sites (e.g., the first $k$ sites).  Therefore, as the number of total sites $n$ increases, $\text{dim}(\mathcal{S}_k)$ remains bounded by a number independent of $n$.  
In other words, even though using $k$-local $G$-invariant unitaries we can  simulate  Hamiltonians  that  are  not $k$-local,  they  can only  have  charge  vectors  which  are  allowed  for $k$-local $G$-invariant Hamiltonians. This explains why the upper bound on $\text{dim}(\mathcal{S}_k)$  in Eq.(\ref{dims}) does not depend on the system size.  \\

\noindent\textbf{Example-SU(2) symmetry with spin $s$ systems:} 
In the case of SU(2) symmetry, consider $n$ spin $s$ systems, each with the Hilbert space of dimension $2s+1$.
Recall that irreps of SU(2) can be labeled by the eigenvalues of the squared angular momentum operator  $J^2=J_x^2+J_y^2+J_z^2$. The eigenvalues  have the form of $j(j+1)$, where $j$ is half-integer 
which takes values $j=1/2,3/2, \cdots, ns$ if $s$ is not integer and $n$ is odd, and values  $j=0,1,\cdots, n$, otherwise.  In both cases the total number of distinct irreps is $|\text{Irreps}_{\text{SU}(2)}(n)|=\lfloor ns \rfloor +1$. Because SU(2) is a connected group, the bound in Eq.(\ref{dims}) holds as equality, i.e., $\text{dim}(\mathcal{S}_k)=\lfloor ks \rfloor +1$.   Furthermore, Eq.(\ref{Eq:dim3}) implies that the difference between the dimensions of the manifolds of all  SU(2)-invariant unitaries and those realizable by $k$-local SU(2)-invariant unitaries is lower bounded by  
\be
\text{dim}(\mathcal{V}_n^{_{\text{SU}(2)}})-\text{dim}(\mathcal{V}^{_{\text{SU}(2)}}_k)\ge \lfloor n s \rfloor- \lfloor k s \rfloor\ .
\ee
For integer spin $s$, this means 
that for any $k<n$, there are $(k+1)$-local unitaries that cannot be realized using $k$-local unitaries. Similarly, for non-integer  $s$,  there are $(k+2)$-local unitaries that cannot be realized using $k$-local unitaries.\\

\subsection{Full characterization of realizable U(1)-invariant Hamiltonians for qubits}\label{Sec:C:Methods}

In Supplementary Note 3, we study the example of U(1) symmetry for qubit systems.  Interestingly, it turns out that  in this example  the constraints imposed by the charge vectors fully characterize the set of realizable Hamiltonians.  The theorem below states these conditions.  

For a system with $n$ qubits define Hermitian operators $C_l: l=0,\cdots, n$ as 
\be\label{defCl}
C_l\equiv
\sum_{\substack{{\bf{b}}:  w({\bf{b}})=l}\ }  {\bf{Z}}^{\bf{b}}=\sum^n_{m=0} c_l(m) \  \Pi_m\ ,
\ee
where the first  summation is   over all bit strings $\textbf{b}=b_1\cdots b_n\in\{0,1\}^n$ with Hamming weight $w(\textbf{b})\equiv \sum_{j=1}^n b_j$ equal to $l$,   
and ${\bf{Z}}^{\bf{b}}=Z_1^{b_1}\cdots Z_n^{b_n}$.  In the second term,  $\Pi_m$ is the 
projector to the eigen-subspace of $N=\sum_{j=1}^n (I-Z_j)/2$ with 
eigenvalue $m$, and 
\be\label{bvb3}
c_l(m)=\sum_{s=0}^m (-1)^s  {{m}\choose{s}} {{n-m}\choose{l-s}}\ ,
\ee
 is the eigenvalue of $C_l$ in this subspace (recall that the binomial coefficient ${{a}\choose{b}}=0$ for $b>a$. See Supplementary Note 3 for derivation of Eq.(\ref{bvb3})). We prove  \\

\noindent\textbf{Theorem:} For any  U(1)-invariant Hamiltonian $H$ on $n$ qubits the family of unitaries $\{e^{-i t H}: t\in\mathbb{R}\}$ can be implemented using $k$-local U(1)-invariant unitaries for $k\ge 2$, if, and only if
\be\label{App:conds2}
\Tr(H C_l)=0\ \ \ \  :\ l=k+1,\cdots, n\ .\\
\ee 

Note that using Eq.(\ref{defCl}) these conditions can be rewritten in terms of the charge vector $|\chi_H\rangle=\sum_{m=0}^n \Tr(H \Pi_m)|m\rangle$ of Hamiltonian $H$, where $\{|m\rangle\}$ is a basis for an abstract $(n+1)$-dimensional vector space.

Eqs.(\ref{App:conds2}) impose  exactly $n-k$ independent constrains on the set of realizable Hamiltonians.  Hence, the difference between the dimension of realizable U(1)-invariant Hamiltonians and all U(1)-invariant Hamiltonians is exactly $n-k$, which means that in this case the general bound in Eq.(\ref{Eq:dim3}) holds as equality. This theorem is proven in the Supplementary Note 3. \\

\noindent\textbf{Proofs:}  All the results in the paper are rigorously proven in the Supplementary Notes 1-7.


\section*{Acknowledgments}
 
  I am grateful to Austin Hulse, David Jennings,  Hanqing Liu,   Hadi Salmasian, and Nicole Yunger-Halpern  for reading the manuscript carefully and providing many useful comments.  This work was supported by   NSF FET-1910571, NSF Phy-2046195  and  Army Research Office (W911NF-21-1-0005).  \\




%

\renewcommand\bibsection{\textbf{\large{References}}: \newline}

\newpage

\onecolumngrid

\onecolumngrid

\newpage

\maketitle
\vspace{-5in}
\begin{center}

\Large{Supplementary Material:\\ $ $ \\   Restrictions on realizable unitary operations imposed by symmetry and locality}
\end{center}
\appendix

	
	\title{\textbf{Supplementary Material}:\\ $ $ \\   Restrictions on realizable unitary operations imposed by symmetry and locality}
	
	\author{Iman Marvian}
\affiliation{Departments of Physics \& Electrical and Computer Engineering, Duke University, Durham, North Carolina 27708, USA}



\maketitle

\tableofcontents

\newpage

\section{Supplementary Note 1: Characterizing the group generated by local symmetric quantum circuits}\label{App:Char}
In this section we use a Lie algebraic approach to characterize the family of unitaries that can be implemented using local  symmetric quantum circuits. We start with a general setting, where the systems are not necessarily identical.

Consider $n$ systems, labeled as $j=1,\cdots, n$. Assume the dimension of the Hilbert space of system $j$ is $d_j<\infty$. Therefore, the total Hilbert space of the composite system is
\be\label{Hilbert}
\bigotimes_{j=1}^n \mathbb{C}^{d_j}\ .
\ee 
Suppose for each system we are given a unitary representation of a symmetry group $G$. In particular, let $u^{(j)}(g)$ be the action of group element $g\in G$ on the Hilbert space of system $j$. Then, the action of this group element  on the joint system  is described by the unitary
\be
U(g)=\bigotimes_{j=1}^n u^{(j)}(g)\ : \ \ g\in G\ .
\ee

\subsection{Lie Algebraic formulation of the problem}
 
Suppose one can implement unitary transformations  $\{e^{- i t K_j}: t\in \mathbb{R}, j=1,\cdots \}$, generated by a set of Hamiltonians $\{\pm K_j\}_j$. Composing these  unitaries one can implement  the group of unitaries  
\be
\Big\{e^{-i t_m K_{j_m}} \cdots e^{-i t_2 K_{j_2}} e^{-i t_1 K_{j_1}}\ :\ t_1,\cdots,t_m \in\mathbb{R}, m\in \mathbb{N}\Big\} \ .
\ee
Using the standard results in the theory of Lie groups and control theory \cite{d2007introduction, jurdjevic1972control}, it turns out that this group is a connected Lie group and is fully characterized by its corresponding Lie algebra. In particular, for any Hermitian operator $H$, the one-parameter family of unitaries $\{e^{-i H t}: t\in \mathbb{R}\}$ is in this group if, and only if, $i H$ is in the real Lie algebra generated by skew-Hermitian operators $\{i K_j\}_j$, denoted by $\mathfrak{alg}\{i K_j\}_j$, such that
\begin{align}\label{fwq}
i H=\sum_j &\alpha_j\ i K_j+\sum_{j_1,j_2} \beta_{j_1,j_2}\  [i K_{j_1}, i K_{j_2}] +\sum_{j_1,j_2, j_3} \gamma_{j_1,j_2,j_3}\ \big[[i K_{j_1}, i K_{j_2}], i K_{j_3}]\big]+ \cdots\ ,\nonumber
\end{align}
for some real coefficients    $\alpha_j,  \beta_{j_1,j_2}, \gamma_{j_1,j_2,j_3} \cdots$.  

We apply the above fact to study 
\be
\mathcal{V}_k^G\equiv \big\langle V: VV^\dag=I, V \text{ is $k$-local},   [V, U(g)]=0, \forall g\in G \big\rangle\ ,
\ee
i.e., the group  generated by $k$-local symmetric unitaries, where $I$ is the identity operator on the total Hilbert space in Eq.(\ref{Hilbert}) and $\langle\cdot\rangle$ denotes the generated group. Note that for any $k$-local symmetric Hamiltonian $h$, the family of unitaries generated by $h$  i.e., $\{e^{-i h t}: t\in\mathbb{R}\}$ are all $k$-local and symmetric. Conversely,  any $k$-local symmetric unitary $V$ can be obtained by applying a $k$-local symmetric Hamiltonian on the system for a finite time. Hence, we conclude that
\begin{proposition}\label{lkejf}
Let 
\be
\mathfrak{h}_{k}\equiv \mathfrak{alg}_\mathbb{R}\big\{A: \text{$k$-local}, A+A^\dag=0\ ,\ [A, U(g)]=0:\forall g\in G\big\}\ ,
\ee
be the real Lie algebra generated by the $k$-local, skew-Hermitian, $G$-invariant operators. For any Hermitian operator $H$, the family of unitaries $\{e^{-i H t}: t\in\mathbb{R}\}$ can be implemented  using $k$-local symmetric unitaries, if and only if, $iH\in \mathfrak{h}_{k}$, i.e.
\be
\forall t: e^{-i H t} \in \mathcal{V}_k^G  \ \ \Longleftrightarrow \ \  iH\in \mathfrak{h}_{k}\ .
\ee
\end{proposition}
Therefore, in the following, to characterize $\mathcal{V}_k^G$, we study the Lie algebra $\mathfrak{h}_{k}$.  Note that $\mathfrak{h}_{k}$ corresponds to the tangent space (at the identity) of the manifold  associated  to $\mathcal{V}^G_k$. In particular, the dimension of  $\mathcal{V}^G_k$ as a  manifold, is equal to the dimension of $\mathfrak{h}_{k}$ as a vector space (over the field  $\mathbb{R}$), i.e.
\be
\text{dim}(\mathcal{V}^G_k)=\text{dim}(\mathfrak{h}_{k})\ .
\ee
It is also worth noting that for $l>k$, $\mathcal{V}_k^G$ is a subgroup of $\mathcal{V}_l^G$ and $\mathfrak{h}_{k}$ is a sub-algebra of $\mathfrak{h}_{l}$. If $\text{dim}(\mathfrak{h}_{l})>\text{dim}(\mathfrak{h}_{k})$,  
then there are unitaries that can be implemented with $l$-local symmetric unitaries but not with $k$-local symmetric unitaries.

\subsection{Compactness of $\mathcal{V}^G_k$ and the impossibility of approximate implementation of symmetric unitaries}



As we show next, the group $\mathcal{V}^G_k$ generated by  $k$-local $G$-invariant unitaries is compact.  To prove this, we use the following fact, which can be proven using the techniques of the theory of algebraic groups (See e.g., chapter 5 of \cite{onishchik2012lie} and \cite{brylinski2002universal}).\\

\noindent \textbf{Fact 1:} For any compact Lie group, the subgroup generated by a finite set of compact connected subgroups is itself a compact connected Lie group. \\

Recall that in this paper, we only consider systems with finite-dimensional Hilbert spaces, and therefore the unitary group defined on the Hilbert space of the systems is a compact connected Lie group. Then, using this fact we can show that 

\begin{proposition}
For any symmetry group $G$, the group of unitaries $\mathcal{V}_k^G$ generated by $k$-local $G$-invariant unitaries is a compact connected Lie group.
\end{proposition}
\begin{proof}
First, it can be easily seen that the group of $G$-invariant unitaries $\mathcal{V}^G$ is a connected Lie group (i.e., there is a smooth path between the identity and any other group element $V$). Furthermore, defined as the commutant of a set of operators, the subgroup of symmetric unitaries $\mathcal{V}^G$  is closed. Therefore, we conclude that it is a compact connected subgroup of the unitary group. In fact, as we discuss later, $\mathcal{V}^G$   has a simple characterization, as the direct product of groups isomorphic to the unitary group U($m_\mu$), for different integers $m_\mu$.

Using similar arguments, we can also see that for any finite subset of sites, the group of $G$-invariant unitaries  which act non-trivially only on those sites is also a connected compact Lie group. Finally, recall that $\mathcal{V}_k^G$ is generated by  $k$-local $G$-invariant unitaries, i.e., $G$-invariant unitaries that act non-trivially on arbitrary subset of $k$ sites. Assuming the system has a finite number of sites $n$, this means that $\mathcal{V}_k^G$ is generated by a finite set of connected compact  Lie groups. Using Fact 1 about compact Lie groups, we conclude that  $\mathcal{V}_k^G$  itself is a compact connected Lie group. 
\end{proof}


The compactness of the group $\mathcal{V}_k^G$  has several useful implications. For instance, as we mentioned before,  compactness implies that $\mathcal{V}_k^G$  is uniformly finitely generated by $k$-local symmetric unitaries  \cite{d2007introduction}. Another useful implication of compactness follows from the following fact about Lie groups:  \\

\noindent \textbf{Fact 2:} For compact connected Lie groups, the exponential map from the Lie algebra to the Lie group is surjective.\\

It follows that
\begin{corollary}\label{cor31}
A unitary $V$ can be implemented using $k$-local symmetric unitaries, i.e., $V\in \mathcal{V}^G_k$ if, and only if, there exists $C\in \mathfrak{h}_{k}$ such that $V=e^{C}$. In other words, $\mathcal{V}^G_k=e^{\mathfrak{h}_{k}}$. 
\end{corollary}
 Therefore, by characterizing the Lie algebra $\mathfrak{h}_{k}$ we also find a full and direct characterization of $\mathcal{V}^G_k$. Hence, in the following, we focus on the study of the Lie algebra $\mathfrak{h}_{k}$.

 \subsection{Unitary evolution generated by local symmetric Hamiltonians}
Although we defined the group  $\mathcal{V}_k^G$ in terms of local symmetric quantum circuits, it can also be equivalently defined in terms of the unitary evolutions generated by local symmetric Hamiltonians.  Clearly, any quantum circuit can be interpreted as the unitary time evolution generated under a time-dependent Hamiltonian, whose symmetries and locality 
are determined by the symmetries and locality of the  quantum circuit.  In particular, any unitary $W\in \mathcal{V}_k^G$  can be generated by a Hamiltonian $H(t)=\sum_j h_j(t)$,  where each term $h_j(t)$ is $k$-local and $G$-invariant. In particular, we can choose k-local, $G$-invariant terms $h_j(t)$ such that  the family of unitaries  $V(t)$, generated by Hamiltonian $H(t)=\sum_j h_j(t)$ under the Schr\"{o}dinger equation,
\be\label{Sch}
\frac{d V(t)}{dt}=-i H(t) V(t)=-i\  \big[\sum_j h_j(t)\big]\  V(t)  ,
\ee
satisfy $V(t=0)=I$ and $V(t=1)=W$, where $I$ is the identity operator on the total Hilbert space.

 In the following, we argue that the converse is also true, i.e., the time evolution generated by any local symmetric Hamiltonian can also be realized by a finite local symmetric quantum circuit.

\begin{proposition} 
Suppose for all time $t\ge 0$, the Hermitian operator $H(t)$ is $G$-invariant and can be written as the sum of $k$-local terms. Then, for all time $t\ge 0$, the unitary evolution $V(t)$ generated by Hamiltonian $H(t)$ according to the Schr\"{o}dinger equation belongs to the Lie group $\mathcal{V}^G_k$, i.e., can be implemented by a quantum circuit with a finite number of $k$-local $G$-invariant gates.  Conversely, any unitary in $\mathcal{V}^G_k$ can be generated using a $G$-invariant Hamiltonian $H(t)$ that can be written as the sum of $k$-local terms. 
 \end{proposition}
\begin{proof}
As we explained above the second part of this proposition  is trivial. 
To prove the first part, suppose $H(t)=\sum_j h_j(t)$ is $G$-invariant. 
This does not imply that the $k$-local terms $\{h_j(t)\}$ are also $G$-invariant. However, we can easily show that $H(t)$ can be written as the sum of $k$-local $G$-invariant terms, i.e.  $H(t)=\sum_j \tilde{h}_j(t)$, where each $\tilde{h}_j(t)$ is both $k$-local and $G$-invariant. This can be shown, for instance, by twirling over group $G$, using the uniform (Haar) distribution over group $G$. E.g. for a compact Lie group $G$, suppose we define 
\be
\tilde{h}_j(t)\equiv\int dg\  U(g) {h}_j(t) {U(g)^\dag}\ ,   
\ee
where $dg$ is the uniform measure over group $G$. It can be easily seen that $\tilde{h}_j(t)$ becomes $G$-invariant. Furthermore, because $U(g) =\bigotimes_{s=1}^n u^{(s)}(g)$, the operator  $U(g) {h}_j(t) U^\dag(g)$ acts trivially on all systems except on the $k$ systems, where ${h}_j(t)$ acts non-trivially. It follows that $\tilde{h}_j(t)$ is also $k$-local. Finally, the assumption that $H(t)=\sum_j h_j(t)$ is $G$-invariant implies
\be
H(t)=\int dg\  U(g) H(t) {U(g)^\dag}=\int dg\  U(g) \big[\sum_j h_j(t)\big] {U(g)^\dag}=\sum_j \tilde{h}_j(t)\ .
\ee 
Since all operators $\{\tilde{h}_j(t): t\ge 0\}_j$ are $k$-local and $G$-invariant, the Lie algebra generated by operators  $\{i\tilde{h}_j(t): t\ge 0\}_j$ is a sub-algebra of $\mathfrak{h}_k$, the Lie algebra associated to the Lie group $\mathcal{V}_k^G$. 

Finally, we use a standard result of quantum control theory \cite{d2007introduction, jurdjevic1972control}, which implies that the 1-parameter family of unitaries $\{V(t): t\ge 0\}$ satisfying the Schr\"{o}dinger equation in Eq.(\ref{Sch}) with the initial condition $V(t=0)=I$ are in the Lie group associated to the Lie algebra generated by the set of operators $\{i h_j(t): t\ge 0\}$. Together with the above result  this implies that the family of unitaries $\{V(t): t\ge 0\}$ belongs to  $\mathcal{V}_k^G$ for all $t\ge 0$.

\end{proof}

\newpage

 \section{Supplementary Note 2: Charge vectors and constraints on realizable Hamiltonians}
In this section we further study and develop the idea of charge vectors and  explain its  Lie-algebraic interpretation, as the projection to the center of the Lie algebra of $G$-invariant Hamiltonians.  Using this technique we prove the bound 
\be\label{Eq:dim3}
\text{dim}(\mathcal{V}^{G})-\text{dim}(\mathcal{V}^{G}_k)\ge |\text{Irreps}_G(n)|-|\text{Irreps}_G(k)|\ .
\ee
Furthermore, in the special case of connected compact Lie groups, such as U(1) and SU(2),  we obtain a more fine-grained version of this bound. Namely, we show that  for any integer $l$ in the interval $k\le l\le n$, it holds that 
\be\label{Eq:dim}
\text{dim}(\mathcal{V}_l^{G})-\text{dim}(\mathcal{V}^{G}_k)\ge |\text{Irreps}_G(l)|-|\text{Irreps}_G(k)|\ .
\ee
This means that, if $|\text{Irreps}_G(l)|>|\text{Irreps}_G(k)|$, then there are symmetric unitaries that can be implemented with $l$-local symmetric unitaries, but not with $k$-local symmetric unitaries. In other words, 
as $k$ increases from 1 to $n$, i.e., as the local unitaries become more \emph{non-local}, the subgroup $\mathcal{V}_k^G$ generated by $k$-local symmetric unitaries gradually becomes larger.  

The following theorem contains a useful summary of some other results in this section.
\begin{theorem}\label{Thm10}
Consider a system with $n$ identical sites, and let $\{U(g)=u(g)^{\otimes n}: g\in G\}$ be the unitary representation of an  arbitrary finite or compact Lie group  $G$ on this system. Suppose the family of unitaries $\{e^{-i H t}: t\in \mathbb{R}\}$ generated by Hamiltonian $H$, belongs to $\mathcal{V}_k^G$, i.e., can be implemented by $G$-invariant $k$-local unitaries. Then, there exists a set of real numbers $c_\nu\in\mathbb{R}$, such that for all group elements $g\in G$, it holds that
\be\label{kjh}
\Tr\big(H\ u(g)^{\otimes n}\big) =[\Tr(u(g))]^{n-k}\times \sum_{\nu\in \text{Irreps}_G(k)} c_\nu\  f_\nu(g) \ ,
\ee
where the summation is over all inequivalent irreducible representations of $G$ appearing in the representation $\{u(g)^{\otimes k}: g\in G\}$, and $f_\nu: G\rightarrow\mathbb{C}$ is the character of the irrep $\nu$. Conversely, for any set of real numbers $c_\nu\in\mathbb{R}$, there exists a Hermitian operator $H$ satisfying the above equality, such that all unitaries $\{e^{-i H t}: t\in \mathbb{R}\}$ generated by $H$ belong to $\mathcal{V}_k^G$.  
 Furthermore, for any unitary $V\in \mathcal{V}_k^G$, there exists a $G$-invariant Hermitian operator $H$, such that $V=e^{-i H}$, and $H$ satisfies 
the above condition for a set of  real numbers $c_\nu\in\mathbb{R}$.
\end{theorem}
As we explain later, this theorem follows from lemma \ref{Thm2} (See the discussion after this lemma).

\subsection{Charge vectors of general symmetric Hamiltonians }

 Consider the decomposition of the  unitary representation $\{U(g): g\in G\}$ into the irreps of group $G$.  Under this decomposition the Hilbert space  decomposes as 
 \be
 \mathcal{H}\cong \bigoplus_{\mu\in \text{Irreps}(n)} \mathcal{H}_\mu\ ,
 \ee
  where the summation is over $\text{Irreps}(n)$, the set of  inequivalent irreducible representations (irreps) of $G$ appearing in this representation and   $\mathcal{H}_\mu$ is the subspace corresponding to irrep $\mu$, also known as the \emph{isotypic} component of $\mu$.  Using  Schur's lemmas, we can show that any  $G$-invariant operator is block-diagonal with respect to this decomposition and, in general, can have support in any arbitrary subset of these sectors (It is worth noting that in the case of non-Abelian groups Schur's lemmas imply stronger constraints on the form of  $G$-invariant operators. See Eq.(\ref{Schur})).  \color{black}  However, as we will see in the following, for Hamiltonians generated by a fixed set of $G$-invariant Hamiltonians $\{H_j\}_j$, the supports in different subspaces $\{\mathcal{H}_\mu\}$ satisfy  particular constraints, dictated by  Hamiltonians $\{H_j\}_j$.

For any operator $A$, consider the vector 
\be\label{epj}
|\chi_A\rangle\equiv \sum_{\mu\in \text{Irreps}(n)} \Tr(A \Pi_\mu)\ |\mu\rangle \ ,
\ee
where $\Pi_\mu$ is the projector to the subspace $\mathcal{H}_\mu$ and $\{|\mu\rangle:  \mu\in \text{Irreps}(n)\}$ is a set of orthonormal vectors in an abstract vector space (not the state space of the system).  Vector $|\chi_A\rangle$, which  will be called the \emph{charge vector} of operator $A$,  encodes information about the components of  this operator in the subspace spanned by $\{\Pi_\mu: \mu\in \text{Irreps}(n)\}$. As we explain in remark \ref{remark61}, it can be shown  that $\{i\Pi_\mu: \mu\in \text{Irreps}(n)\}$ is the center of  $\mathfrak{h}_n$, the Lie algebra of  skew-Hermitian $G$-invariant operators (Recall that the center of a Lie algebra is the subalgebra formed from the  elements of the algebra that commute with all elements of the Lie algebra).  Therefore, the charge vector $|\chi_A\rangle$  determines the projection of $A$ into the center of  $\mathfrak{h}_n$.

A general $G$-invariant Hamiltonian can have any charge vector with real coefficients. In particular,  for any set of real numbers $\{a_\mu\in\mathbb{R}\}$, the Hermitian operator $\sum_{\mu\in \text{Irreps}(n)} a_\mu \Pi_\mu/\Tr(\Pi_\mu)$ is $G$-invariant and has the charge vector $\sum_{\mu\in \text{Irreps}(n)} a_\mu \ |\mu\rangle$.   In other words, under the linear map $A\rightarrow  |\chi_A\rangle$,   the image of the Lie algebra of skew-Hermitian $G$-invariant operators    is  $\{i\sum_\mu a_\mu |\mu\rangle: a_\mu\in\mathbb{R} \}$, which is a vector space  over field $\mathbb{R}$, with dimension equal to $|\text{Irreps}_G(n)|$.

\begin{lemma}\label{Thm1}
Consider a set of $G$-invariant Hermitian operators $\{H_j\}_j$.  For any operator $A$, if $iA\in\mathfrak{alg}\{i H_j\}_j$ then the charge vector of $A$ is in the subspace spanned by the charge vectors of $\{H_j\}$, i.e., $|\chi_A\rangle\in \text{Span}_\mathbb{R}\{|\chi_j\rangle\}_j$, where $|\chi_j\rangle=\sum_{\mu\in \text{Irreps}(n)} \Tr(H_j \Pi_\mu)\ |\mu\rangle$, $\text{Irreps}(n)$ is the set of inequivalent irreps of $G$ in the representation $\{U(g): g\in G\}$ and $\Pi_\mu$ is the projector to the subspace corresponding to irrep $\mu$.  Furthermore, 
\be\label{lem:eq}
\text{dim}(\mathfrak{h}_n)-\text{dim}(\mathfrak{alg}_\mathbb{R}\{iH_j\})\ge |\text{Irreps}(n)|-\text{dim}(\text{Span}_\mathbb{R}\{|\chi_j\rangle\}_j)\ ,
\ee
where $\mathfrak{h}_n$ is the Lie algebra of all skew-Hermitian $G$-invariant operators.
\end{lemma}
\begin{proof}
Recall that the elements of the Lie algebra  $\mathfrak{alg}_\mathbb{R}\{iH_j\}$ can be written as  linear combination of terms $\{iH_j\}_j$, and their nested commutators  $\{[i H_{j_1}, i H_{j_2}], [[i H_{j_1}, i H_{j_2}],i H_{j_3}],..\}$. This implies that if $iA \in \mathfrak{alg}_\mathbb{R}\{iH_j\}_j$, then there exists a set of real coefficients $\{a_j\in\mathbb{R}\}$ and a Hermitian operator $B$, such that
 \be
 A= \sum_j a_j H_j+ i [B, H_j]\ .
 \ee
Let $\Pi_\mu$ be the projector to the subspace corresponding to irrep $\mu$. Then, this equation implies
 \begin{align}
\Tr(A \Pi_\mu) &= \sum_j  a_j  \Tr(\Pi_\mu H_j)+ i\Tr(\Pi_\mu [H_j, B])\\ &=\sum_j  a_j  \Tr(\Pi_\mu H_j)+i\Tr([\Pi_\mu,H_j] B) \\ &=\sum_j  a_j  \Tr(\Pi_\mu H_j)\ ,
 \end{align}
where the second line follows from the cyclic property of trace and the last line follows from the assumption that for each $j$, operator $H_j$  is $G$-invariant, which implies $[\Pi_\mu, H_j]=0$. Therefore, we conclude that
\be
|\chi_A\rangle\equiv \sum_\mu \Tr(A \Pi_\mu)  |\mu\rangle=\sum_j  a_j \sum_\mu \Tr(\Pi_\mu H_j) |\mu\rangle=\sum_j  a_j |\chi_j\rangle\ ,
\ee
i.e., $|\chi_A\rangle\in \text{span}_\mathbb{R}\{|\chi_j\rangle\}$, which proves the first part of lemma. 

To prove the second part, we
use the rank-nullity theorem for the linear map  $ C\rightarrow i |\chi_C\rangle$.
 Rank-nullity theorem \cite{axler1997linear} states that  for any linear map, the dimension of the domain is equal to the sum of the  dimensions of its image (i.e., the rank of the map) and its kernel (i.e., the nullity of the map).   Since  $\mathfrak{alg}_\mathbb{R}\{i H_j\}$ is a subspace of $\mathfrak{h}_n$,  
using the rank-nullity theorem, we find that the difference between the dimensions of  $\mathfrak{h}_n$ and its  subspace $\mathfrak{alg}_\mathbb{R}\{iH_j\}$  is larger than, or equal to, the  difference between the dimensions of their images under the linear map  $ C\rightarrow i |\chi_C\rangle$. As we have seen before, the dimension of these images are respectively,  $|\text{Irreps}(n)|$ and    
$\text{dim}(\text{Span}_\mathbb{R}\{|\chi_j\rangle\}_j)$. This immediately implies   Eq.(\ref{lem:eq}).
 \end{proof}

\begin{remark}
In Eq.(\ref{epj}) we defined charge vectors based on projectors $\{\Pi_\mu\}$. This basis spans a subspace with dimension $|\text{Irreps}(n)|$ of the operator space, which corresponds to the center of $\mathfrak{h}_n$ (See remark \ref{remark61}). In general, charge vectors can be defined in terms of any arbitrary basis for this space. 
\end{remark}


\subsection{Charge vector and its Fourier Transform} 

Here, we describe a sightly different way of formulating the constraints on the charge vectors found in  lemma \ref{Thm1}.  This formulation does not explicitly depend on the decomposition  of symmetry to irreducible representations.  \color{black} For any operator $A$, consider the function $\chi_A: G\rightarrow \mathbb{C}$ defined by equation 
 \be
 \chi_A(g)\equiv \Tr(A U(g))\ : \ g\in G  \ .
 \ee
 Then, using an argument  similar to the  argument used in the proof of lemma \ref{Thm1},  we can easily prove
 \begin{lemma}\label{lem2}
Assume $i A\in \mathfrak{alg}_\mathbb{R}\{iH_j\}$, where $\{H_j\}_j$ are $G$-invariant Hermitian operators. Then,  
$\chi_A\in \text{Span}_\mathbb{R}\{\chi_j\}_j$, where $\chi_j(g)=\Tr(H_j U(g)): \forall g\in G$\ . Furthermore,
\be
\text{dim}(\mathfrak{h}_n)-\text{dim}(\mathfrak{alg}_\mathbb{R}\{iH_j\})\ge |\text{Irreps}(n)|-\text{dim}(\text{Span}_\mathbb{R}\{\chi_j\}_j)\ .
\ee
\end{lemma}
For any $G$-invariant operator $A$, the function $\chi_A$ and the charge vector $|\chi_A\rangle$ are related via Fourier transform. To understand the connection better, consider the decomposition of the representation $\{U(g): g\in G\}$ to the irreducible representations  of $G$. If $G$ is a finite or compact Lie group, then every finite-dimensional representation is completely reducible, i.e., there exists a unitary  $W$ such that  
\be\label{rhp}
 W U(g) W^\dag= \bigoplus_{\mu\in \text{Irreps}(n)} u^{(\mu)}(g)\otimes I_{m_\mu}\ \ ,\ \ \ \ \forall g\in G
\ee
and the Hilbert space $\mathcal{H}$ decomposes as 
\be
\mathcal{H}\cong \bigoplus_{\mu\in \text{Irreps}(n)} \mathcal{H}_\mu\ \cong \bigoplus_{\mu\in \text{Irreps}(n)} \mathbb{C}^{d_\mu}\otimes \mathbb{C}^{m_\mu} \ ,
\ee
  where $\text{Irreps}(n)$ is the set of inequivalent  irreps of $G$ appearing in the representation $\{U(g): g\in G\}$,  $\{u^{(\mu)}(g): g\in G\}$ is the irreducible representation  which acts irreducibly on $\mathbb{C}^{d_\mu}$, $d_\mu$ is the dimension of  irrep $\mu$ and  ${m_\mu}$ is its multiplicity, and $I_{m_\mu}$ is the identity operator on $\mathbb{C}^{m_\mu} $.   Using Schur's lemmas, one can show that in this basis any $G$-invariant operator $A$ can be written as 
\be\label{Schur}
W A  W^\dag= \bigoplus_{\mu\in \text{Irreps}(n)} I_{d_\mu} \otimes A^{(\mu)}\ ,
\ee
 where $I_{d_\mu} $ is the identity operator on $\mathbb{C}^{d_\mu}$, and $A^{(\mu)}$ is an operator acting on $\mathbb{C}^{m_\mu} $ (See e.g. \cite{QRF_BRS_07}). Using this decomposition,  the charge vector of operator $A$ is 
 \be
 |\chi_A\rangle=\sum_{\mu\in \text{Irreps}(n)} \Tr(\Pi_\mu A)\ |\mu\rangle= \sum_{\mu\in \text{Irreps}(n)} d_\mu\times \Tr( A^{(\mu)})\ |\mu\rangle=\sum_{\mu\in \text{Irreps}(n)} d_\mu\times a_\mu \ |\mu\rangle\  ,
 \ee
where $a_\mu= \Tr(A^{(\mu)})$. On the other hand, 
\be\label{kohqef}
\chi_A(g)=\Tr(A U(g))=\Tr\big(\bigoplus_{\mu\in \text{Irreps}(n)} u^{(\mu)}(g) \otimes A^{(\mu)}\big)=\sum_{\mu\in \text{Irreps}(n)} \Tr(u^{(\mu)}(g))\times  \Tr(A^{(\mu)})=\sum_{\mu\in \text{Irreps}(n)} a_\mu\ f_\mu(g)\ ,
\ee
where $f_\mu(g)=\Tr(u^{(\mu)}(g))$ is the character of 
irrep $\mu$.  Recall the orthogonality relation  
for the characters \cite{georgi2018lie}, namely
\be
\int_G dg\  f_\nu(g) f^\ast_\mu(g)=\delta_{\mu,\nu}\ ,  
\ee
where the integral is over the Haar (uniform) measure for compact Lie group $G$ and $f^\ast_\mu$ is the complex conjugate of the character $f_\mu$. 
\color{black}
Using this, one can obtain $|\chi_A\rangle$ from  $\chi_A(g)$ and vice versa, using Fourier transforms. In particular, 
 \be\label{Fourier}
 |\chi_A\rangle=\int_G dg\  \chi_A(g) \ \sum_{\mu\in \text{Irreps}(n)} d_\mu\  f^\ast_\mu(g)   |\mu\rangle\ ,
 \ee 
where  $d_\mu$ is  the dimension of irrep $\mu$. 

It is worth  noting that Eq.(\ref{Schur}) implies that the group of symmetric unitaries $\mathcal{V}^G$ is isomorphic to the direct product of unitary groups U($m_\mu$), for all $\mu\in \text{Irreps}(n)$, where $m_\mu$ is the multiplicity of irrep $\mu$ in the representation $\{U(g): g\in G\}$.   \\

\begin{remark}\label{remark61}   (Center of the Lie algebra of $G$-invariant Hamiltonians) Decomposition in Eq.(\ref{Schur}) gives a simple  characterization of  $\mathfrak{h}_n$, the Lie algebra of $G$-invariant skew-Hermitian operators. In particular, for any set of skew-Hermitian operators $\{A^{(\mu)}: \mu \in \text{Irreps}(n)\}$ the operator $A$ is in $\mathfrak{h}_n$. Using Schur's lemmas, this immediately implies that the center of $\mathfrak{h}_n$ is spanned by $\{i \Pi_{\mu}: \mu\in \text{Irreps}(n)\}$. This means that for any skew-Hermitian $G$-invariant operator $A$, its charge vector $|\chi_A\rangle\equiv \sum_\mu \Tr(A \Pi_\mu)  |\mu\rangle$ determines the component of $A$ in the center of the Lie algebra $\mathfrak{h}_n$. 
\end{remark}

\subsection{Charge vectors of the Lie algebra generated by $k$-local symmetric Hamiltonians: The case of identical subsystems}

In this section, we focus on the special case where all subsystems $j=1,\cdots n$, have identical Hilbert spaces isomorphic to $\mathbb{C}^d$ for a finite $d$, and identical unitary representation of symmetry $G$. In particular, we assume the action of group element $g\in G$ on the composite system is represented by $U(g)=u(g)^{\otimes n}$. 

Suppose the family of unitaries $\{e^{-i H t}: t\in\mathbb{R}\}$ can be implemented using $k$-local $G$-invariant unitaries, i.e., suppose $iH\in\mathfrak{h}_k$. Then,  lemma  \ref{Thm1} implies that the charge vector of $H$ should be in the subspace spanned by the charge vector of $k$-local $G$-invariant Hamiltonians, denoted by
\be\label{defS}
\mathcal{S}_k\equiv  \{|\chi_H\rangle: e^{-i H t }\in \mathcal{V}_k^G , \forall t\in\mathbb{R}\}=\{|\chi_A\rangle: i A\in\mathfrak{h}_k\} =\{|\chi_A\rangle: A=A^\dag,  A\text{ is $k$-local},\ [A, U(g)]=0\ : \forall g\in G\}\ ,
\ee
where the last equality follows from lemma  \ref{Thm1}. As we discussed in remark \ref{remark61} charge vector of a $G$-invariant Hamiltonian $H$ determines the component of $i H$ in the center of $\mathfrak{h}_n$. This means $\mathcal{S}_k$ has  a simple interpretation:   
It determines the subspace of the center of $\mathfrak{h}_n$ which is included in $\mathfrak{h}_k$, i.e., is generated by $k$-local $G$-invariant Hamiltonians. In particular, $\text{dim}(\mathcal{S}_k)$ is the dimension of the subspace of the center of $\mathfrak{h}_n$ that is also included in  $\mathfrak{h}_k$.

Using Eq.(\ref{defS}), we can immediately see that the  dimension of $\mathcal{S}_k$ does not grow with $n$, the total number of sites, and only depends of $k$: First, note that the charge vectors are invariant under permutations of sites, i.e., for any permutation $S$, the charge vector of $A$ and $S A S^\dag$ are equal. This follows from the assumption that all sites have identical representation of symmetry and any such permutation leaves the total charge in the system invariant, i.e., $[S, \Pi_mu]=0$ for all $\mu\in\text{Irreps}(n)$.

Furthermore, by applying a proper   permutation $S$,  any $k$-local operator $A$ can be converted to an operator $S A S^\dag$, which acts non-trivially only  on a fixed set of $k$ sites. Therefore, the dimension of $\mathcal{S}_k$ is equal to the dimension of the subspace spanned by the charge vectors of $G$-invariant operators which act nontrivially only on a fixed $k$ sites. This immediately implies the dimension of   $\mathcal{S}_k$ does not grow with $n$. In fact, as we show next, it can  be easily seen that  $\text{dim}(\mathcal{S}_k)\le |\text{Irreps}_G(k)|$, 
where $|\text{Irreps}_G(k)|$ is the number of distinct irreps of $G$  appearing in representation $\{u(g)^{\otimes k}: g\in G\}$. This  follows from the following lemma, which also gives a simple criterion for testing whether  the charge vector of $A$ is 
in $\mathcal{S}_k$, or not.

\begin{lemma}\label{Thm2}
Suppose  $iB\in \mathfrak{h}_k\equiv \mathfrak{alg}\{A: \text{$k$-local}, A+A^\dag=0, [A, u(g)^{\otimes n}]=0:\forall g\in G\}$ . Then,  there exists a set of real coefficients  $\{b_\mu\in \mathbb{R}\}$, such that
 \be\label{wfefe}
 \forall g\in G: \ \ \ \ \ \ \  \Tr(u(g)^{\otimes n} B)=[\Tr(u(g))]^{n-k}\times \sum_{\mu\in\text{Irreps}_G(k)} b_\mu\  f_\mu(g)\ ,
 \ee
 where the summation is over $\text{Irreps}_G(k)$, the set of distinct irreps of $G$  appearing in representation $\{u(g)^{\otimes k}: g\in G\}$ and  $f_\mu$ is the character of irrep $\mu$.   Furthermore, for any set of real numbers  $\{b_\mu\}$, there exists a Hermitian $G$-invariant  operator $B$, satisfying this equation. 
\end{lemma}
Another way to phrase the condition in the lemma  is that 
\be\label{ytty}
|\chi_B\rangle\in \mathcal{S}_k \Longleftrightarrow\ \chi_B\in \text{Span}_\mathbb{R}\{r^{(n-k)} f_\nu: \nu\in \text{Irreps}_G(k) \}\ ,
\ee
where $r(g)=\Tr(u(g))$. The  subspace in the right-hand side has dimension less than or equal to  $|\text{Irreps}_G(k)|$.  Since $\chi_B$ and $|\chi_B\rangle$ are related via an invertible linear map, it follows that $\text{dim}(\mathcal{S}_k)\le |\text{Irreps}_G(k)|$.  Note that for a general $G$-invariant Hamiltonian $B$, $\chi_B$ can be any function in the subspace $\text{Span}_\mathbb{R}\{f_\nu: \nu\in \text{Irreps}_G(n) \}$, which has dimension  $|\text{Irreps}_G(n)|$ (This follows from the linear independence of the characters of different irreps).

\begin{proof}
For any $k$-local Hermitian operator  $B$, there exists a  
permutation operator $S$ such that $S B S^\dag=\tilde{B}\otimes I_d^{\otimes (n-k)}$, where  $\tilde{B}$ is a Hermitian operator  defined on the first $k$ systems and $I_d$ is the identity operator 
on the Hilbert space of each system. It follows that 
\be
\Tr(B\ u(g)^{\otimes n})=Tr([S B S^\dag] u(g)^{\otimes n})=[\Tr(u(g))]^{n-k} \times \Tr(\tilde{B}\  u(g)^{\otimes k})\ .
\ee
Furthermore, if $[B, u(g)^{\otimes n}]=0$, then  
 $[\tilde{B}, u(g)^{\otimes k}]=0$, for all $g\in G$. Then, using Eq.(\ref{kohqef}) we find
 \be\label{gd}
\Tr(\tilde{B}\ u(g)^{\otimes k})=\sum_{\mu\in \text{Irreps}(k)} b_\mu\ f_\mu(g)\ ,
 \ee 
where $\{b_\mu\}$ is a set of real coefficients. This proves the first part of the lemma. To see the second part, we note that if in  the left-hand side of Eq.(\ref{gd}) we choose $\tilde{B}$ to be the projector to the subspace corresponding to irrep $\mu\in\text{Irreps}(k)$, then $\Tr(\tilde{B}\ u(g)^{\otimes k})=c_\mu f_\mu(g), \forall g\in G$, for a constant  $c_\mu>0$. More generally, considering the linear combinations of projectors corresponding to all  $\text{Irreps}_G(k)$, we conclude that  for any set of real coefficients $\{b_\mu: \mu\in\text{Irreps}_G(k)\}$,  there is a Hermitian $G$-invariant  operator $\tilde{B}$ acting on $(\mathbb{C}^d)^{\otimes n}$, such that $\sum_{\mu\in\text{Irreps}_G(k)} b_\mu\  f_\mu(g)$. Then, operator $B=\tilde{B}\otimes I_d^{\otimes n-k}$, is a $k$-local $G$-invariant Hermitian operator and satisfies Eq.(\ref{gd}). This proves the the second part of the lemma.

\end{proof}

This lemma together with proposition \ref{lkejf} and corollary \ref{cor31} imply theorem \ref{Thm10}. 
 An interesting corollary of this result is 
\begin{corollary}
Suppose $G$ is an Abelian group. If there exists a group element $g\in G$ such that $\Tr(u(g))=0$, then for any  $k<n$,  $\text{dim}(\mathfrak{h}_k)<\text{dim}(\mathfrak{h}_n)$.  
\end{corollary}
\begin{proof}
We prove a slightly more general result: Suppose the representation $\{u(g)^{\otimes n}: g\in G\}$ contains a 1D irrep of group $G$, and let $\Pi$ be the projector to the subspace corresponding to this 1D irrep. In the case of Abelian groups all irreps 
are 1D and therefore this assumption is always satisfied.  Now we show that if there exists $g\in G$ such that $\Tr(u(g))=0$ and if $k<n$, then $i\Pi\notin \mathfrak{h}_k$, whereas clearly, $i\Pi\in \mathfrak{h}_n$. This implies $\text{dim}(\mathfrak{h}_k)<\text{dim}(\mathfrak{h}_n)$. To prove this claim, we assume it is not true, i.e.,  $i\Pi\in \mathfrak{h}_k$ and derive a contradiction. If $i\Pi\in \mathfrak{h}_k$, then  Eq.(\ref{wfefe})  should be satisfied for $B=\Pi$ and certain coefficients $b_\mu\in\mathbb{R}$. However, if there exists a group elements $g\in G$ such that $\Tr(u(g))=0$ and $k<n$ , then the right-hand side of this equation is zero for this group element. On the other hand, the left-hand side, i.e., $\Tr(\Pi u(g)^{\otimes n})$ is non-zero: because $\Pi$ is a projector to a 1D irrep, $|\Tr(\Pi u(g)^{\otimes n})|=\Tr(\Pi)>0$. Therefore, the assumption that $i\Pi\in \mathfrak{h}_k$ is in contradiction with the assumptions that $k<n$ and $\Tr(u(g))=0$, for some $g\in G$. This completes the proof. 
\end{proof}

Another useful corollary of lemma \ref{Thm2} is the following result. 

\begin{corollary}\label{corcor}
Recall the definition $\mathcal{S}_k\equiv  \{|\chi_H\rangle: e^{-i H t }\in \mathcal{V}_k^G , \forall t\in\mathbb{R}\}=\{|\chi_A\rangle: i A\in\mathfrak{h}_k\}$. For any group $G$, $\text{dim}(\mathcal{S}_k)\le\big|\text{Irreps}_G(k)\big|$. Furthermore, for any connected Lie group $G$, this holds as as equality, i.e., $\text{dim}(\mathcal{S}_k)=\big|\text{Irreps}_G(k)\big|$. 
\end{corollary}
\begin{proof}
Using Eq.(\ref{defS}), we have
\begin{align}
\text{dim}(\mathcal{S}_k)&=\text{dim}\big(\{|\chi_A\rangle: A=A^\dag,  A\text{ is $k$-local},\ [A, U(g)]=0\ : \forall g\in G\}\big)\\ &=\text{dim}\big(\{\chi_A: A=A^\dag,  A\text{ is $k$-local},\ [A, U(g)]=0\ : \forall g\in G\}\big)\ ,\label{wr}
\end{align}
where $\chi_A: G\rightarrow \mathbb{C}$ is defined by $\chi_A(g)=\Tr(A U(g))=\Tr(A u(g)^{\otimes n})$, and to get the second line we have used the fact that there is an invertible linear  map between the charge vector $|\chi_A\rangle$ and function $\chi_A$, namely 
the Fourier transform in Eq.(\ref{Fourier}).  Define
\be
v(g)\equiv \Tr(u(g)):\ \ g\in G\ .
\ee
Then, using Eq.(\ref{wfefe}) 
\be
\Big\{\chi_A: A=A^\dag,  A\text{ is $k$-local},\ [A, U(g)]=0\ : \forall g\in G\Big\}=\Big\{v^{n-k}\times \sum_{\mu\in\text{Irreps}_G(k)} a_\mu\  f_\mu\ : a_\mu\in \mathbb{R}\Big\}\ .
\ee
Together with Eq.(\ref{wr}) this immediately implies that
\begin{align}
\text{dim}(\mathcal{S}_k)&=\text{dim}\big(\{\chi_A: A=A^\dag,  A\text{ is $k$-local},\ [A, U(g)]=0\ : \forall g\in G\}\big)\\ &=\text{dim}\big(\Big\{v^{n-k}\times \sum_{\mu\in\text{Irreps}_G(k)} a_\mu\  f_\mu\ : a_\mu\in \mathbb{R}\Big\}\big)\ , \label{2ve} \\ &\le |\text{Irreps}_G(k)|\ .
\end{align}
This bound holds for any group $G$.
In the special case of connected Lie group, functions $\{v^{n-k} \times f_\mu: \mu\in \text{Irreps}_G(k)\}$ are linearly independent.  
To prove this we assume the contrary holds, i.e., there exists a set of real coefficients $a_\mu$, such that
\be\label{kjbf}
v^{n-k} \sum_{\mu\in \text{Irreps}_G(k)} a_\mu  \times f_\mu=0\ .
\ee
Assuming $G$ is a connected Lie group, then  there is a finite neighborhood around the group identity, such that for all group elements $g$ in the neighborhood, $v(g)=\Tr(u(g))\neq 0$ (Recall that the representation is finite-dimensional). Therefore, if Eq.(\ref{kjbf}) holds then for all group elements $g$ in this neighborhood  $\sum_{\mu\in \text{Irreps}_G(k)} a_\mu  f_\mu(g)=0$. But, because the characters 
 are analytic functions, if $\sum_{\mu\in \text{Irreps}_G(k)} a_\mu  f_\mu(g)=0$ is zero in a finite neighborhood, then it should vanish everywhere. Finally, using the fact that characters are linearly independent, we find that Eq.(\ref{kjbf}) holds if and only if all $a_\mu=0$ for all $\mu\in \text{Irreps}_G(k)$. We conclude that  the set of functions $\{v^{n-k} \times f_\mu: \mu\in \text{Irreps}_G(k)\}$ are linearly independent. Therefore,  if $G$ is a connected Lie group,  Eq.(\ref{2ve}) implies that  $\text{dim}(\mathcal{S}_k)=|\text{Irreps}_G(k)|$.
 \end{proof}

\subsection{Non-universality of LSQC and proof of the bound in Eq.(\ref{Eq:dim}): The case of identical subsystems }

Recall the following definitions 
\begin{align}\nonumber
\mathcal{V}^G_k&: \text{The Lie group generated by $k$-local $G$-invariant unitaries} \\
 \mathfrak{h}_k&:  \text{The Lie algebra corresponding to $\mathcal{V}^G_k$} \nonumber\\
   \mathcal{S}_k&: \text{The corresponding charge vectors}\nonumber
\end{align}
Note that 
\be
\text{dim}(\mathcal{V}^G_{k})=\text{dim}(\mathfrak{h}_k)\ ,
\ee
where the left-hand side is the dimension of $\mathcal{V}^G_{k}$ as a manifold and the right-hand side is the dimension of $\mathfrak{h}_k$ as a vector space. 
\begin{theorem}\label{Thm0040}
For a general group $G$, and integer $k\le n$
\begin{align}
\text{dim}(\mathfrak{h}_n)-\text{dim}(\mathfrak{h}_k)&\ge \text{dim}(\mathcal{S}_n)-\text{dim}(\mathcal{S}_k)\ge \text{Irreps}_G(n)-\text{Irreps}_G(k)\ .
\end{align}
Furthermore, if $G$ is connected Lie group, then for any  integers $k$ and $l$, satisfying $1\le k\le l\le n$, it holds that
\begin{align}
\text{dim}(\mathfrak{h}_l)-\text{dim}(\mathfrak{h}_k)&\ge \text{dim}(\mathcal{S}_l)-\text{dim}(\mathcal{S}_k)=\text{Irreps}_G(l)-\text{Irreps}_G(k)\ .
\end{align}
\end{theorem}

\begin{proof}
This theorem follows immediately by applying   the rank-nullity theorem for the  linear map $A\rightarrow |\chi_A\rangle$, together with corollary \ref{corcor}. In particular, note that for $k\le l$, the Lie algebra  $\mathfrak{h}_k$ is a subspace of $\mathfrak{h}_l$. Furthermore, $\mathcal{S}_k$ and $\mathcal{S}_l$ are their images under a linear map (up to the imaginary $i$).  Recall that according to the 
 rank-nullity theorem \cite{axler1997linear},  for any linear map, the dimension of the domain is equal to the sum of the  dimensions of its image (i.e., the rank of the map) and its kernel (i.e., the nullity of the map).  Since $\mathfrak{h}_k$ is a subspace of $\mathfrak{h}_l$, the kernel of the map  $A\rightarrow |\chi_A\rangle$ when restricted to  $\mathfrak{h}_k$ is contained in the kernel of this map, when the domain is $\mathfrak{h}_l$. It follows that 
\begin{align}
\text{dim}(\mathfrak{h}_l)-\text{dim}(\mathfrak{h}_k)&\ge \text{dim}(\mathcal{S}_l)-\text{dim}(\mathcal{S}_k)\ .
\end{align}
Combining this with corollary \ref{corcor} together with the fact that  $\text{dim}(\mathcal{S}_n)=\text{Irreps}_G(n)$,  proves the theorem. 
\end{proof}

\begin{remark}
Recall that $\text{dim}(\mathcal{S}_k)$ is the dimension of the  subspace of the center of $\mathfrak{h}_n$ that is included in  $\mathfrak{h}_k$. This means that the above lower bound  on  $\text{dim}(\mathfrak{h}_n)-\text{dim}(\mathfrak{h}_k)$ is due to the fact that  part of the center of $\mathfrak{h}_n$ is not produced by $k$-local $G$-invariant skew-Hermitian operators.
\end{remark}

\subsection{The general case of non-identical subsystems}

In the previous section we  focused on the special case where all the subsystems are identical and, in particular, they carry the same representation of the group $G$. However,  note that the general argument about charge vectors and, in particular, lemma \ref{Thm1} and \ref{lem2} are valid in the case of non-identical subsystems.
Using these lemmas it can be easily seen that the argument that proves the non-universality of local symmetric unitaries  can be generalized to the more general case  where the subsystems are not identical.  Here, we sketch the main idea.

  Assume there are a finite number of \emph{types}  of subsystems, where each type carries a particular representation of  group $G$. More precisely, suppose each subsystem has one of $T$  possible representations  $\{v^{(1)},\cdots ,v^{(T)} \}$, where for each $t\in \{1,\cdots, T\}$,  $\{v^{(t)}(g): g\in G\}$ is a finite-dimensional unitary representation of group $G$. 

Then, our previous argument can be easily generalized to show that $\mathcal{S}_k$, the set of charge vectors for $k$-local $G$-invariant Hamiltonians, is a finite-dimensional subspace, whose dimension is bounded by a number which is independent of $n$, the total number of sites. In fact, the dimension of $\mathcal{S}_k$ is upper bounded by  the total number of inequivalent irreps of $G$, which appear in all tensor product representations
\be
\bigotimes_{i=1}^k v^{(t_i)}  \ \ : \ t_1,\cdots, t_k\in\{1,\cdots , T\}\ .
\ee
This follows from the fact that any $k$-local operator  can act non-trivially on at most $k$ sites, and the representation of group $G$ on those $k$ sites is equivalent to one of the representations listed above. Clearly, the total number of inequivalent irreps appearing in the above representations, is independent of $n$, the total number of sites.

On the other hand, let $\bigotimes_{i=1}^n v^{(t_i)}$ be the representation of group $G$ on the total system, where   $v^{(t_i)}$ is the representation of group $G$ on  site $i$ and $t_i\in\{1,\cdots, T\}$.  For a compact connected Lie group $G$, such as U(1) and SU(2), as  the number of sites carrying a non-trivial representation of  $G$ grows, the number of distinct irreps which appear in this representation also increases unboundedly, and for sufficiently large $n$, this will be larger than the dimension of $\mathcal{S}_k$. 
  Therefore, by lemma \ref{Thm1} we conclude that for sufficiently large $n$, there are $G$-invariant unitaries which cannot be implemented using $k$-local $G$-invariant unitaries.



\newpage
\section{Supplementary Note 3: U(1) symmetry for a system of qubits}\label{App:B}

In this section we present a full characterization of Hamiltonians that can be implemented using $k$-local U(1)-invariant Hamiltonians on a system of qubits.  
It turns out that in this case the constraints found in the previous section in terms of charge  vectors, are both the necessary and sufficient conditions. 

While in this section we focus on the example of qubit systems, it is worth noting that similar results can be extended  to other representations of U(1) symmetry. This is relevant, for instance, in the context of quantum thermodynamics for systems with periodic Hamiltonians. These are Hamiltonians for which the ratio of the gaps between different eigenvalues are all rational numbers. That is, up to a constant shift, any energy level is an integer times a fixed energy scale  (Otherwise, if the eigenvalues do not have this form the generated group will not be compact and, in particular, will not be isomorphic to U(1)).  In Supplementary Note 7 We consider the specific case of these Hamiltonians where energy levels are equidistance. 

\subsection{Preliminaries}


Consider $n$ qubits, labeled as $j=1,\cdots, n$, and  the group of rotations around the z axis, i.e., unitaries 
\be\label{gerto}
U(e^{i\theta})=(e^{i\theta Z})^{\otimes n}=e^{i\theta \sum^n_{j=1} Z_j}=\sum_{m=0}^{n} e^{i\theta (n-2m)}\ \Pi_m=\sum_{l=0}^{n} (i\sin\theta)^{l}(\cos\theta)^{n-l}\  C_l   \ \ \ \ : \theta\in [0,2\pi) \ ,
\ee
where $ \Pi_m$ is the projector to the subspace corresponding to the eigenvalue $m$ of operator $ \sum_j (I-Z_j)/2$, also known as the subspace with the Hamming weight $m$ or charge $m$. The last equality in Eq.(\ref{gerto}) follows from the identity $e^{i\theta Z}=\cos\theta I+i\sin \theta Z$ together  with the definition
\be\label{defCl}
C_l\equiv
\sum_{\substack{{\bf{b}}:  w({\bf{b}})=l}\ }  {\bf{Z}}^{\bf{b}}=\sum^n_{m=0} c_l(m) \  \Pi_m\ \ \ \ :\ \ \ l=0,\cdots, n\ ,
\ee
where the first  summation in Eq.(\ref{defCl})  is over all bit strings with Hamming weight $l$  
and ${\bf{Z}}^{\bf{b}}=Z_1^{b_1}\cdots Z_n^{b_n}$, and  coefficient 
 \be\label{coef}
c_l(m)=\sum_{s=0}^m (-1)^s  {{m}\choose{s}} {{n-m}\choose{l-s}}\ ,
\ee
 is the eigenvalue of $C_l$ in the subspace with Hamming weight $m$. Note that  $C_0$ is the identity operator.  
 To see the second equality in Eq.(\ref{defCl}), first note that operator $C_l$ as defined by the first equality, is diagonal in the computational basis, and because of the permutational symmetry, the corresponding eigenvalue for the basis element $|z_1,\cdots, z_n\rangle$ only  depends on the Hamming weight of $z_1\cdots z_n$. This implies    $C_l=\sum_m c_l(m)\Pi_m$ with eigenvalues $\{c_l(m)\}$. Considering the expectation value of $C_l$ for the eigenvector $|1\rangle^{\otimes m} |0\rangle^{\otimes (n-m)}$, we find that $c_l(m)$ is equal to the sum of the expectation values of operators $\textbf{Z}^{\textbf{b}}$ in this  state,  
for all bit strings $\textbf{b}$ with Hamming weight $l$. Each expectation value is $\pm 1$. Then, a simple counting argument implies Eq.(\ref{coef}).  See Table \ref{default} in  Supplementary Note 4 for the example of $n=5$ qubits. 

In the following we use the fact that 
\be
\text{Span}_\mathbb{R}\{C_l: l=0,\cdots, n\}= \text{Span}_\mathbb{R}\{\Pi_m: m=0,\cdots, n\}\ .
\ee
Note that, up to an imaginary $i$ factor, this subspace is the center of $\mathfrak{h}_n$.    
  
  We study the group $\mathcal{V}_k^{\text{U}(1)}$ generated by $k$-local unitaries that are invariant under this symmetry.  It turns out that the constraints in lemma \ref{Thm1} on the charge vectors are the only constraints on Hamiltonians that can be simulated using these unitaries.

In the following, we refer to the basis $\{|0\rangle,|1\rangle\}^{\otimes n}$, as the computational basis of $n$ qubits  (Here, $\{|0\rangle, |1\rangle\}$ are eigenvectors of Pauli $Z$).

\subsection{Summary of main results}\label{Methods:char}

In this section we prove that
\begin{theorem}\label{Thm21}
For any  U(1)-invariant Hamiltonian $H$ on $n$ qubits the family of unitaries $\{e^{-i t H}: t\in\mathbb{R}\}$ can be implemented using $k$-local U(1)-invariant unitaries for $k\ge 2$, if, and only if
\be\label{App:conds2}
\Tr(H C_l)=0\ \ \ \  :\ l=k+1,\cdots, n\ .\\
\ee 
\end{theorem}
In terms of the corresponding Lie algebras this means
\begin{align}
i H \in \mathfrak{h}_{k} \  \ \Longleftrightarrow\  \ i H \in \mathfrak{h}_{n}\ \text{, and} \ \   \Tr(H C_l)=0\  \ \ \ \ \  :   l=k+1,\cdots, n\  ,
\end{align}
where 
\be
\mathfrak{h}_k\equiv \mathfrak{alg}_\mathbb{R}\big\{A: \text{$k$-local}, A+A^\dag=0\ ,\ [A, \sum_j Z_j]=0 \big\}\ ,
\ee
is the Lie algebra generated by $k$-local U(1)-invariant skew-Hermitian operators.  Note that  operators $iC_l\in \mathfrak{h}_n$ for $l=0,\cdots , n$ form a  linearly-independent set inside $\mathfrak{h}_n$. Therefore, each constraint  $\Tr(H C_l)=0$ in the above equation, reduces the dimension of the Lie algebra by one. Therefore, comparing $\mathfrak{h}_n$ and  $\mathfrak{h}_k$, we see that  operators in $\mathfrak{h}_k$ satisfy  $n-k$ additional  independent constraints, which means
\be
\text{dim}(\mathcal{V}_n^{\text{U}(1)})-\text{dim}(\mathcal{V}_k^{\text{U}(1)})=\text{dim}(\mathfrak{h}_n)-\text{dim}(\mathfrak{h}_k)=n-k\ .
\ee
We conclude that in this case our general lower bound in Eq.(\ref{Eq:dim3}) holds as equality. \\

\noindent \textbf{Overview of the proof of theorem \ref{Thm21}:} As we show in the following,  the necessity of the constraints in Eq.(\ref{App:conds2}) follows immediately from the constraints on the charge vectors in lemma \ref{Thm1} (See Sec.\ref{App:nec2021}). The  sufficiency of these conditions is proven in two steps: first, in Sec.\ref{App:Suff2021} we prove it in the special case of Hamiltonians that are diagonal in the computational basis $\{|0\rangle,|1\rangle\}^{\otimes n}$. Then, to extend the result to the  general case, we note that a general  U(1)-invariant Hamiltonian $H$ has a decomposition as $H=D+O$,  for a  Hermitian operator $D$ that is diagonal in the computational basis $\{|z_1,\cdots, z_n\rangle: z_j\in\{0,1\}\}$, 
plus a Hermitian operator $O$ with zero diagonal elements in this basis. 
This, in particular,  implies  $\Tr(O \Pi_m)=\Tr(O C_l)=0$ for $l, m=0,\cdots, n$.  We conclude that if $H=D+O$ satisfies the constraints in  Eq.(\ref{App:conds2}), then both $D$ and $O$ satisfy this constraint. From the special case of the result for diagonal Hamiltonians we know that $D$ can be realized by $k$-local U(1)-invariant unitaries, i.e., $iD\in\mathfrak{h}_k$.  Therefore, to extend the result to the general case it suffices to show that any U(1)-invariant Hamiltonian  $O$ with vanishing diagonal elements can be implemented  using  $k$-local U(1)-invariant unitaries, i.e., $iO\in \mathfrak{h}_k$. This follows from  the following lemma \ref{lem05}, which is proven in Sec.\ref{App:Sec:complete}.  
\begin{lemma}\label{lem05}
Suppose U(1)-invariant Hamiltonian $L$ satisfies the condition  $\Tr(L \Pi_m)=0$  for  $m=0,\cdots, n$. Then, for any time $t\in\mathbb{R}$, the unitary $\exp(-i t L)$ is in the group
\be
G_2\equiv \big\langle e^{i\theta (X_rX_s+Y_rY_s)}, e^{i\theta Z_r}:\theta\in[0,2\pi),r\neq s\in\{1... n\}\big\rangle\ ,
\ee 
i.e.,  can be implemented by single-qubit unitaries $\exp(i\theta Z)$ and 2-qubits unitaries $\exp(i\theta (XX+YY))$. 
\end{lemma}
For completeness, we rephrase the statement of this lemma in terms of the corresponding Lie algebras.\\

\noindent\textbf{Restatement of lemma} \ref{lem05}: 
The set of skew-Hermitian U(1)-invariant operators 
$\{L\in \mathfrak{h}_n: \Tr(L \Pi_m)=0\ :\ m=0,\cdots, n \}$ is a sub-algebra of  the Lie algebra generated by  operators $iR_{r,s}=i(X_rX_{s}+Y_rY_{s})/2$ together with $iZ_r$ for $r\neq s\in\{1,\cdots, n\}$, i.e., 
\be\label{hghgwq2}
\{L\in \mathfrak{h}_n: \Tr(L \Pi_m)=0\ ,\ m=0,\cdots, n \} \subset \mathfrak{alg}\{ i Z_r , i R_{r,s}: r\neq s\in\{1,\cdots n\}\}\ .
\ee

Therefore, proving this lemma completes the proof of theorem \ref{Thm21}. As we further explain in Sec.\ref{App:Sec:complete}, this lemma essentially means that the Lie algebra generated by operators $iR_{r,s}=i(X_rX_{s}+Y_rY_{s})/2$ together with $iZ_r$ for $r\neq s\in\{1,\cdots, n\}$ contains the  commutant sub-algebra of $\mathfrak{h}_n$. \\

This lemma has another useful corollary:

\begin{corollary}\label{cor81}
Any U(1)-invariant unitary $W=\bigoplus_{m=0}^n W_m$, satisfying  $\text{det}(W_m)=1: m=0,\cdots, n$   can be realized using 2-local unitaries $\exp(i (X_r X_s+Y_rY_s)\theta)$ and $\exp(i\theta Z_r)$, for $\theta\in[0,2\pi)$, $r\neq s\in\{1,\cdots, n\}$.   
\end{corollary}
This corollary follows from the fact that any unitary $W_m$ satisfying $\text{det}(W_m)=1$, can be written as $\exp(i H_m)$ for a traceless Hermitian operator $H_m$. Therefore, there exists a U(1)-invariant Hermitian operator $H=\bigoplus_{m=0}^n H_m$, such that $W=\exp(i H)$, and $\Tr(H_m)=0$. Then, the corollary follows immediately from the above lemma.

 \newpage
 
 \color{red}




\newpage

\color{black}

\newpage


\color{black}

\subsection{Charge vectors in the Lie algebra generated by local U(1)-invariant Hamiltonians \\ (Necessity of conditions in Eq.(\ref{App:conds2}) of theorem \ref{Thm21})}\label{App:nec2021}

For any operator $A$ acting on $(\mathbb{C}^2)^{\otimes n}$, the charge vector  is
\be
|\chi_A\rangle\equiv \sum_{m=0}^n \Tr(\Pi_m A)\ |m\rangle\ ,
\ee
and its Fourier transform is the function
\be\label{hgjfpq}
\chi_A(e^{i\theta})=\Tr(A U(e^{i\theta}) )=\Tr(A (e^{i\theta Z})^{\otimes n})=\sum_{l=0}^{n}\Tr(A C_l)\  (i\sin\theta)^{l}(\cos\theta)^{n-l}\  = \sum_{l=0}^{n} \Tr(A C_l)\ \xi_l(e^{i\theta}) \     \ \ \ \ : \theta\in [0,2\pi)\ ,
\ee
where
\be
 \xi_l(e^{i\theta})\equiv (\cos \theta)^{n-l}\ (i \sin\theta)^{l}\ ,
\ee
and we have used Eq.(\ref{gerto}). 

According to lemma \ref{lem2}, if operator $i A\in \mathfrak{h}_k$ then $\chi_A\in \tilde{\mathcal{S}}_k$, where 
\be
\tilde{\mathcal{S}}_k\equiv \{\chi_A:\  i A\in\mathfrak{h}_k  \ \}= \text{Span}_\mathbb{R}\Big\{\chi_A:\ A=A^\dag ,\  [A, \sum_r Z_r]=0,  A \text{ is $k$-local}\   \Big\}\ ,
\ee
i.e., $\tilde{\mathcal{S}}_k$ is the span of functions $\chi_A$ for all $k-$local Hermitian, U(1)-invariant operators.  Using Eq.(\ref{hgjfpq}) this can be rewritten as
\be
\tilde{\mathcal{S}}_k= \text{Span}_\mathbb{R}\Big\{\sum_{l=0}^{n} \Tr(A C_l)\ \xi_l :\ A=A^\dag ,\  [A, \sum_r Z_r]=0,  A \text{ is $k$-local}\   \Big\}\ .
\ee
For any $k$-local operator $A$, $\Tr(A C_l)=0$ for $l>k$. Furthermore, for any Hermitian operator $A$, $\Tr(A C_l)$ is a real number. This means $\tilde{\mathcal{S}}_k\subseteq \text{Span}_\mathbb{R}\{\xi_l: 0 \le l\le k \} $. Next, we show that this holds as equality. To see this consider U(1)-invariant Hermitian operator   $A= {\bf{Z}}^{\bf{b}}=Z_1^{b_1}\cdots Z_n^{b_n}$ for ${\bf{b}}=b_1\cdots b_n\in\{0,1\}^n$. This operator  satisfies  $\Tr(A C_l)=\Tr({\bf{Z}}^{\bf{b}} C_l)=2^n \delta_{w({\bf{b}}) , l}$,
 where $w({\bf{b}})=\sum_j b_j$. Using Eq.(\ref{hgjfpq}), this   implies $\chi_A= 2^n \xi_{w({\bf{b}})}$. Furthermore, if $w({\bf{b}})\le k$  then $A= {\bf{Z}}^{\bf{b}}$ is $k$-local, which implies $2^n \xi_{w({\bf{b}})}\in \tilde{\mathcal{S}}_k$. We conclude that
\be\label{9lkwgj}
\tilde{\mathcal{S}}_k=\text{Span}_\mathbb{R}\{\xi_l: 0 \le l\le k \} \ .
\ee
Therefore, lemma \ref{lem2} implies that if $iA\in \mathfrak{h}_k$, then 
\be
\sum_{l=0}^{n} \Tr(A C_{l})\ \xi_l\in \tilde{\mathcal{S}}_k= \text{Span}_\mathbb{R}\{\xi_{l'}: 0 \le l'\le k \}\ ,  
\ee
which means 
\be
\Tr(A C_{l})=0  :\  k< l\le n\ .
\ee
This proves the necessity of conditions in Eq.(\ref{App:conds2}). Equivalently, it implies
\be
\mathfrak{h}_{k}\subseteq \big\{A\in\mathfrak{h}_n :  \Tr(A C_l)=0,\  l\in\{ k+1,\cdots, n\}\big\}\ . 
\ee
In the rest of this Appendix, we prove the sufficiency of these conditions. We start with the case of diagonal operators.


\newpage

\subsection{Diagonal operators in the Lie algebra generated by local U(1)-invariant Hamiltonians\\ (Proof of theorem \ref{Thm21} in the special case of diagonal Hamiltonians)}\label{App:Suff2021}


Consider an arbitrary diagonal Hamiltonian
\be
H=\sum_{{\bf{z}}\in\{0,1\}^n}  h({\bf{z}})\ |{\bf{z}}\rangle\langle {\bf{z}}|=\sum_{{\bf{b}}\in\{0,1\}^n}  \tilde{h}({\bf{b}})\ {\bf{Z}}^{\bf{b}}\ ,
\ee
where 
\be
\tilde{h}({\bf{b}})=2^{-n}\sum_{{\bf{z}}\in\{0,1\}^n} (-1)^{{\bf{b}}\cdot {\bf{z}}}\ h({\bf{z}})=2^{-n}\ \Tr({\bf{Z}}^{\bf{b}} H)\ ,
\ee
 is the Fourier transform of $h(\bf{z})$.   Then,  the condition in Eq.(\ref{App:conds2}) is equivalent to
\be\label{App:conds22}
 \sum_{{\bf{b}}\in\{0,1\}^n: w({\bf{b}})=l } \tilde{h}({\bf{b}})=0 \ \ \ \  :\ l=k+1,\cdots, n\ ,
\ee
where the summation is over all bit strings with Hamming weight $l$.  We find that $iH \in\mathfrak{h}_k$, and hence the family of diagonal unitaries $\{e^{-i t H}: t\in\mathbb{R}\}$ can be generated using $k$-local U(1)-invariant Hamiltonians  only if Eq.(\ref{App:conds22}) holds. The sufficiency of these conditions follows immediately from the following lemma.

Recall the definitions
 \be
 R_{rs}=\frac{X_r X_s+Y_r Y_s}{2}\ \  ,\ \  \ T_{rs}= \frac{i}{2} [Z_r, R_{rs}]=\frac{X_r Y_s-Y_r X_s}{2}\ .
 \ee
We prove the following lemma,   which is a special case of lemma \ref{lem05} for diagonal operators.    \color{black} 
\begin{lemma}\label{lem41}
For  $n\ge 2$ qubits labeled as $1,\cdots, n$, it holds that
\begin{align}\label{kjeb}
\Big\{A=i\sum_{\textbf{b}} a_{\textbf{b}}\ \textbf{Z}^{\textbf{b}}: \ A+A^\dag=0 \ ,  \Tr(A \Pi_m)=0\ :   m=0,\cdots, n\Big\}\ &=\Big\{i\sum_{\textbf{b}} a_{\textbf{b}}\ \textbf{Z}^{\textbf{b}}: \ a_{\textbf{b}}\in\mathbb{R} ,\sum_{\textbf{b}: w(\textbf{b}) =l}a_{\textbf{b}}=0\ :   l=0,\cdots, n\Big\}\ \color{black} \nonumber\\  &\subset \mathfrak{alg}_\mathbb{R}\Big\{i R_{rs}, i Z_r: r\neq s\in\{1,\cdots, n\} \Big\}\ ,
\end{align}
where $\sum_{\textbf{b}: w(\textbf{b}) =l}$ is the summation over all bit strings with Hamming weight $l$, and the second line is the real Lie algebra generated by 2-local U(1)-invariant operators $\{i R_{rs}, i Z_r : r\neq s\in\{1,\cdots, n\} \}$.  
 \end{lemma}

The lemma implies that   the subspace in Eq.(\ref{kjeb}) is a subspace of $\mathfrak{h}_k$ for $k\ge 2 $, where  $\mathfrak{h}_k$ is the Lie algebra generated by $k$-local U(1)-invariant skew-Hermitian operators.  By definition, in addition to this subspace, $\mathfrak{h}_k$ also includes  arbitrary linear combinations of operators $\{i\textbf{Z}^{\textbf{b}}: w(\textbf{b})\le k\}$. Linear combinations of these operators with the set of operators in the left-hand side of Eq.(\ref{kjeb}), yield all diagonal Hamiltonians satisfying  condition in Eq.(\ref{App:conds22}) which is equivalent to the condition in Eq.(\ref{App:conds2}) in theorem \ref{Thm21}. This complete the proof of the theorem \ref{Thm21}  in the special case of diagonal Hamiltonians.

In the rest of this section, we prove lemma \ref{lem41}. To prove this lemma we use the fact that,  for any subset $t\le n$ distinct qubits  $l_1, l_2,\cdots, l_{t}\in\{1,\cdots, n\} $, we have 
\begin{align}\label{kjf00}
i c_t \times (Z_{l_1}...Z_{l_{t-1}}-Z_{l_2}...Z_{l_{t}})=  
&\big[[..[ [i R_{l_1l_2}, i R_{l_2l_3} ], i R_{l_3l_4}].., i R_{l_{t-1}, l_{t}}], i R_{l_{t}, l_{1}}\big]&&: t \text{ odd} ,\\  
&\big[[..[ [i R_{l_1l_2}, i R_{l_2l_3} ], i R_{l_3l_4}].., i R_{l_{t-1}, l_{t}}],  i T_{l_{t}, l_{1}}\big]&&: t\  \text{even}\ ,\nonumber 
 \end{align}
 where $c_t=\pm1$, depending on $t$. Because of the usefulness of this commutation relation, for completeness we repeat  it again, using a slightly  different notation: 

\color{black}
\begin{empheq}[box=\fbox]{align}
\nonumber \\ 
\ \ &R_{rs}\equiv\frac{X_r X_s+Y_r Y_s}{2} \  \  :r\neq s\nonumber\\ 
\nonumber & \\ 
\ \ &\forall m\ge 2:\ \  \ \ \ \  (Z_1-Z_m)\ Z_{1}\cdots Z_{m}=  
\begin{cases}\!
  \begin{aligned}[b]    
  && c_m\ \Big[R_{1,{m}}\ ,\   \big[R_{m,{m-1}}\ , \ \cdots  [R_{4,3}\ ,\  [R_{3,2}, R_{2,1}]]\cdots \big]\Big]\ \ \  : &&&m \text{ odd} \\ 
    \nonumber \\
        && c_m\ \Big[R_{1,{m}}\ ,\   \big[R_{m,{m-1}}\ , \cdots\  [R_{3,2}\ , \ [R_{2,1}, \frac{Z_{1}}{2}]]\cdots \big]\Big]
    \ \ \  : &&&m\  \text{even}\
  \end{aligned} \end{cases}
\nonumber \\ \nonumber 
\\ \ \ \ &c_m=\pm1 \nonumber\\ \nonumber 
\end{empheq}

\begin{proof}(lemma \ref{lem41})
To see the first line of Eq.(\ref{kjeb}) note that the set of operators $\{\Pi_m: m=0,\cdots, n\}$ and $\{C_l=\sum_{\textbf{b}: w(\textbf{b})=l} \textbf{Z}^{\textbf{b}} : l=0,\cdots , n\}$ span the same $(n+1)$-dimensional space. Therefore, the condition $\Tr(A \Pi_m)=0$ for all $m=0,\cdots, n$ is equivalent to the condition $\Tr(A C_l)=0$ for all $l=0,\cdots, n$. 
For diagonal operator $A=\sum_{\textbf{b}} a_{\textbf{b}} \textbf{Z}^{\textbf{b}}$,  using the relation $\Tr(\textbf{Z}^{\textbf{b}} C_l)=2^n \delta_{l, w(\textbf{b})}$, we find that this condition is equivalent to $\sum_{\textbf{b}: w(\textbf{b})=l} a_{\textbf{b}} =0$ for $l=0,\cdots , n$. This proves the first line in Eq.(\ref{kjeb}). In the following, we show that  any diagonal operator  satisfying this condition can be written as a linear combination of  the commutators in Eq.(\ref{kjf00}), and hence is in $\mathfrak{alg}_\mathbb{R}\big\{i R_{rs}, i Z_r: r\neq s\in\{1,\cdots, n\} \big\}$.

For any pair of bit strings  $\textbf{b}_1, \textbf{b}_2\in\{0,1\}^n$,  let  $d(\textbf{b}_1, \textbf{b}_2) $ be their Hamming distance, i.e., the number of bits that should be flipped to transform one bit string to another. Using Eq.(\ref{kjf00}),  for any pair of bit strings $\textbf{b}_1, \textbf{b}_2\in\{0,1\}^n$, the operator $ i(\textbf{Z}^{\textbf{b}_1}-\textbf{Z}^{\textbf{b}_2})$ can be obtained from these commutators,  provided that $\textbf{b}_1$ and $\textbf{b}_2 $ have equal Hamming weights, i.e. $w(\textbf{b}_1)=w(\textbf{b}_2)=t-1$  and their Hamming distance $d(\textbf{b}_1, \textbf{b}_2)=2$. This means that the linear span of operators in Eq.(\ref{kjf00}) for a fixed $t$ in the interval $2\le t\le n$ contains all operators
\be
\text{Span}_\mathbb{R}\Big\{ i(\textbf{Z}^{\textbf{b}_1}-\textbf{Z}^{\textbf{b}_2}): w(\textbf{b}_1)=w(\textbf{b}_2)=t-1, d(\textbf{b}_1, \textbf{b}_2)=2 \Big\}\subset \mathfrak{alg}_\mathbb{R}\big\{i R_{rs}, i Z_r: r\neq s\in\{1,\cdots, n\} \big\}  \ .
\ee
Next, we prove that  the restriction $d(\textbf{b}_1, \textbf{b}_2)=2$ in the left-hand side can be removed, that is  
\begin{align}\label{lkwf21}
&\text{Span}_\mathbb{R}\Big\{ i(\textbf{Z}^{\textbf{b}_1}-\textbf{Z}^{\textbf{b}_2}): w(\textbf{b}_1)=w(\textbf{b}_2)=t-1, d(\textbf{b}_1, \textbf{b}_2)=2 \Big\}
=\text{Span}_\mathbb{R}\Big\{i(\textbf{Z}^{\textbf{b}_1}-\textbf{Z}^{\textbf{b}_2}) : \   w(\textbf{b}_1)=w(\textbf{b}_2)=t-1\Big\}\ ,
\end{align}
i To prove this we use  the fact that any pair of bit strings strings $\textbf{c}_1, \textbf{c}_2\in\{0,1\}^n$ with equal Hamming weights $w(\textbf{c}_1)=w(\textbf{c}_2)=t-1$ are related to each other by a permutation of bits. Furthermore, any permutation can be realized by a sequence of  \emph{transpositions}, i.e., 2-bit permutations, which only exchange the value of two-bits. It follows that for any pair of bit strings $\textbf{c}_1, \textbf{c}_2\in\{0,1\}^n$ with equal Hamming weights $w(\textbf{c}_1)=w(\textbf{c}_2)=t-1$, there is a path in the space of bit strings with Hamming weight $t-1$  from $\textbf{c}_1$ to $\textbf{c}_2$,  i.e.,  
\be
\textbf{f}_1,\cdots, \textbf{f}_L\in \{0,1\}^n :\ \   \ w(\textbf{f}_k)=t-1,\ \ \  \textbf{f} _1=\textbf{c}_1,\ \textbf{f}_L=\textbf{c}_2\ ,
\ee
and
\be
\textbf{c}_1=\textbf{f}_1\longrightarrow \textbf{f}_2\longrightarrow \cdots \longrightarrow  \textbf{f}_L=\textbf{c}_2\ ,
\ee
where each consecutive pair of bit strings have Hamming distance 2, i.e.
\be
d(\textbf{f}_r, \textbf{f}_{r+1})=2:\ \ \   1\le r\le L-1\ .
\ee

 Therefore, $i(\textbf{Z}^{\textbf{c}_1}-\textbf{Z}^{\textbf{c}_2})$ can be obtained using the linear combination 
\be
i(\textbf{Z}^{\textbf{c}_1}-\textbf{Z}^{\textbf{c}_2})=i(\textbf{Z}^{\textbf{f}_1}-\textbf{Z}^{\textbf{f}_L})= i(\textbf{Z}^{\textbf{f}_1}-\textbf{Z}^{\textbf{f}_2})+i(\textbf{Z}^{\textbf{f}_2}-\textbf{Z}^{\textbf{f}_3})+\cdots +i(\textbf{Z}^{\textbf{f}_{L-1}}-\textbf{Z}^{\textbf{f}_L})\ .
\ee
  This proves Eq.(\ref{lkwf21}).  Next, it can be easily seen that 
\be\label{kwfkjf76}
\text{Span}_\mathbb{R}\Big\{i(\textbf{Z}^{\textbf{b}_1}-\textbf{Z}^{\textbf{b}_2}) : \   w(\textbf{b}_1)=w(\textbf{b}_2)=t-1\Big\}=\Big\{i\sum_{\textbf{b}: w(\textbf{b})=t-1} a_{\textbf{b}}\ \textbf{Z}^{\textbf{b}}: \ \ a_{\textbf{b}}\in\mathbb{R}\ , \  \sum_{\textbf{b}: w(\textbf{b}) =t-1} a_{\textbf{b}} =0  \Big\}\ ,
\ee
where the right-hand side is the subspace of  all linear combinations $i\sum_{\textbf{b}: w(\textbf{b})=t-1} a_{\textbf{b}}\ \textbf{Z}^{\textbf{b}}$    for bit strings with Hamming weight $t-1$, which satisfy the linear constraint $\sum_{\textbf{b}: w(\textbf{b}) =t-1} a_{\textbf{b}} =0$. Recall that  $t$ can take all values in $\{2,\cdots, n\}$, which means the Hamming weight of bit strings $\textbf{b}$ takes values between $1$ to $n-1$. In other words, the two cases of bit strings $\textbf{b}=0^n$ and  $\textbf{b}=1^n$, which correspond to operators $I$ and $Z^{\otimes n}$ cannot be obtained in this way. It follows that the linear combination of operators in  Eq.(\ref{kwfkjf76}) is equal to 
\be
\big\{i\sum_{\textbf{b}} a_{\textbf{b}}\ \textbf{Z}^{\textbf{b}}: \ a_{\textbf{b}}\in\mathbb{R} ,\sum_{\textbf{b}: w(\textbf{b}) =l}a_{\textbf{b}}=0\ :   l=0,\cdots, n\big\}\ . 
\ee
We conclude that this set of operators 
can be obtained as a linear combination of the commutators in Eq.(\ref{kjf00}), and therefore is contained in $\mathfrak{alg}_\mathbb{R}\big\{i R_{rs}, i Z_r: r\neq s\in\{1,\cdots, n\} \big\}$.  This completes the proof of lemma \ref{lem41}.
\end{proof}

\subsection{From diagonal Hamiltonians to all symmetric Hamiltonians}\label{Sec:fg} 
In this section, we prove that if one can implement  
all diagonal Hamiltonians as well as 2-local Hamiltonians  $\{R_{j, j+1}=(X_jX_{j+1}+Y_jY_{j+1})/2 : j=1,\cdots, n\}$,  then one can implement all U(1)-invariant unitaries, i.e., those commuting with $\sum_j Z_j$.  In the next section, we use this result  to prove theorem \ref{Thm21} in the general case. We also apply this result in Appendix \ref{App:C} to prove theorem \ref{Thm:energy} and show that a single ancillary qubit suffices to circumvent the no-go theorem.

The formal version of this result is stated in the following theorem: 

\begin{theorem}\label{Thm5}
The real Lie algebra generated by the set of diagonal skew-Hermitian operators   
and operators $\{i R_{j, j+1}=i(X_jX_{j+1}+Y_{j}Y_{j+1})/2 : j=1,\cdots, n-1\}$ is equal to the set of all skew-Hermitian $U(1)$-invariant operators, i.e., those commuting with $\sum_j Z_j$. In other words,
\be
\mathfrak{h}_n\equiv \Big\{A\in \mathcal{L}({(\mathbb{C}^2)}^{\otimes n}): A+A^\dag=0\ , \  \big[A, \sum^n_{r=1} Z_r\big]=0  \Big\}= \mathfrak{alg}_\mathbb{R}\Big\{i R_{j, j+1} , i\textbf{Z}^{\textbf{b}}: \textbf{b}\in\{0,1\}^n, j\in \{1,\cdots, n-1\} \Big\}\ .
 \ee
\end{theorem} 
\begin{proof} 
It is clear that the Lie algebra $\mathfrak{alg}_\mathbb{R}\Big\{i R_{j, j+1} , i\textbf{Z}^{\textbf{b}}: \textbf{b}\in\{0,1\}^n, j\in \{1,\cdots, n-1\} \Big\}$ is contained in $\mathfrak{h}_n$. Here, we prove the converse, i.e., we show 
\be
\mathfrak{h}_n \subseteq \mathfrak{alg}_\mathbb{R}\Big\{i R_{j, j+1} , i\textbf{Z}^{\textbf{b}}: \textbf{b}\in\{0,1\}^n, j\in \{1,\cdots, n-1\} \Big\}\ .
\ee
Any arbitrary operator $A\in \mathcal{L}({(\mathbb{C}^2)}^{\otimes n})$ can be written as 
\be
A=\sum_{{\bf{b}},{\bf{b}}'\in\{0,1\}^n} a_{{\bf{b}},{\bf{b}}'}\ |{\bf{b}}\rangle\langle{\bf{b}}'|\ .
\ee
Using the fact that
\be
\big(\sum^n_{r=1} Z_r\big) |{\bf{b}}\rangle=[n-2 w({\bf{b}})]\  |{\bf{b}}\rangle\ ,
\ee
we find that
\be
[A, \sum^n_{r=1} Z_r]=2\sum_{{\bf{b}},{\bf{b}}'} a_{{\bf{b}},{\bf{b}}'} [w({\bf{b}})-w({\bf{b}'})]\ |{\bf{b}}\rangle\langle{\bf{b}}'|\ .
\ee
This implies that if $[A, \sum^n_{r=1} Z_r]=0$, then
\be
a_{{\bf{b}},{\bf{b}}'}=0\ \ \ \  \text{for}\ \  w({\bf{b}})\neq w({\bf{b}}')\ .
\ee
In other words, the off-diagonal terms for bit strings with different Hamming weights vanish. Therefore, the space of $U(1)$-invariant operators is spanned by 
\be
\{|{\bf{b}}\rangle\langle{\bf{b}}'|: w({\bf{b}})=w({\bf{b}}');\  {\bf{b}}, {\bf{b}'}\in\{0,1\}^n \}\ .
\ee
This implies that $\mathfrak{h}_n$, the space of skew-Hermitian $U(1)$-invariant operators  is spanned by 
\be
\mathfrak{h}_n=\text{Span}_\mathbb{R}\Big\{i\big(|{\bf{b}}\rangle\langle {\bf{b}}'|+|{\bf{b}}'\rangle\langle {\bf{b}}|\big) ,   |{\bf{b}}\rangle\langle {\bf{b}}'|-|{\bf{b}}'\rangle\langle {\bf{b}}|:   w({\bf{b}})= w({\bf{b}}')\ ; \ {\bf{b}}, {\bf{b}}' \in\{0,1\}^n \Big\}\  .
\ee
Using the fact that for any pair of bit strings ${\bf{b}},{\bf{b}'}\in\{0,1\}^n$,  
\be
\Big[ i |{\bf{b}}\rangle\langle {\bf{b}}|\ ,\  ( |{\bf{b}}\rangle\langle {\bf{b}}'|-|{\bf{b}}'\rangle\langle {\bf{b}}|) \Big]=i \big(|{\bf{b}}\rangle\langle {\bf{b}}'|+|{\bf{b}'}\rangle\langle {\bf{b}}|\big)\ ,
\ee
we find that the Lie algebra $\mathfrak{h}_n$ is generated by 
\be
\mathfrak{h}_n =\mathfrak{alg}\Big\{i |{\bf{b}}\rangle\langle {\bf{b}}|\ ,\  |{\bf{b}}\rangle\langle {\bf{b}}'|-|{\bf{b}}'\rangle\langle {\bf{b}}|:   w({\bf{b}})= w({\bf{b}}')\ ; \ {\bf{b}}, {\bf{b}}' \in\{0,1\}^n \Big\}\ .
 \ee
Next, we prove that this algebra is generated by the following set of operators
\be
\big\{i |{\bf{b}}\rangle\langle {\bf{b}}|\big\}\cup \big\{i R_{j, j+1}=i(X_jX_{j+1}+Y_{j}Y_{j+1})/2 : j=1,\cdots, n-1 \big\}\ ,
\ee
i.e., we prove that
\bes\label{ope}
\begin{align}
&\mathfrak{alg}\Big\{ i |{\bf{b}}\rangle\langle {\bf{b}}|, i R_{j, j+1}=i(X_jX_{j+1}+Y_{j}Y_{j+1})/2 : j=1,\cdots, n-1 , {\bf{b}}\in\{0,1\}^n  \Big\}\\ &\ \ \ \ \ \ \ \ = \mathfrak{alg}\Big\{i |{\bf{b}}\rangle\langle {\bf{b}}|\ ,\  |{\bf{b}}\rangle\langle {\bf{b}}'|-|{\bf{b}}'\rangle\langle {\bf{b}}|:   w({\bf{b}})= w({\bf{b}}')\ ; \ {\bf{b}}, {\bf{b}}' \in\{0,1\}^n \Big\}=\mathfrak{h}_n\ .
 \end{align}
\ees
To prove this claim, first note that for any bit string ${\bf{b}}\in\{0,1\}^n$, and any pair of distinct qubits $l,r\in\{1,\cdots, n\}$, it holds that
\begin{align}
\big[ i|{\bf{b}}\rangle\langle {\bf{b}}| , i R_{l r}  \big]= |{\bf{b}'}\rangle\langle {\bf{b}}|-|{\bf{b}}\rangle\langle {\bf{b}'}|\equiv F({\bf{b}'}, {\bf{b}})\ ,
\end{align}
where ${\bf{b}'}$ is the bit string obtained by exchanging bits $l$ and $r$ of bit string ${\bf{b}}$, and for any pair of bit strings ${\bf{d}}$ and ${\bf{e}}$, we have defined the notation
\be
F({\bf{d}}, {\bf{e}})\equiv |{\bf{d}}\rangle\langle {\bf{e}}|-|{\bf{e}}\rangle\langle {\bf{d}}|\ . 
\ee
Next, note that for any three distinct bit strings 
${\bf{b}}, {\bf{b}'}, {\bf{b}''}\in\{0,1\}^n$, it holds that 
\be\label{lkhwr}
F({\bf{b}}, {\bf{b}''})=\Big[F({\bf{b}}, {\bf{b}'}), F({\bf{b}'}, {\bf{b}''}) \Big]\ . 
\ee
By combining these two steps, we can obtain  $F({\bf{c}}_1, {\bf{c}}_2)= |{\bf{c}}_1\rangle\langle {\bf{c}}_2|-|{\bf{c}}_2\rangle\langle {\bf{c}}_1|$, for any pair of bit strings ${\bf{c}}_1, {\bf{c}}_2\in\{0,1\}^n$ with equal Hamming weights:  Recall that any pair of bit strings with equal Hamming weights are related via 
a permutation and any such permutation can be realized by combining \emph{transpositions}, i.e., 2-bit permutations. Therefore, there exists a sequence 
\be\label{keke}
{\bf{c}_1}={\bf{b}_1}\longrightarrow{\bf{b}_2}\ \cdots, \longrightarrow {\bf{b}_L}={\bf{c}_2}\ ,
\ee
where
${\bf{b}_1}={\bf{c}_1}$, 
 ${\bf{b}_L}={\bf{c}_2}$, $d({{\bf{b}}_p},{{\bf{b}}_{p+1}})=2$ for $1\le p\le L-1$. In fact, because any permutation on $n$ bits can be generated by transpositions  on nearest-neighbor pairs of bits $j$ and $j+1$, for $j=1,\cdots, n-1$, in the chain in Eq.(\ref{keke}), we can assume any two consecutive bit strings  ${\bf{b}_p}$  and ${\bf{b}_{p+1}}$ are identical for all bits, except a pair of nearest-neighbor bits $j$ and $j+1$.

 Then, using Eq.(\ref{lkhwr}) we have
 \begin{align}
F({\bf{c}_1}, {\bf{c}_2}) =F({\bf{b}_1}, {\bf{b}_L})=\Big[\big[\cdots \big[\big[[F({\bf{b}_1}, {\bf{b}_2}), F({\bf{b}_2}, {\bf{b}_3}) ],F({\bf{b}_3}, {\bf{b}_4})\big],F({\bf{b}_4}, {\bf{b}_5})  \big]\cdots \big],F({\bf{b}_{L-1}}, {\bf{b}_L})  \Big]\ .
\end{align}
This proves Eq.(\ref{ope}), i.e., the Lie algebra $\mathfrak{h}_n$ is generated by operators  $\{ i |{\bf{b}}\rangle\langle {\bf{b}}|, i R_{j , j+1}\}$, and  completes proof of theorem \ref{Thm5}.
\end{proof}

\subsection{ The commutant sub-algebra of U(1)-invariant operators $\mathfrak{h}_n$ (Proof of lemma \ref{lem05})}\label{App:Sec:complete}

In this section we prove lemma \ref{lem05}, which completes the proof of theorem \ref{Thm21}. 
This lemma states that any U(1)-invariant skew-Hermitian operator whose trace is zero in all charge sectors is in the Lie algebra generated by  operators $iR_{r,s}=i(X_rX_{s}+Y_rY_{s})/2$ together with $iZ_r$ for $r\neq s\in\{1,\cdots, n\}$.  More precisely, we show that
\be\label{haka}
\big\{L\in \mathfrak{h}_n: \Tr(L \Pi_m)=0\ :\ m=0,\cdots, n \big\}=[\mathfrak{h}_n, \mathfrak{h}_n]\subset  \mathfrak{alg}\{i Z_r , i R_{r,s}: r\neq s\in\{1,\cdots n\}\}\ ,
\ee
where  $[\mathfrak{h}_n, \mathfrak{h}_n]$ denotes  the commutator sub-algebra of $\mathfrak{h}_n$, i.e.,  the  Lie algebra generated by $[ A_1, A_2]$, for all $A_1,A_2\in \mathfrak{h}_n$.

To see the equality in Eq.(\ref{haka}), first, 
 note that $\Tr(\Pi_m [A_1,A_2])=\Tr([\Pi_m, A_1]A_2)=0$, which follows from the fact that $A_1, A_2\in \mathfrak{h}_n$ are block-diagonal with respect to $\{\Pi_m\}$. Since all elements of $[\mathfrak{h}_n, \mathfrak{h}_n]$ can be written as linear combinations of such commutators, we conclude that  
 \be
[\mathfrak{h}_n, \mathfrak{h}_n]\subseteq \{L\in \mathfrak{h}_n: \Tr(L \Pi_m)=0\ :\ m=0,\cdots, n \}\ .
 \ee
To show that the equality holds consider an arbitrary operator $L\in \mathfrak{h}_n$ satisfying $ \Tr(L \Pi_m)=0\ :\ m=0,\cdots, n $.  Because of U(1) symmetry, $L$ is block-diagonal with respect to $\{\Pi_m\}$, that is  $L=\sum_{m=0}^n \Pi_m L\Pi_m$. Furthermore, the condition $\Tr(L\Pi_m)=0$, implies that $\Pi_mL\Pi_m$ is a traceless skew-Hermitian operator with support restricted to the subspace of states with Hamming weight $m$. Recall  that any traceless skew-Hermitian operator in $\mathbb{C}^d$  for $d\ge 2$,   can be written as a linear combination of the commutators of skew-Hermitian operators. In other words, the commutator sub-algebra of $\mathfrak{su}(d)$ is equal to $\mathfrak{su}(d)$, i.e.,  $\mathfrak{su}(d)=[\mathfrak{su}(d),\mathfrak{su}(d)]$  \cite{fulton2013representation}. Therefore, the traceless skew-Hermitian operator   $\Pi_m L\Pi_m$   can be written as a linear combination of the commutators of skew-Hermitian operators with support restricted to the subspace with Hamming weight $m$. Because all such skew-Hermitian operators belong to $\mathfrak{h}_n$, we conclude that  $\Pi_m L\Pi_m\in [\mathfrak{h}_n,\mathfrak{h}_n]$, which in turn implies $L=\sum_m \Pi_m L\Pi_m\in  [\mathfrak{h}_n,\mathfrak{h}_n]$.  In summary, we found  
  \be\label{lab2}
\{L\in \mathfrak{h}_n: \Tr(L \Pi_m)=0\ :\ m=0,\cdots, n \}=[\mathfrak{h}_n, \mathfrak{h}_n]\ . 
 \ee
Next, applying lemma \ref{lem41} and theorem \ref{Thm5}, we argue that this Lie algebra is a sub-algebra of 
$\mathfrak{alg}\{i Z_r , i R_{r,s}: r\neq s\in\{1,\cdots n\}\}$. 

Let $\mathfrak{z}_n\subset \mathfrak{h}_n$ be the set of all diagonal skew-Hermitian operators. Lemma \ref{lem41} implies that
\be
\mathfrak{z}_n\cap [\mathfrak{h}_n, \mathfrak{h}_n]=\big\{A\in\mathfrak{z}_n: \Tr(A \Pi_m)=0 , m=0,\cdots, n \big\} \subset  \mathfrak{alg}\{ i Z_r , i R_{r,s}: r\neq s\in\{1,\cdots n\}\}\ ,
\ee
where the first equality follows from Eq.(\ref{lab2}). This means that $\mathfrak{alg}\{ i Z_r , i R_{r,s}: r\neq s\in\{1,\cdots n\}\}$ contains all elements of $\mathfrak{z}_n$  satisfying the constraint $\Tr(A \Pi_m)=0: m=0,\cdots, n $.  An arbitrary element of $\mathfrak{z}_n$  can be written as a linear combination of an operator satisfying this constraint  with operators   $\{i\Pi_m: m=0,\cdots n\}$. In other words, by adding $\{i\Pi_m: m=0,\cdots n\}$ to  $\mathfrak{alg}\{ i Z_r , i R_{r,s}: r\neq s\in\{1,\cdots n\}\}$, we obtain   all diagonal skew-Hermitian operators, i.e.,   
\be
\mathfrak{z}_n \subset  \mathfrak{alg}\Big\{i\Pi_m, i Z_r , i R_{r,s}: r\neq s\in\{1,\cdots n\}, m\in \{0,\cdots, n\}\Big\}\ .
\ee
Next, recall that, according to theorem \ref{Thm5},  diagonal skew-Hermitian operators together with operators $i R_{j,j+1}=i (X_j X_{j+1}+Y_j Y_{j+1}):\  j=1,\cdots, n-1$ generate all  U(1)-invariant skew-Hermitian  operators, that is
\be
\mathfrak{h}_n=\mathfrak{alg}(\mathfrak{z}_n \cup \{i R_{j,j+1}: j=1,\cdots n-1\})\ .
\ee
Combining the above two equations, we conclude that
\be
\mathfrak{h}_n= \mathfrak{alg}\big\{i\Pi_m,   i Z_r , i R_{r,s}: r\neq s\in\{1,\cdots n\}\ , m\in \{0,\cdots, n\} \big\}\ .
\ee
Finally, we consider the commutator sub-algebra of both sides. In the right-hand side, because  $i\Pi_m$ commutes with all elements of $\mathfrak{h}_n$, i.e., is in the center of $\mathfrak{h}_n$, it disappears in the commutator sub-algebra. That is the commutator subs-algebra of the right-hand side is a sub-algebra of $\mathfrak{alg}\{i Z_r , i R_{r,s}: r\neq s\in\{1,\cdots n\}\}$. We conclude that  
\be
[\mathfrak{h}_n,\mathfrak{h}_n] \subset  \mathfrak{alg}\{i Z_r , i R_{r,s}: r\neq s\in\{1,\cdots n\}\}\  .
\ee
Together with Eq.(\ref{lab2}), this implies Eq.(\ref{haka}) and proves lemma \ref{lem05}. This completes the proof of theorem \ref{Thm21}.

\color{black}

\newpage

 \section{Supplementary Note 4: Restrictions  on the realizable unitaries}\label{Sec:Method:lbody}

So far, we have focused on the constraints imposed by the locality and symmetry on the set of realizable Hamiltonians. On the other hand, for certain  applications, it is useful to also characterize the constraints on the realizable  unitaries.  Let $V$ be the unitary generated  by the time evolution of the system under $G$-invariant Hamiltonian $H(t)$, from time $t=0$ to $T$. This unitary is given by the time-ordered integral 
\be
V=\mathcal{T}\big\{\exp\big( {-i \int_{0}^T H(t) dt}\big) \big\}=\lim_{L\rightarrow \infty}\prod_{j=1}^L \exp({-\frac{i T}{L} H(\frac{T j}{L}) })\ .
\ee

The symmetry  implies  $H(t)$ and $V$ decompose
as $\bigoplus_{\mu\in\text{Irreps}_G(n)} H_\mu(t)$ and
$\bigoplus_{\mu\in\text{Irreps}_G(n)} V_{\mu}$, respectively. Here,  $H_\mu(t)$ and $V_\mu$ act on the subspace corresponding to irrep $\mu$,  and  
\be
V_\mu=\lim_{L\rightarrow \infty}\prod_{j=1}^L \exp({-\frac{i T}{L} H_\mu(\frac{T j}{L}) })\ .
\ee 
Using the fact that for any pair of Hermitian operators $A_1$ and $A_2$, $\text{det}(\exp({i A_1})\exp({i A_2}))=\exp({i\Tr(A_1)+i\Tr(A_2)})$, this implies   
\be
\text{det}(V_\mu)=\exp({-i\int^T_0  \Tr(H(t)\Pi_\mu) dt})\ .
\ee
 We conclude  that  for any set of integers $\{c(\mu)\}$, it holds that 
\be\label{ary2}
\hspace{-2mm}\Phi\equiv\hspace{-4mm}\sum_{\mu\in\text{Irreps}_G(n)}\hspace{-2mm} \hspace{-3mm}  c(\mu)\ \text{arg}(\text{det}(V_\mu))\hspace{-1mm} =-\hspace{-1mm}\int^T_0 \hspace{-1mm}dt\ \Tr(H(t)C)\ \ \ \ \ \ \ \ \ :\text{mod }2\pi\ ,  
\ee
where $\text{arg}(\text{det}(V_\mu))\in(-\pi, \pi]$ is the phase of $\text{det}(V_\mu)$, and  $C=\sum_{\mu} c(\mu)\Pi_\mu$. This, in particular, means that if Hamiltonians $H_1(t)$ and $H_2(t)$ generate the same unitary $V$, then 
\be
\int_0^T dt\  \Tr(H_1(t)C)=\int_0^T dt\  \Tr(H_2(t)C)\ \ \ \ \ \ \ \ \   :\text{mod } 2\pi\ .
\ee  

If  operator $C$ is traceless then $\Phi$ remains invariant under a global phase transformation $V\hspace{-1mm}\rightarrow\hspace{0mm} e^{i\alpha} V$. This can be seen, e.g., using  the 
 fact that if Hamiltonian $H(t): 0\le t\le T$ realizes $V$, then Hamiltonian  $H(t)-\frac{\alpha}{T} I: 0\le t\le T$ realizes $e^{i\alpha} V$. For traceless operator $C$, $\Tr(H(t) C)$ remains invariant under this transformation. Then,  the second equality in Eq.(\ref{ary2}) implies that $\Phi$ remains invariant.

Finally, note that Eq.(\ref{ary2}) can be rewritten as 
\be\label{lasteq}
\Phi=-\int_0^T dt\  \langle \zeta | \chi_{H(t)}\rangle\ \ \ \ :\text{mod }2\pi\ ,  
\ee
where $| \chi_{H(t)}\rangle=\sum_\mu \Tr(H(t)\Pi_\mu)|\mu\rangle$ is the charge vector of $H(t)$ and $|\zeta\rangle=\sum_\mu c(\mu) |\mu\rangle$. Therefore, by measuring $\Phi$,  we can obtain information about the integral of the charge vector $|\chi_{H(t)}\rangle$.

In summary, for any traceless operator $C=\sum_\mu c(\mu)  \Pi_\mu$ with integer eigenvalues $\{c(\mu)\}$, the phase $\Phi$ defined in Eq.(\ref{ary2}) is an  observable quantity.  According to Eq.(\ref{lasteq}), by measuring this phase we can obtain information about the charge vector of the Hamiltonian that realizes unitary $V$, which in turn contains information about the locality of this Hamiltonian.     
 Next, we focus on the example of U(1) symmetry. \\

\subsection{$l$-body phases for U(1)-invariant  unitaries}

Next, we consider the special case of  U(1) symmetry with qubit systems. In this case,  choosing operator $C$ in  Eq.(\ref{ary2}) to be the operator $C_l=\sum_{m} c_l(m) \Pi_m$ defined in Eq.(\ref{defCl}), we obtain the notion of $l$-body phase,  
\be\label{ary4}
 \Phi_l\equiv    \sum_{m=0}^nc_l(m) \theta_m=-\hspace{-1mm}\int_0^T\hspace{-2mm} dt  \hspace{-2mm}\sum_{\substack{{\bf{b}}: w({\bf{b}})=l}} \hspace{-
 2mm} \Tr(H(t) {\bf{Z}}^{\bf{b}})\ \ \  \ \text{: mod} \ 2\pi\ ,
\ee
defined in Eq.(4) of the main paper. Recall that here  $\theta_m=\text{arg}(\text{det}(V_m))\in(-\pi,\pi]$ is the phase of the
 determinant of $V_m$, and $c_l(m)=\sum_{s=0}^m (-1)^s  {{m}\choose{s}} {{n-m}\choose{l-s}}$. 
   Because for $l\ge 1$, operator $C_l$ is traceless,  the $l$-body phase $\Phi_l$ is physically observable for $l\ge 1$.  Note that $\{c_l(m)\}$, which are eigenvalues of $C_l$, are all  integer, which is crucial for the validity of Eq.(\ref{ary2}).      
 The table below shows example of coefficients $c_l(m)$ for a system with $n=5$ qubits.

 \begin{table}[htp]
\begin{center}
\begin{tabular}{|c|c|c|c|c|c|c|}
\cline{2-7}
\multicolumn{1}{c|}{}& \multicolumn{6}{|c|}{Charge Sector} 
\\ \hline
   &$m=0$\ & \ $m=1$\ &\  $m=2$\ &\  $m=3$\ &\  $m=4$\ &\  $m=5$\ \  \\
\hline
$l=0$ \text{body}  &1 & 1 & 1& 1& 1& 1\\
\hline
$l=1$ \text{body}  &5 & 3 & 1& -1& -3& -5\\
\hline
$l=2$ \text{body}   &10 & 2 & -2& -2& 2& 10\\
\hline
$l=3$ \text{body}   &10 & -2 & -2& 2& 2& -10\\ \hline
$l=4$ \text{body}&5 & -3& 1& 1& -3 & 5\\ \hline
$l=5$ \text{body}&1 & -1& 1& -1& 1& -1\\
\hline
\end{tabular}
\end{center}
\caption{Coefficients $c_l(m)$, i.e., the eigenvalues of operators $C_l$, for $n=5$ qubits.}
\label{default}
\end{table}%

As an example, consider the unitary $\exp(i\alpha \textbf{Z}^{\textbf{b}})$. Then,  using the identity $\Tr(C_l \textbf{Z}^{\textbf{b}})= 2^n \delta_{l, w(\textbf{b})}$, and applying the second equality in  Eq.(\ref{ary4}), we see that the $l$-body phase of this unitary  is  equal to $\Phi_l=2^n \alpha\times \delta_{l, w(\textbf{b})}$.  In particular, for the unitary $e^{i\alpha} I$,  all the $l$-body phases are zero, except $l=0$.


A useful property of $l$-body phases, which follows immediately from the definition  in Eq.(\ref{ary4}) is their additivity:  The $l$-body phase of $V_2V_1$ is the sum of the $l$-body phases of $V_1$ and $V_2$, mod $2\pi$. Furthermore, the $l$-body phase of $V^\dag$ is equal to minus $l$-body phase of $V$.    More abstractly, we can think of Eq.(\ref{ary4}) as a homomorphism from the group  of U(1)-invariant unitaries on $n$ qubits to the group $\text{U}(1)^{n+1}$.





\subsection{A Characterization of U(1)-invariant unitaries on qubit systems in terms of $l$-body phases}

 Next, we present a general characterization of U(1)-invariant unitaries in terms of their $l$-body phases (See theorem \ref{Thm:body}). First, 
 recall that, as discussed in corollary \ref{cor81},   any U(1)-invariant unitary     
 $W=\bigoplus_{m=0}^n W_m$ that satisfies the  constraint $\text{det}(W_m)=1: m=0,\cdots, n$ is contained in the group 
 \be\label{def:G2}
G_2\equiv \big\langle e^{i\theta (X_rX_s+Y_rY_s)}, e^{i\theta Z_r}:\theta\in[0,2\pi),r\neq s\in\{1... n\}\big\rangle\ .
\ee 
 For a general  U(1)-invariant unitary $V$ consider  the unitary 
 \be\label{Eq:U}
 U=\exp(-i\sum_{m=0}^n  \frac{\theta_m}{\Tr(\Pi_m)} \Pi_m)\ ,
 \ee 
where
\be\label{Eq:theta}
\theta_m=\text{arg}(\text{det}(V_m))=-\int_{t=0}^T\hspace{-2mm} dt\  \Tr(\Pi_m H(t))\ \ :\ (\text{mod} 2\pi)\ ,
\ee
Hermitian operator $H(t)$ is any $U(1)$-invariant Hamiltonian that realizes $V$, such that  $V=\mathcal{T}\big\{\exp\big( {-i \int_{0}^T H(t) dt}\big) \big\}$, and the second equality in Eq.(\ref{Eq:theta}) follows from Eq.(\ref{ary2}).   Then, for the unitary $V_2=UV$ the determinant in each charge sector is one, which means  $V_2\in G_2$.  Note that each $\theta_m$ is only defined mod $2\pi$, and in the definition of $U$ in Eq.(\ref{Eq:U}), we can use any set of $\{\theta_m\}$ that satisfies Eq.(\ref{Eq:theta}). In the following, we determine $\{\theta_m\}$ in terms of $l$-body phases  of unitary $V$.
For $l=0,\cdots, n$, define 
 \be\label{Def:D_l}
  D_l\equiv\frac{C_l}{\Tr(C^2_l)} =\frac{1}{2^n {{n}\choose{l}}}  \sum_{\substack{{\bf{b}}:  w({\bf{b}})=l}\ }  {\bf{Z}}^{\bf{b}}=  \sum_{m=0}^n \frac{c_l(m)}{2^n {{n}\choose{l}}}\ \Pi_m \ .
  \ee
 Note that the three sets of operators $\{C_l\}$, $\{D_l\}$ and $\{\Pi_m\}$ all span the same $(n+1)$-dimensional space and $ \Tr(C_{l} D_{l'})=\delta_{l,l'}$, which  implies 
 \be
 \Pi_m=\sum_{l=0}^n\ \Tr(D_l\Pi_m) C_l\ .
 \ee
  Putting this in Eq.(\ref{Eq:theta}), we find 
\bes
\begin{align}
\theta_m&=\sum_{l=0}^n\ \Tr(D_l\Pi_m) \times  \Big[-\int_{t=0}^T\hspace{-2mm} dt\   \Tr(C_l H(t))\Big]= \sum_{l=0}^n\ \Tr(D_l\Pi_m) \times  \big[\Phi_l+2\pi r_l\big]\  \ \ :\ (\text{mod} 2\pi)\ ,
\label{transform54}
\end{align}
\ees
for an unspecified  set of integers $\{r_l\}$. Here, to get the second equality we have applied Eq.(\ref{ary4}) (the unknown integers $\{r_l\}$ appear because Eq.(\ref{ary4})  holds mod $2\pi$).  Putting this into Eq.(\ref{Eq:U}), and using the fact that $D_l=\sum_m \Tr(\Pi_m D_l)\Pi_m/\Tr(\Pi_m)$, we find   
\be
U=\prod_{l=0}^n \exp(i(\Phi_l+2\pi r_l)D_l)\ .
\ee
 Note that $\prod_{l=0}^n \exp\big(i 2\pi r_l D_l \big)$ is an element of the group
\be\label{G0}
G_0\equiv \big\langle \exp(i 2\pi D_l ): l=0, \cdots, n \big\rangle\ ,
\ee
generated by unitaries $\exp(i 2\pi D_l )$ for $l=0,\cdots , n$.  Because operators $\{D_l\}$ commute with each other and their eigenvalues are  rational numbers, this group is finite.     Putting everything together, we arrive at  
\begin{theorem}\label{Thm:body}
Any U(1)-invariant unitary transformation $V$  on $n$ qubits has a decomposition    as 
\be\label{Eq:671}
V=  V_0 \big[\prod_{l=0}^n \exp\big(i\Phi_l D_l \big)\big] V_2 \ ,
\ee
where $\Phi_l$ is the $l$-body phase of $V$ defined in Eq.(\ref{ary4}), $D_l$ is  the Hermitian operator defined in Eq.(\ref{Def:D_l}),  $V_0$ is a unitary in the finite group $G_0$ defined in Eq.(\ref{G0}),  and $V_2$ is in the group $G_2$ defined in Eq.(\ref{def:G2}), i.e., can be realized using Hamiltonians $XX+YY$ and $Z$. 
\end{theorem}
According to Eq.(\ref{Def:D_l}), operator $D_l$ can be written as a linear combination of commuting $l$-local  operators and therefore for any $\theta$ the unitary $\exp(i\theta D_l)$ can be implemented  using $l$-local U(1)-invariant unitaries.  We conclude that  if for unitary V, $l$-body phase $\Phi_l=0$ for all $l>k\ge 2$, then $V$ is realizable using $k$-local U(1)-invariant unitaries, up to a unitary $V_0$ in the fixed finite group $G_0$ defined in Eq.(\ref{G0}).  It is also worth noting that in decomposition in Eq.(\ref{Eq:671}), we can replace each term $\exp\big(i\Phi_l D_l \big)$ with $\exp(\Phi_l Z^{\otimes l}\otimes I^{\otimes (n-l)}/2^n)$. That is, there exist $V'_0\in G_0$ and $V'_2\in G_2$ such that
  \be
 V= V'_0 \big[\prod_{l=0}^n \exp\big(i\frac{\Phi_l}{2^n} Z^{\otimes l}\otimes I^{\otimes (n-l)}\big)\big]V'_2   \ .
 \ee
To see this note that by multiplying $V$ in $\prod_{l=0}^n \exp\big(-i\frac{\Phi_l}{2^n} Z^{\otimes l}\otimes I^{\otimes (n-l)}\big)$, we obtain a unitary whose $l$-body phases vanish and, therefore by the above theorem, it can be written as $V'_2V'_0$ for $V'_0\in G_0$ and $V'_2\in G_2$.


To understand the appearance of the finite group $G_0$ in theorem \ref{Thm:body} and the unspecified integers $\{r_l\}$ in Eq.(\ref{transform54}), it is useful to recall the geometric interpretation of the $l$-body phases $\{\Phi_l\}$, as a coordinate system for the $(n+1)$-torus defined by phases $\{\theta_m\}$. Definition  $\Phi_l=\sum_m c_l(m) \theta_m$ together with Eq.(\ref{transform54}), allow us to go from one coordinate system to the other. However, because the coordinate system defined by $\{\Phi_l\}$ is degenerate, this relation is not 1-to-1. The finite group $G_0$ describes all possible points on the $(n+1)-$torus that have the same coordinates relative to  $\{\Phi_l\}$.
     
\subsection{U(1)-invariant Projective Measurements}

Finally, we present a useful corollary of the above theorem, namely the fact that, in the case of group U(1), locality does not restrict realizable rank-1 projective measurements, i.e.,
 
 \begin{corollary}\label{cor:Meas}
Let $\{|\phi_v\rangle: v=1,\cdots, 2^n\}$ be an orthonormal basis for the Hilbert space of $n$ qubits, with the property that each element of the basis  is an eigenvector of $\sum_{j=1}^n Z_j$, i.e., is invariant under U(1) symmetry. Then, the projective measurement in this basis can be realized by performing a sequence of 2-local unitaries $\exp(i\theta (XX+YY))$ and single-qubit unitary $\exp(i\theta Z)$, for $\theta\in[0,2\pi)$, followed by the measurement of all qubits in the z basis.  
\end{corollary}
To see this, first note that any pair of orthonormal bases satisfying the property described in the corollary can be converted to each other by a U(1)-invariant unitary. In particular, because the computational basis satisfies this property, there is a U(1)-invariant unitary $V$, such that $V|\phi_v\rangle=|\textbf{b}(v)\rangle$, where 
$|\textbf{b}(v)\rangle$ is an element of the computational basis.   
 Therefore, if we first perform unitary $V$ and then measure all qubits in $\{|0\rangle,|1\rangle\}$ basis,  the overall effect is equivalent to measuring  in 
$\{|\phi_v\rangle\}$ basis. Now suppose instead of unitary $V$, we perform unitary $V_2=UV $, for $U$ defined in Eq.(\ref{Eq:U}). Then,  
because $U$ is diagonal in the computational basis,  and the final measurement is also performed in this basis,  
the probability of outcomes    
 do not change. Since $V_2\in G_2$ it can be realized using Hamiltonians $XX+YY$ and local $Z$. This proves the corollary.

\color{black}

\newpage

 \section{Supplementary Note 5: Local symmetric process tomography}\label{Sec:Meth:Meas}

In this section we present a scheme for characterizing an unknown U(1)-invariant unitary and measuring its $l$-body phases. The main feature of this scheme, which makes it different from the standard process tomography methods, is the fact that it only requires initial states, 2-local unitaries, and single-qubit measurements that all respect the symmetry.  

Consider an unknown U(1)-invariant unitary $V$ on $n$ qubits labeled as $j=1,\cdots, n$. As depicted in Fig.\ref{Fig:gha:App}, in this scheme  all the $n$  qubits   are initially prepared in states $|z_j\rangle$, where $z_j\in\{0,1\}$,
except one qubit labeled as $r\in\{1,\cdots, n\}$. Qubit $r$, on the other hand, is entangled with an ancillary qubit, in the joint state $(|0\rangle |1\rangle+|1\rangle |0\rangle)/\sqrt{2}$. Then, the joint initial state of  $n$ qubits and the ancilla can be written as 
\be
\frac{|\textbf{z}\rangle\otimes |0\rangle+|\textbf{z}_-\rangle \otimes |1\rangle}{\sqrt{2}}=\frac{|\textbf{z}\rangle\otimes |0\rangle+X_r|\textbf{z}\rangle \otimes |1\rangle}{\sqrt{2}}\ ,
\ee
where we use the convention that $z_r=1$, i.e., 
\be
|\textbf{z}\rangle=|z_1\cdots, z_{r-1}, 1 , z_{r+1},\cdots, z_n\rangle\ .
\ee
After applying the unitary $V$, this state transforms to 
\be\label{state-befor}
\frac{V|\textbf{z}\rangle\otimes |0\rangle+V|\textbf{z}_-\rangle \otimes |1\rangle}{\sqrt{2}}\ .
\ee
Now suppose we perform a measurement in an orthonormal basis that includes the state 
\be\label{basis-meas}
\frac{|\textbf{z}'\rangle\otimes |0\rangle+ e^{i\theta} |\textbf{z}'_- \rangle\otimes |1\rangle}{\sqrt{2}}\ ,
\ee
where $\textbf{z}', \textbf{z}_-'\in\{0,1\}^n$ and they satisfy the property that $w(\textbf{z})= w(\textbf{z}')=w(\textbf{z}'_-)+1$, where $w(\textbf{z}')=\sum_{j=1}^n z'_j$ is the Hamming weight of bit string $\textbf{z}'$.
  This condition means that state in Eq.(\ref{basis-meas}) is restricted to a sector with a definite total charge, which is equal to the total charge of state in Eq.(\ref{state-befor}). Therefore state in Eq.(\ref{basis-meas}) can be an element of an orthonormal basis satisfying the condition in corollary \ref{cor:Meas}. Then, the corollary implies that  the  measurement in this basis can be performed  by applying  Hamiltonians $XX+YY$ and local $Z$ on qubits, and then measuring them in the computational basis.  

Performing this measurement on state in Eq.(\ref{state-befor}), we get the outcome corresponding to state in Eq.(\ref{basis-meas}) with probability
\be
\frac{1}{4}\Big|  \langle\textbf{z}'|V|\textbf{z}\rangle+e^{-i\theta} \langle\textbf{z}'_-|V|\textbf{z}_-\rangle \Big|^2\ .
\ee
Performing this measurement sufficiently many times for different values of $\theta\in(\pi,\pi]$, we can estimate the cross term
\be\label{cross}
\langle\textbf{z}'_-|V|\textbf{z}_-\rangle^\ast \times  \langle\textbf{z}'|V|\textbf{z}\rangle\ .
\ee
Note that $|\textbf{z}'_-\rangle$  and $|\textbf{z}_-\rangle$ both live in the same charge sector. If we know the value of the matrix element $\langle\textbf{z}'_-|V|\textbf{z}_-\rangle$, and if this value is non-zero, then by estimating the quantity in Eq.(\ref{cross}), we can infer the value of $ \langle\textbf{z}'|V|\textbf{z}\rangle$.  Note that U(1)-invariance together with the unitarity of $V$ guarantees that 
for any $|\textbf{z}_-\rangle$, there is a state   $|\textbf{z}'_-\rangle$ in the computational basis  with the property that  $\langle\textbf{z}'_-|V|\textbf{z}_-\rangle\neq 0$.  Therefore, applying this technique recursively, we can determine $V$ in any arbitrary sector. The sectors with the Hamming weights  $0$ and $n$ are 1-D. Then, to remove the global phase freedom in characterizing unitary $V$, we can choose $\langle0|^{\otimes n}V|0\rangle^{\otimes n}=1$. 


\color{black}

\subsection{Characterizing 3-qubit U(1)-invariant unitaries}\label{App:meas}

  \begin{figure}{\includegraphics[scale=.31]{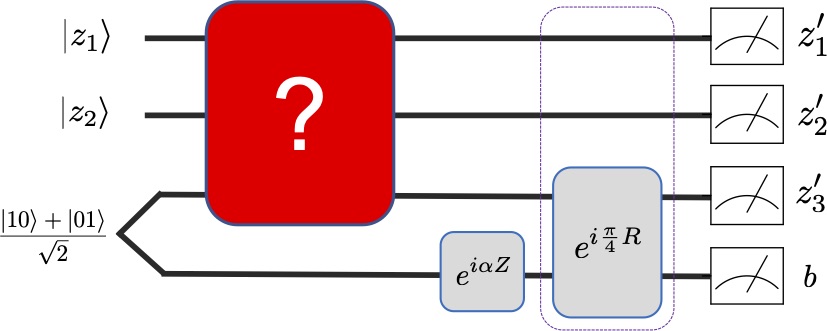}}\caption{\textbf{A scheme for characterizing an unknown U(1)-invariant unitary .}  
 }\label{Fig:gha:App}
\end{figure}

In the following, we demonstrate this scheme for the example of $n=3$ qubits. 
In this special case, which is also considered in Fig.(\ref{Fig:gha:App}), the unitary performed before the single-qubit measurements is the single-qibit unitary $\exp(i\alpha Z_\text{anc})$ on the ancilla, followed by a single two-qubit unitary $\exp(\frac{i\pi}{8} (X_{j}X_{\text{anc}}+Y_{j}Y_{\text{anc}}))$ for $j=1, 2, 3$.


Let $V$ be an  unknown U(1)-invariant unitary on 3 qubits. The symmetry implies that $|\langle 0|^{\otimes 3}V|0\rangle^{\otimes 3}|=1$, and to fix the global phase of unitary, we can choose 
\be\label{sectorzero}
\langle 000|V|000\rangle=1\ .
\ee
At the input of the above circuit, qubit $r\in\{1,2,3\}$ is entangled with the ancillary qubit, in the joint state $(|01\rangle+|10\rangle)/\sqrt{2}$. 
The other two qubits are  prepared in states $|0\rangle$ and $|1\rangle$. Therefore, we represent the joint state at the input as 
\be
\frac{|\textbf{z}\rangle\otimes |0\rangle_\text{anc}+X_r|\textbf{z}\rangle \otimes |1\rangle_\text{anc}}{\sqrt{2}}\  : \ \ z_r=1\ ,
\ee
where $\textbf{z}=z_1z_2 z_3$, and we use the convention that $z_r=1$ (In the example in Fig.\ref{Fig:gha:App}, we have $r=3$). After applying unitary $V$, we obtain state 
\be\label{stateV}
\frac{V|\textbf{z}\rangle\otimes |0\rangle_\text{anc}+VX_r|\textbf{z}\rangle \otimes |1\rangle_\text{anc}}{\sqrt{2}}\  
: \ \ z_r=1
\ee
Then, we apply the unitary $\exp({i\alpha Z_\text{anc}})$ on the ancillary qubit,  the unitary
 $\exp(\frac{i\pi}{4} R_{s,\text{anc}})$ 
 on qubit $s\in\{1, 2, 3\}$ and the ancillary qubit, and finally measure all qubits in the computational basis, and obtain the corresponding outcomes 
 $\textbf{z}'=z'_1z'_2z'_3\in\{0,1\}^3$, and 
 $ b\in\{0,1\}$.  
The probability of these outcomes  are determined by the overlap of state in Eq.(\ref{stateV}) with state 
 \begin{align}\label{meas09}
\exp({-i{\alpha} Z_\text{anc}})\exp(-\frac{i\pi}{4} R_{s,\text{anc}}) \big(|z'_1z'_2 z'_3\rangle\otimes |b\rangle_{\text{anc}}\big)\ \ \ , \ s\in\{1,2,3\}.
\end{align}
The two cases where $b=0$ and $z'_s=1$, or  $b=1$ and $z'_s=0$, 
are of special interest, because in these cases the probability of outcomes depend non-trivially on the matrix elements of $V$ in two different sectors (In other words, we can observe interference between two branches of wavefunction that goes through two different sectors). 

To simplify the analysis and notation, in the following we focus on the case of  $z'_s=1$ and $b=0$. In this case state in Eq.(\ref{meas09}) becomes  
\begin{align}
\exp({-i{\alpha} Z_\text{anc}})\exp(-\frac{i\pi}{4} R_{s,\text{anc}}) \big(|z'_1z'_2 z'_3\rangle\otimes |0\rangle_{\text{anc}}\big)=\frac{ e^{-i\alpha} |\textbf{z}'\rangle\otimes |0\rangle_{\text{anc}} - i e^{i\alpha} X_s |\textbf{z}'\rangle\otimes |1\rangle_{\text{anc}}}{\sqrt{2}}\  \ \ :\ z'_s=1\ .
\end{align}
Therefore, if $z'_s=1$ and $b=0$, the probability of outcome  $\textbf{z}'=z'_1z'_2z'_3$ is  
\be
\frac{1}{4}\Big|\langle\textbf{z}'|V|\textbf{z}\rangle+i e^{-i 2\alpha} \langle\textbf{z}'|X_s V X_r |\textbf{z}\rangle \Big|^2\ \ \ \ \  \ :\ z'_s, z_r=1 \ .
\ee 
By measuring this quantity for different values of $\alpha$, we can determine \be\label{est91}
\langle\textbf{z}|X_s V X_r |\textbf{z}\rangle^\ast \times \langle\textbf{z}'| V |\textbf{z}\rangle\ \ \ : \ z'_s, z_r=1 \ .
\ee 
Then, if we also know the value of $\langle\textbf{z}'|X_s V X_r |\textbf{z}\rangle$ and if this value is non-zero,  we can infer the matrix element $\langle\textbf{z}'| V |\textbf{z}\rangle$. As we see below, by applying this technique recursively, we can characterize the action of $V$ in all sectors. 

First, consider the sector corresponding to Hamming weight $m=1$, spanned by three states $\{001,010, 100\}$. Using the convention in Eq.(\ref{sectorzero}) for this sector, the coefficient  $\langle\textbf{z}'|X_s V X_r |\textbf{z}\rangle^\ast =\langle000|V|000\rangle=1$, and therefore, by estimating the quantity in Eq.(\ref{est91}), we can determine all matrix elements $\langle\textbf{z}'| V |\textbf{z}\rangle$ for $\textbf{z}', \textbf{z}\in \{001,010, 100\}$.

Next, in the sector with $m=2$, the outcome probabilities determine the value of 
\be
\langle\textbf{x}|X_s V  X_r|\textbf{z}\rangle^\ast \times \langle\textbf{x}| V |\textbf{z}\rangle\ \ \ : \ z'_s, z_r=1 \ .
\ee 
for $\textbf{z}' ,\textbf{z}\in \{011,101,110\}$. For instance, in the case of $\textbf{z}'=\textbf{z}=011$, we can choose $r, s\in\{2,3\}$. Then, by estimating the quantity in Eq.(\ref{est91}), we find the value of  $c^\ast \times  \langle{011}| V |{011}\rangle$, where  
\begin{align}\label{cccy}
c=(\langle{0}|\otimes\langle\textbf{y}|) V (|0\rangle\otimes |{\textbf{y}'}\rangle) \ \ \ :\  \textbf{y}, \textbf{y}'\in\{01,10\} \ .
\end{align}
These matrix elements are determined in the previous step, when we characterized the sector with $m=1$. Therefore, if any of these four matrix elements are non-zero, then we can determine  the matrix element $  \langle{011}| V |{011}\rangle$. We claim that, at least, one of the matrix elements in Eq.(\ref{cccy}) is non-zero . This follows from U(1) symmetry together with the unitarity of $V$: The sector with $m=1$ is 3 dimensional and is spanned  by $\{|100\rangle,|010\rangle, |001\rangle\}$ and unitary $V$ acts as a $3\times 3$ unitary in this subspace. If all 4 matrix elements in Eq.(\ref{cccy}) vanish, then this $3\times 3$  matrix   has a $2\times 2$ sub-matrix which is equal to zero. But, this contradicts unitarity (If a $3\times 3$ matrix has a zero $2\times 2$ submatrix, then at, least, two different columns of the matrix are linearly dependent, which contradicts with  unitarity). We conclude that, at least, one of the four matrix elements in Eq.(\ref{cccy}) are non-zero. This allows us to infer the value of $ \langle{011}| V |{011}\rangle$ , using the estimated value of Eq.(\ref{est91}). Similarly, we can determine $\langle\textbf{z}'| V |\textbf{z}\rangle$ for any pairs of  $\textbf{z}' ,\textbf{z}\in \{011,101,110\}$.

Finally, for $\textbf{z}'=\textbf{z}=111$, we can estimate
\be
\langle 111|X_s V X_r|111\rangle^\ast \times \langle 111|V|111\rangle\ ,
\ee 
for $r,s\in\{1,2,3\}$. Again, because $V$ acts unitarily in the $3$-dimensional subspace $\{011,101,110\}$, coefficients $
\langle 111|X_s V X_r|111\rangle$ cannot all be zero, and therefore in this way we can determine  $ \langle 111|V|111\rangle$. Hence, we can completely characterize a U(1)-invariant unitary $V$.

\newpage


 \section{Supplementary Note 6: Circumventing the no-go theorem--An intuitive explanation via Jordan-Winger Transform}
\vspace{-3mm}


As we discussed in the paper and formally prove in the Supplementary Note 7,  in the case of group U(1)  the no-go theorem can be circumvented using ancillary qubits. Here, we further discuss the example presented in Fig. 4 
of the main paper, and present a nice interpretation of this scheme.


Recall that in Fig. 4 of the paper we consider a system with $n$ qubits together with a pair of ancillary qubits a and b that form a closed loop.   
Considering the nested commutators 
of $R=(XX+YY)/2$ on pairs of neighboring qubits, we find
\begin{align}
\hspace{-2mm}K\equiv \hspace{-1mm}\pm \Big[R_{\text{b},{n}}\ , \cdots\hspace{-1mm}  \big[R_{3,2}\ ,\  [R_{2,1}, R_{1,\text{a}}]\big]\cdots\hspace{-1mm}\Big]=A_{\text{a}, \text{b}}\otimes Z^{\otimes n}\ ,   
 \end{align}
where  $A_{\text{a}, \text{b}}= R_{\text{a}, \text{b}}$ for even $n$  and  $A_{\text{a}, \text{b}}=T_{\text{a}, \text{b}}\equiv X_\text{a}Y_\text{b}-Y_\text{a}X_\text{b}$ for odd $n$. The above nested commutator means that by turning on and off interactions $XX+YY$ between all nearest-neighbor qubits in the loop, except between a and b, we can realize Hamiltonian $K=A_{\text{a}, \text{b}}\otimes Z^{\otimes n}$. In the following, to simplify the notation  we assume $n$ is even, which means $K=R_{\text{a}, \text{b}}\otimes Z^{\otimes n}$  (The case of odd $n$ is similar). 

 Now suppose the initial joint state of the $n$ qubits in the system and the ancillary qubits a and b is $|\psi\rangle|1\rangle_\text{a}  |0\rangle_\text{b} $. Then, applying Hamiltonian $K$ for a sufficiently short time interval $\delta t$, we obtain state
\be
e^{-i K \delta t} (|\psi\rangle |1\rangle_\text{a} |0\rangle_\text{b})=|\psi\rangle |1\rangle_\text{a}|0\rangle_\text{b} -  i \delta t  Z^{\otimes n}|\psi\rangle |0\rangle_\text{a}  |1\rangle_\text{b}\ ,
\ee
where we have neglected terms of order $\delta t^2$ and higher. 
Roughly speaking, the second term in the right-hand side corresponds to the event  in which  the charge which was initially located at qubit a moves to qubit b through the chain, and during this process obtains a phase $\pm 1$, depending on the value of $Z^{\otimes n}$.  Next,  by applying the unitary $\exp{(i\pi R_{\text{a}, \text{b}}/4)}\exp{(i\pi Z_\text{b}/4)}$ on a and b
we close the loop and (up to a global phase) obtain state
\begin{align}
&\frac{1}{\sqrt{2}}e^{ -iZ^{\otimes n}\delta t}|\psi\rangle|1\rangle_\text{a} |0\rangle_\text{b}+\frac{i}{\sqrt{2}}e^{+iZ^{\otimes n}\delta t}|\psi\rangle|0\rangle_\text{a} |1\rangle_\text{b}\ ,
\end{align}   
where again we have ignored terms of order $\delta t^2$ and higher.    Unitary $\exp{(i\pi R_{\text{a}, \text{b}}/4)}\exp{(i\pi Z_\text{b}/4)}$  allows the charge to go back and forth  directly between a and b.  
 Finally, by measuring one of the ancillary qubits in $\{|0\rangle, |1\rangle\}$ basis we can determine  the final location of the charge. Depending on the outcome, the state of the main system collapses to one of states $\exp({\pm i \delta t\ Z^{\otimes n}})|\psi\rangle+\mathcal{O}(\delta t^2)$. Repeating these steps we can approximately realize the unitary $\exp({i\theta Z^{\otimes n})}$, where the value of $\theta$ is determined by  a random walk.  Hence, in principle, by repeating these steps with sufficiently small $\delta t$, we can realize any desired unitary $\exp(i \phi  Z^{\otimes n})$ with $\phi\in(-\pi,\pi]$,  with arbitrary accuracy and probability of success approaching one. As we further discuss in Supplementary Note 7, it is indeed  possible to implement this scheme deterministically, without using measurements.  

Finally, we discuss an interesting interpretation of this scheme, based on the fermionic description of this system, obtained via the Jordan-Wigner transform  \cite{jordan1928pauli, fradkin1989jordan, nielsen2005fermionic}. Suppose we  label  qubits as $j=0,\cdots, n+1$, where $j=0 $ and $j=n+1$, corresponds to the ancillary qubits a and b, respectively. Then, this correspondence can be defined by 
\be\label{Jordan}
\frac{X_j+i Y_j}{2}\ \longleftrightarrow\  c_j^\dag\ \prod_{l=j+1}^{n+1}\hspace{-1mm}  (-1)^{c_l^\dag c_l } \ \ :j=0,\cdots, n+1\ ,\\
\ee
where $c_j^\dag: j=0,\cdots , n+1$ are the fermionic creation operators, satisfying the standard anti-commutation relations $\{c^\dag_j,c_k\}=\delta_{j,k}$ and  $\{c^\dag_j,c^\dag_k\}=0$. This, in particular, implies $Z_j=2c_j^\dag c_j-I=-(-1)^{c_j^\dag c_j}$,  where $I$ is the identity operator. This means that  the charge of a site in the qubit picture is determined by the occupation number of the corresponding site in the fermionic picture. Furthermore, this   implies the correspondence
\begin{align}
R_{\text{a,b}}\ &\longleftrightarrow \ \ (c_\text{a}^\dag c_\text{b}+c_\text{b}^\dag c_\text{a})   (-1)^{\sum_{l=1}^n  c_l^\dag c_l } \ ,\nonumber\\ 
K=R_{\text{a,b}}\otimes Z^{\otimes n} \ &\longleftrightarrow \ \    (c_\text{a}^\dag c_\text{b}+c_\text{b}^\dag c_\text{a}) \ \label{Jordan2} ,
\end{align}
where we have assumed $n$ is even. 
Note that $(-1)^{\sum_{l=1}^n c_l^\dag c_l }$ is $\pm 1$ depending on the parity of the total number of  ``particles"  in sites 1 to $n$. From this point of view, the difference  between $R_{\text{a,b}}$ and $K=R_{\text{a,b}}\otimes Z^{\otimes n}$, can be understood as a consequence of the anti-symmetry of the fermionic wavefunction: As we move a fermion from a to b through the chain, the state obtains $\pm 1$ sign depending on whether each site is occupied or not.  Hence, the overall phase depends on the total number of fermions in the chain (In the qubit picture this corresponds to the  observable $Z^{\otimes n}$).  On the other hand, if the particle moves  directly between a and b,  it does not obtain this phase.

From this point of view, the above scheme essentially uses interference between two branches of the wavefunction,  to determine the parity of the total charge in the system; one branch is going directly between a and b and the other is going through the chain of $n$ qubits. Then, the ancillary qubits can be interpreted as an internal  quantum reference frame that allows us to measure this relative phase.   

It is worth noting that the Jordan-Wigner transform defined in Eq.(\ref{Jordan}) is not unique. For instance,  if we periodically shift the labels by one, i.e., label qubit b as $j=0$,  qubit a as $j=1$, etc,  then the fermioinc operators in the top and bottom lines of Eq.(\ref{Jordan2}) will be swapped (This freedom can be formulated as a gauge potential \cite{fradkin1989jordan}). However, in both cases the operator $Z^{\otimes n}$ corresponds to  $(-1)^{\sum_{j=1}^n c^\dag_j c_j} $ .

\newpage

\color{black}

\section{Supplementary Note 7: Schemes for implementing general energy-conserving unitaries on composite systems}\label{App:C}

\subsection{Overview and an illustrative example}

In this section we focus on the case of U(1)-invariant unitaries, or, equivalently, energy-conserving unitaries and 
study various techniques and constructions for circumventing the no-go theorem with 
ancillary qubits. In particular, we prove the following theorem in the paper.

\begin{theorem}\label{Thm:energy} (Informal version) Consider a finite set of closed systems  with the property that for each system the  gap between any consecutive pairs of energy levels is $\Delta E$. Then, any global energy-conserving unitary transformation on these systems can be implemented by a finite sequence of  2-local energy-conserving unitaries, provided that the systems can interact with a single ancillary qubit  with the energy gap $\Delta E$ between its two levels.
\end{theorem}

 We also show that if one can use a second ancillary qubit, then a general energy-conserving unitary can be implemented  without any direct interactions between the systems, i.e., just using system-ancilla interactions.  Theorem  \ref{Thm:iejrg} and corollary \ref{erh} contain the precise statements of these results.\\

Before going into the details, here  we explain the main idea with an illustrative example for a system with 3 qubits.  Suppose the goal is to implement the family of unitaries $\{e^{i\theta Z_1Z_2Z_3}: \theta\in[0,2\pi)\}$ on a system with $n=3$ qubits, labeled as 1, 2, 3. The condition in Eq.(\ref{App:conds2}) implies that this family cannot be generated by 2-local U(1)-invariant unitaries.  Now suppose in addition to these 3 qubits, we can use an ancillary qubit labeled as $a$. As we explain in Fig.\ref{loop}, the commutation relations in Eq.(\ref{kjf00}) imply that by applying the unitaries generated by $XX+YY$ and local $Z$ Hamiltonians, which are both 2-local and invariant under rotations around z,  one can simulate Hamiltonian $(Z_3-Z_a) Z_1Z_2$, where $Z_a$ is Pauli $Z$ on the ancillary qubit $a$ tensor product the identity operators on the rest of qubits.  Assuming the ancillary qubit is initially prepared  in state $|0\rangle_a$ and qubits 1, 2, 3 are in an arbitrary state $|\psi\rangle$, under the time evolution generated by this Hamiltonian, the initial state $|\psi\rangle\otimes |0\rangle_a$ evolves to  
\be\label{Eq:loop}
e^{i \theta (Z_3-Z_a) Z_1Z_2 }\ (|\psi\rangle_{}\otimes |0\rangle_a)= \big(e^{i \theta(Z_1Z_2Z_3-Z_1Z_2)} |\psi\rangle\big)\otimes |0\rangle_a\ ,
\ee
for $\theta\in[0,2\pi)$.  Therefore, if after applying this unitary we apply the 2-local symmetric unitary $e^{i\theta Z_1Z_2}$ on qubits 1 and 2, the overall unitary evolution of qubits 1, 2, 3 will be the desired unitary  $e^{i \theta Z_1Z_2Z_3}$. Note that because at the end of the process the ancillary qubit goes back to its initial state, we can use it again to implement other unitary transformations. As we discuss in Fig.\ref{loop} and prove  in the following section, the commutation relations in  Eq.(\ref{kjf00}) imply that  this technique can be generalized to implement all diagonal Hamiltonians, just using interactions $XX+YY$, and local $Z$ on the ancillary qubit.

Finally, recall that combining diagonal unitaries with unitaries generated by 2-local interaction $XX+YY$, one obtains all U(1)-invariant unitaries (See theorem \ref{Thm5}). In summary, we conclude that: \emph{all unitaries that are invariant under rotations around z, i.e., those preserving  $\sum_j Z_j$,  can be implemented using a single ancillary qubit and via interactions $XX+YY$ and local $Z$ on the ancillary qubit}. 

We also show that this result remains valid if  there are further geometric constraints on the interactions between qubits. In particular, if the qubits in the system form a chain and only  nearest-neighbor $XX+YY$  interactions between them are allowed, we can still implement general U(1)-invariant unitaries,  provided that the ancillary qubit  can interact with all the qubits via $XX+YY$ interaction. Alternatively, if  in addition to $XX+YY$ interaction, one can also apply $ZZ$ interactions to the nearest-neighbor qubits, then the ancillary qubit only needs to interact with one qubit in the chain, e.g., the qubit at one end of the chain.

  \begin{figure}
{\includegraphics[scale=.127]{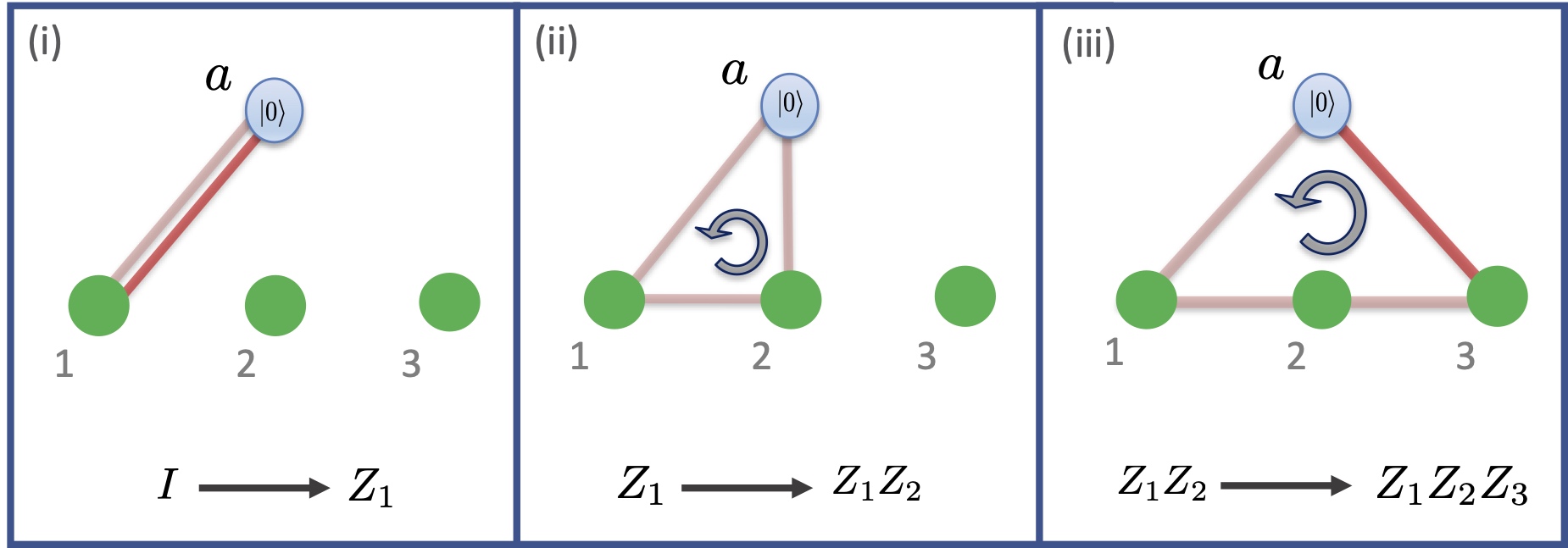}}\caption{\textbf{A protocol for implementing U(1)-invariant unitaries using interaction \emph{XX+YY}.} We show that using 2-local interactions that are invariant under rotations around the z axis, the family of unitaries $\{e^{i\theta Z_1Z_2 Z_3}: \theta\in[0,2\pi)\}$ cannot be implemented, unless one uses ancillary systems. Using a single ancillary qubit $a$, we can implement this family of unitraies just using interactions $R_{rs}=(X_rX_s+Y_rY_s)/2$ for $r,s\in\{1, 2, 3, a\}$  and local $Z_a$ on the ancillary qubit, which are both invariant under rotations around z. The ancillary qubit $a$, highlighted by blue, is initially prepared in state $|0\rangle$ and returns to the same state at the end of the process. 
In part (i) we use Hamiltonians $R_{a1}$ and $Z_a$ to simulate Hamiltonian $Z_1$. To achieve this, we consider the commutator $i[R_{a1},T_{1a}]=Z_1-Z_a$, where  $T_{1a}= \frac{i}{2} [Z_1, R_{1a}]$. Since $a$ is initially in state $|0\rangle$, the effect of Hamiltonian $Z_1-Z_a$ is equivalent to Hamiltonian $Z_1-I$, where $I$ is the identity operator. Therefore, by applying $R_{a1}$ and $Z_a$ in a proper order, we can implement unitaries generated by $Z_1$, up to a global phase. Part (ii)  corresponds to   the commutator $\big[\big[ [R_{a1}, R_{12}], R_{2a} \big]\big]=Z_1(Z_2-Z_a)$. The equality means that by applying $R_{a1}$, $R_{12}$ and $R_{2a}$ in a proper order, we can implement Hamiltonian $Z_1 Z_2-Z_1$. Combining it with $Z_1$ obtained in the step (i) we obtain $Z_1 Z_2$. Part (iii) corresponds to the commutator $i\big[\big[ [R_{a1}, R_{12}], R_{23} \big], T_{3a}\big]=Z_1Z_2(Z_3-Z_a)$. Since qubit $a$ is initially in state $|0\rangle$, the effect of this time evolution on qubits 1, 2, 3 is equivalent to the time evolution generated by Hamiltonian $Z_1Z_2Z_3-Z_1Z_2$ (See Eq.\ref{Eq:loop}).  
Combining this with Hamiltonian $Z_1Z_2$ obtained from step (ii), we obtain  $Z_1Z_2Z_3$.  }\label{loop}
\end{figure}

  \subsection{U(1) symmetry with systems of qubits }

Here, we show how  U(1)-invariant unitaries on qubit systems can be  implemented using a single ancillary qubit. The argument is a straightforward generalization of the idea discussed in Eq.(\ref{Eq:loop}).  
 This  proves   theorem \ref{Thm:energy}  on energy-conserving unitaries, for the special case where all the systems are qubits (As we discussed before, for systems considered in theorem \ref{Thm:energy} energy conservation is equivalent to a U(1) symmetry. To focus on the main idea, in this subsection we phrase the arguments in terms of U(1) symmetry).

Consider a pair of bit strings  $\textbf{b}, \textbf{b}'\in\{0,1\}^n$, such that $\textbf{b}'$ can be obtained from  $\textbf{b}$ by flipping a bit with value $1$, i.e.,  its Hamming weight is $w(\textbf{b}')=w(\textbf{b})-1$ and its  Hamming distance with  $\textbf{b}$ is $d(\textbf{b}',\textbf{b})=1$. Let  $Z_a$ and $I_a$ be, respectively, the Pauli $z$ and the identity operator on the ancillary qubit. Then, Eq.(\ref{kjf00}) implies that  the  Hamiltonian   
$$\textbf{Z}^{\textbf{b}}\otimes I_{a}-\textbf{Z}^{\textbf{b}'}\otimes Z_{{a}}$$
 can be generated using  Hamiltonians $R_{rs}:\  r,s\in\{1,\cdots , n\}\cup \{a\}$ together with Pauli $z$ on the ancillary qubit (This can also be seen using lemma \ref{lem05}).  Assuming this qubit is initially in state $|0\rangle_a$, under  this  Hamiltonian any initial state $|\psi\rangle$ of $n$ qubits evolves to  
 \be\label{jhwv}
e^{i\theta [\textbf{Z}^{\textbf{b}}\otimes I_{a}-\textbf{Z}^{\textbf{b}'}\otimes Z_{{a}}]}\ (|\psi\rangle\otimes  |0\rangle_{a}) = (e^{i\theta [\textbf{Z}^{\textbf{b}}-\textbf{Z}^{\textbf{b}'}]}|\psi\rangle)\otimes |0\rangle_{a} \ .
\ee 
Note that at the end of the process, the ancillary qubit goes back to its initial state.  
Therefore, repeating this, we can implement all Hamiltonians 
\be\label{kh}
\Big\{\textbf{Z}^{\textbf{b}}-\textbf{Z}^{\textbf{b}'}: \textbf{b},\textbf{b}'\in\{0,1\}^n,  w(\textbf{b}')=w(\textbf{b})-1, d(\textbf{b}',\textbf{b})=1 \Big\}\ .
\ee
Furthermore,  by applying Pauli $z$ Hamiltonian on the ancillary qubit, i.e., Hamiltonian $I^{\otimes n}\otimes Z_a$ on the total system, we can also implement the  constant Hamiltonian $I^{\otimes n}$ on $n$ qubits. Linear combinations of this Hamiltonian with Hamiltonians in Eq.(\ref{kh}), give all diagonal Hamiltonians (Note that a similar argument works if rather than state $|0\rangle_a$, the ancillary qubit is prepared in state $|1\rangle_a$). 

Finally, recall that according to theorem \ref{Thm5}, combining diagonal Hamiltonians with Hamiltonians $X_jX_{j+1}+Y_jY_{j+1}: j=1,\cdots  n-1$, we can generate all  $U(1)$-invariant Hamiltonians. This proves theorem \ref{Thm:energy} in the special case of qubit systems. 

\subsection*{Geometrically local interactions}\label{Sec:Geometry}
So far, in our discussion we have not assumed any particular geometry for the system and the labels $1,\cdots, n$ of $n$ qubits was arbitrary. Next, we assume the qubits lie on an open chain and their labeling corresponds to their order in the chain.  For instance, the qubits are ordered from left to right, and the leftmost  qubit is labeled as qubit 1.  

Suppose one can turn on and off $XX+YY$ interactions between nearest-neighbor qubits. Furthermore, suppose an ancillary qubits $a$ can interact with all qubits with $XX+YY$ interaction. Furthermore,  in addition to $XX+YY$ interactions between  the qubits, suppose one can also apply local Pauli $Z$ on the ancillary qubit. Then, the overall, Hamiltonian is in the form
\be\label{glj}
H(t)=\sum_{j=1}^{n-1} c_j(t)\ (X_jX_{j+1}+Y_jY_{j+1})+     d_j(t)\ (X_jX_{a}+Y_jY_{a})+ z(t) Z_a\ ,
\ee
where $c_j(t)$, $d_l(t)$ and $z(t)$ are arbitrary real functions. As before, we assume the ancilllary qubit is initially in state $|0\rangle$. 

This Hamiltonian does not allow direct interactions between arbitrary pairs of qubits.   Nevertheless, it turns out that using this family of Hamiltonians 
we can implement all U(1)-invariant  unitary transformations  on qubits $1,\cdots, n$. To see this first note that  for any connected subset of qubits, corresponding to a  gapless sequence of integers $j, j+1, \cdots,j'-1, j'$, we can implement the corresponding Hamiltonian $Z_{j} \cdots Z_{j'}$. This follows immediately from Eq.(\ref{kjf00}) together with the argument in Eq.(\ref{jhwv}).  In particular, we can simulate interaction $Z_j Z_{j+1} $ for any neighboring qubits $j$ and $j+1$. Now the key observation is that  by combining interactions $Z_j Z_{j+1}$ and $X_j X_{j+1}+Y_j Y_{j+1}$ one can implement the swap unitary on qubits $j$ and $j+1$, i.e.,  the unitary that exchanges the states of these qubits. In particular, 
\be
e^{i \frac{\pi}{4} (X_j X_{j+1}+Y_j Y_{j+1}+Z_j Z_{j+1})}=e^{i\frac{\pi}{4}} {S}_{j, j+1}\ ,
\ee  
where $S_{j,j+1}$ is the swap operator that exchanges the state of qubits $j$ and $j+1$. By combining swaps on nearest-neighbor qubits, we obtain all permutations on $n$ qubits. Therefore, we can change the order of qubits, arbitrarily. Combining this with the above technique we can implement arbitrary diagonal unitary transformations. For instance, to implement the unitary  $e^{i \theta Z_l Z_{m}} $ between any two arbitrary qubits $l$ and $m$ with $l<m$, we first apply the permutation operator $S_{l+1, m}$ that exchanges the states of qubits $m$ and $l+1$ and leave the other qubits unchanged. Then we apply   
the unitary $e^{i \theta Z_l Z_{l+1}}$ and finally  exchange the state of qubits $l+1$ and $m$ again. In this way, we obtain
\be\label{change}
S_{l+1, m} e^{i \theta Z_l Z_{l+1}} S_{l+1, m}=e^{i \theta Z_l Z_{m}}\ .
\ee    
Finally, recall that according to theorem \ref{Thm5}, 
 by combining diagonal unitaries with unitaries generated by Hamiltonians $\{R_{j,j+1}=\frac{1}{2}(X_j X_{j+1}+Y_{j}Y_{j+1}): j=1,\cdots, n-1\}$, we can implement all U(1)-invariant unitaries.  This proves the claim that using a single ancillary qubit in initial state $|0\rangle$ and by properly choosing functions $c_j$, $d_j$ and $z$, we can implement a general energy-conserving unitary using Hamiltonian $H(t)$ in Eq.(\ref{glj}). 

In the above scheme, the ancillary qubit $a$ needs to interact with all qubits in the system. We can relax this requirement if in addition to interactions  $XX+YY$, we have access to interactions $ZZ$. More precisely, consider  the family of Hamiltonians 
\be\label{glj33}
H'(t)=  r(r) (X_1X_{a}+Y_1Y_{a})+ s(t) Z_1Z_a + z(t) Z_a+\sum_{j=1}^{n-1} c_j(t)\ (X_jX_{j+1}+Y_jY_{j+1})+b_j(t)\ Z_j Z_{j+1}\ ,
\ee
where $c_j$, $b_j$, $r$, $s$ and $z$ are arbitrary real functions. 
Using an argument similar to the above argument, we can easily see that universality can also be achieved using this family of Hamiltonians. Again, the key point is that by combining $XX+YY$ and $ZZ$ interactions on nearest neighbor qubits, we can swap their orders, and therefore,  we can permute the order of all qubits arbitrarily. Hence, the restriction to nearest-neighbor interactions  becomes irrelevant.




 

\newpage

\subsection{General energy-conserving unitaries}

  \begin{figure}
{\includegraphics[scale=.7]{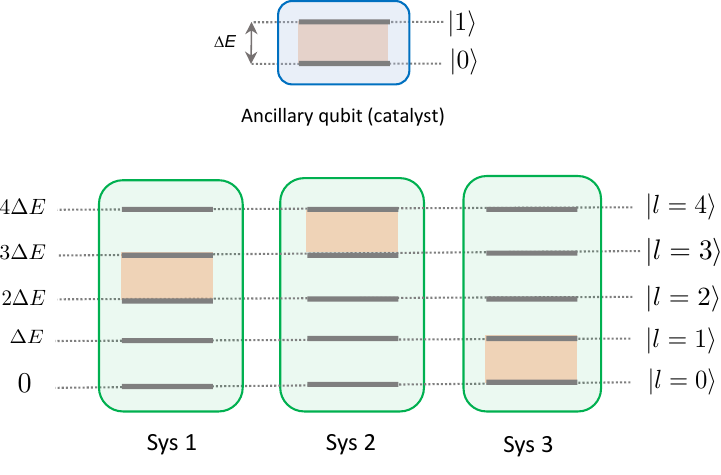}}\caption{\textbf{A scheme for implementing  energy-conserving unitaries on composite systems.} 
In each system we consider a pair of energy eigenstates with energy difference $\Delta E$, which can be interpreted as a single qubit. Then,  we can apply the protocol  discussed in the previous section for qubits with U(1) symmetry (See Fig.\ref{loop}). To implement this protocol,  the  qubits defined in different systems are sequentially   coupled to each other and to an ancillary qubit, via 2-local energy-conserving interactions in Eq.(\ref{int65}) and Eq.(\ref{int}). The ancilla is initially prepared  in its ground state $|0\rangle$, whose energy gap with the excited state $|1\rangle$, is $\Delta E$. This implies that the interaction in Eq.(\ref{int}), which couples the ancilla to the systems, is energy-conserving.  
Using these 2-local energy-conserving interactions and following the protocol introduced in the previous section, we can implement energy-conserving unitaries on the selected energy pairs.  Repeating these steps   
with other  pairs of energy eigenstates with energy difference $\Delta E$, we obtain all energy-conserving unitaries (See Appendix \ref{App:C}).}\label{energies}
\end{figure}


Next, we consider implementation of general energy-conserving unitaries  beyond qubit systems and prove theorem \ref{Thm:energy}. Note  that for the family of Hamiltonians considered in this theorem, energy conservation is equivalent to a U(1) symmetry.  
 The proposed scheme for implementing general energy-conserving unitaries is a generalization of the protocol  used in the qubit case with U(1) symmetry. In particular, similar to that case, there are two main steps in the argument: First, we show how a general \emph{diagonal} energy-conserving unitary can be implemented (lemma \ref{ref:lem23}), and then we show that by combining diagonal unitaries with 2-local energy-conserving unitaries, we obtain all energy-conserving unitaries (lemma \ref{eheheh}).

To simplify the notation and analysis, we  
focus on the case of systems with identical Hilbert spaces and Hamiltonians. We also assume the Hamiltonians are non-degenerate (These assumptions are  non-essential in the argument and can be relaxed). In particular, we consider $n\ge 1$ systems each with a $d$-dimensional Hilbert space and with the intrinsic Hamiltonian $\Delta E \sum^{d-1}_{r=0}  |r\rangle\langle r |$. The systems are labeled by $j=1,\cdots, n$. We assume  before and after applying the energy-conserving unitary, the systems are non-interacting, i.e., their total intrinsic Hamiltonian is  
\be
H_\text{intrinsic}=\sum_{j=1}^n H_j\ , 
\ee
where 
\be
H_j=\Delta E \sum^{d-1}_{r=0} r |r\rangle\langle r |_j \ ,
\ee
 is the Hamiltonian of system $j$ tensor product with the identity operators on the rest of systems.  \\


This scheme  uses an ancillary qubit with Hamiltonian $\Delta E |1\rangle\langle 1|_{\text{a}}$, initially prepared in the ground state $|0\rangle_\text{a}$. 
Suppose in each system we pick a pair of energy levels with energy difference $\Delta E$, which can be interpreted as a qubit (See Fig.\ref{energies}). Then, following the protocol introduced in the previous section,  we can implement all energy-conserving unitaries defined on these qubits using 2-local energy-conserving unitaries.

 The required interactions for implementing this scheme are  
\bes\label{arar2}
\begin{align}
R^{(l)}_{\text{a}, j}&\equiv R^{(l)}_{j, \text{a}} \equiv |l-1\rangle\langle l|_{j}\otimes  |1\rangle\langle0|_{\text{a}}+ |l\rangle\langle l-1|_{j}\otimes  |0\rangle\langle1|_\text{a}           \ \ \ \ &&: l=1,\cdots , d-1\ ;\ \ j=1,\cdots , n\\
 R^{(l,l')}_{j,j+1}&\equiv R^{(l',l)}_{j+1,j}\equiv  |l-1\rangle\langle l|_j\otimes |l'\rangle\langle l'-1|_{j+1}+|l\rangle\langle l-1|_{j}\otimes |l'-1\rangle\langle l'|_{j+1}          \ \ \ \ &&: l,l'=1,\cdots , d-1\ ;\ \ j=1,\cdots , n-1\ .
  \end{align}
  \ees
Note that the above operators are defined  on the Hilbert space $(\mathbb{C}^d)^{\otimes n}\otimes \mathbb{C}^2$, corresponding to $n$ systems and the ancillary qubit. Operator $R^{(l)}_{j, \text{a}}$
acts non-trivially on system $j$ and ancillary qubit $\text{a}$ and $R^{(l',l)}_{j+1,j}$ acts non-trivially on systems $j$ and $j+1$. 
Therefore, these interactions are 2-local. Furthermore, if the systems are placed in the order corresponding to their  labels $j=1,\cdots, n$, then the terms $R^{(l,l')}_{j,j+1}$ are interactions between a pair of nearest-neighbor    systems. 

Also, note that the above interactions are  all energy-conserving, i.e., commute with the total intrinsic  Hamiltonian of $n$ systems and the ancillary qubits  \be
H_\text{tot}\equiv \sum_{j=1}^n H_j + \Delta E |1\rangle\langle1|_{\text{a}}\ .
\ee

 In fact, if  we interpret states  $|l-1\rangle_j$ and $|l\rangle_j$ as states $|0\rangle$ and $|1\rangle$ of a qubit, then $R^{(l)}_{j,\text{a}}$  will be equivalent to the interaction $\frac{1}{2}(XX+YY)$  between a pair of qubits.   Similarly, $R^{(l,l')}_{j,j'}$ can be interpreted as interaction  $\frac{1}{2}(XX+YY)$ between a pair of qubits, defined in systems $j$ and $j'$. Recall that the protocol defined in the previous section can be implemented using interactions $XX+YY$ and local $Z$ on the ancillary qubit. Therefore,  interactions  defined in Eq.(\ref{arar2})   together with local $Z_\text{a}$ on the ancillary qubit, allow us to apply  this protocol and implement all energy-conserving unitaries defined on the the selected pairs of energy eigenstates. Using this idea we show that

\begin{lemma}\label{lem221}
Consider a Hermitian operator  $H_{{\bf{n}}}$ acting on the Hilbert space $(\mathbb{C}^d)^{\otimes n}$ that commutes with $H_\text{intrinsic}=\sum_{j=1}^n H_j$. Then, there exists an operator $\tilde{H}_{{\bf{n}}, \text{a}}$ acting on   $(\mathbb{C}^d)^{\otimes n}\otimes \mathbb{C}^2$, such that  $i \tilde{H}_{{\bf{n}}, \text{a}}$ is in the real Lie algebra, generated by $i Z_{\text{a}}\ , i R^{(l)}_{j, \text{a}}  , i R^{(l,l')}_{j,j+1}$, i.e. 
\be\label{wffw}
i \tilde{H}_{{\bf{n}}, \text{a}}\in \mathfrak{alg}_\mathbb{R}\Big\{i Z_{\text{a}}\ , i R^{(l)}_{j, \text{a}}  , i R^{(l,l')}_{j,j+1}:\  l,l'=1,\cdots , d-1\ ;\ \ j=1,\cdots , n
\Big\}
  \ ,
\ee
and
\be
|\psi\rangle\in (\mathbb{C}^d)^{\otimes n}:\ \  \ \ \ \  \tilde{H}_{{\bf{n}}, \text{a}} (|\psi\rangle\otimes |0\rangle_\text{a})=({H}_{{\bf{n}}}|\psi\rangle)\otimes |0\rangle_\text{a}\ .
\ee

\end{lemma}
Any energy-conserving unitary $V_\textbf{n}$ on the $n$ systems  can be written as $e^{i{H}_{{\bf{n}}}}$, where ${H}_{{\bf{n}}}$ commutes with the intrinsic Hamiltonian $H_\text{intrinsic}=\sum_{j=1}^n H_j$. Therefore, the above lemma implies that  there exists an operator  
$i \tilde{H}_{{\bf{n}}, \text{a}}$ in the real Lie algebra generated by $i Z_{\text{a}}\ , i R^{(l)}_{j, \text{a}}  , i R^{(l,l')}_{j,j+1}$, such that 
\be\label{lwwdhw92}
|\psi\rangle\in (\mathbb{C}^d)^{\otimes n}:\ \  \ \ \ \  e^{i\tilde{H}_{{\bf{n}}, \text{a}}} (|\psi\rangle\otimes |0\rangle_\text{a})=(e^{i{H}_{{\bf{n}}}}|\psi\rangle)\otimes |0\rangle_\text{a}=V_\textbf{n}|\psi\rangle\otimes |0\rangle_\text{a}\ .
\ee
Furthermore, the fact that $i \tilde{H}_{{\bf{n}}, \text{a}}\in \mathfrak{alg}_\mathbb{R}\big\{i Z_{\text{a}}\ , i R^{(l)}_{j, \text{a}}  , i R^{(l,l')}_{j,j+1}:\  l,l'=1,\cdots , d-1\ ;\ \ j=1,\cdots , n
\big\}$ implies that the unitary $e^{i\tilde{H}_{{\bf{n}}, \text{a}}}$ is in the Lie group generated by unitaries 
\be\label{ljdjdj}
\exp({i\theta Z_{\text{a}}})\ ,\  \exp(i\theta R^{(l)}_{j, \text{a}})\ ,\  \exp(i\theta R^{(l)}_{j, \text{a}})\  ,\   \exp(i\theta R^{(l,l')}_{j,j+1})\ \ \ \  : \theta\in[0, 2\pi) ;\   l,l'=1,\cdots , d-1\ ;\ \ j=1,\cdots , n\ ,
\ee
which are all 2-local and energy-conserving. 
Moreover, using {Fact 1} in Supplementary Note 1,  the group generated by these unitaries 
is compact and therefore using the result of 
\cite{d2007introduction}, any unitary in this group is uniformly finitely generated by the generating set in Eq.(\ref{ljdjdj}). This result, which is summarized in theorem \ref{Thm:iejrg}, proves the statement of theorem \ref{Thm:energy}.

This result means that there is a map  from energy-conserving unitaries on $(\mathbb{C}^d)^{\otimes n}$ to energy-conserving unitaries on    $(\mathbb{C}^d)^{\otimes n}\otimes \mathbb{C}^2$, namely 
\be
{V}_{\bf{n}} \longrightarrow\  \tilde{V}_{\bf{n},\text{a}}=  {V}_{\bf{n}} \otimes |0\rangle\langle 0|_\text{a}+ W_{\bf{n}} \otimes |1\rangle\langle1|_\text{a}\ , 
\ee
where $W_{\bf{n}}$ is also an energy-conserving unitary on $(\mathbb{C}^d)^{\otimes n}$, and $\tilde{V}_{\bf{n},\text{a}}$ can be generated by the family of unitaries in Eq.(\ref{ljdjdj}).  \\

%
%
%
%
%
%
%
%
%
%
%
%
%
%
%
%

\noindent{\textbf{Hamiltonian Picture:}}\\

Alternatively, we can understand this result in the Hamiltonian picture. Suppose to implement an energy-conserving unitary, we moodify the intrinsic Hamiltonian of systems and ancilla    $H_\text{tot}=\sum_{j=1}^n H_j + \Delta E |1\rangle\langle1|_{\text{a}}$. In particular, suppose we add  
2-local energy-conserving interactions in Eq.(\ref{arar2}) to $H_\text{tot}$, and obtain the family of Hamiltonians \be\label{ljwj12}
H_{{\bf{n}},\text{a}}(t)=H_\text{tot}+  g_\text{a}(t)\  Z_\text{a}+\sum_{j=1}^n \sum_{l=1}^{d-1}  g^{(l)}_{j, \text{a}}(t)\ R^{(l)}_{j, \text{a}}  + \sum_{j=1}^{n-1} \ \sum_{l,l'=1}^{d-1}  g_{j, j+1}^{(l,l')}(t)\ R^{(l,l')}_{j,  j+1} \ ,
\ee
where $ g_\text{a}$, $g^{(l)}_{j, \text{a}}$ and  $g_{j, j'}^{(l,l')}$ are real functions of time $t$, which vanish for $t<0$, and for sufficiently large $t$.  
 We are interested in the unitary transformations generated by this family of Hamiltonians for different choices of these functions, i.e., unitaries satisfying
\be
\frac{d}{dt} V_{{\bf{n}},\text{a}}(t)=-i H_{{\bf{n}},\text{a}}(t) V_{{\bf{n}},\text{a}}(t)\ , \ \  t\ge 0 
\ee
where $V_{{\bf{n}},\text{a}}(0)$ is the identity operator on $(\mathbb{C}^d)^{\otimes n}\otimes \mathbb{C}^2$. Clearly, 
\be
\forall t\ge0: \ \ \ \ \Big[H_{{\bf{n}},\text{a}}(t) , H_\text{tot}\Big]=0\ ,
\ee
which means the family of unitaries generated by these Hamiltonians are energy-conserving. 
Note that  Hamiltonian $H_{{\bf{n}},\text{a}}(t)$ in Eq.(\ref{ljwj12}) contains a time-independent  term $H_\text{tot}$, which corresponds to the intrinsic Hamiltonians of the $n$ systems and the ancilla. It turns out that the existence of this constant term does not restrict the family of unitaries generated by this family of Hamiltonians. In particular, this family contains the family of unitaries generated by Hamiltonians 
\be\label{ljwj1233}
g_\text{a}(t)\  Z_\text{a}+\sum_{j=1}^n \sum_{l=1}^{d-1}  g^{(l)}_{j, \text{a}}(t)\ R^{(l)}_{j, \text{a}}  + \sum_{j=1}^{n-1} \ \sum_{l,l'=1}^{d-1}  g_{j, j+1}^{(l,l')}(t)\ R^{(l,l')}_{j,  j+1}\ ,
\ee
 where we have dropped the term $H_\text{tot}$ in Eq.(\ref{ljwj12}). This follows from the fact that $H_\text{tot}$ 
commutes with all other terms in the Hamiltonian $H_{{\bf{n}},\text{a}}(t)$ and,  furthermore, it generates a periodic time evolution, with period $2\pi/\Delta E$. Therefore, if the total time of implementing the desired unitary is an integer multiple of $2\pi/\Delta E$, the presence of $H_\text{tot}$  does not have any effect on the implemented unitary. This can always be achieved by adding a time delay less than $2\pi/\Delta E$, during which the other terms are turned off.

As we have seen before,   the standard  results of quantum control theory  \cite{d2007introduction, jurdjevic1972control} imply that, using the family of Hamiltonians in Eq.(\ref{ljwj1233}) and 
assuming  $g_\text{a}, g^{(l)}_{j, \text{a}}$,  and $g_{j, j+1}^{(l,l')}$ are arbitrary real functions,
 we can generate any unitary $e^{i G}$, with $iG$ in the Lie algebra defined in Eq.(\ref{wffw}).   Combining this with the argument in Eq.(\ref{lwwdhw92}), we conclude that
 
\begin{theorem}\label{Thm:iejrg}
Consider an arbitrary  energy-conserving unitary $V_{\bf{n}}$ acting on $(\mathbb{C}^d)^{\otimes n}$  (i.e., a unitary satisfying $[V_{\bf{n}}\ , H_\text{intrinsic}]=0$).  Then, there exists  a unitary  $\tilde{V}_{\bf{n},\text{a}}$ acting on $(\mathbb{C}^d)^{\otimes n}\otimes \mathbb{C}^2$, such that
 \be
\forall |\psi\rangle\in(\mathbb{C}^d)^{\otimes n}:\ \ \  \tilde{V}_{\bf{n},\text{a}} (|\psi\rangle |0\rangle_\text{a})=V_{\bf{n}} |\psi\rangle \otimes |0\rangle_\text{a}\ ,
 \ee
 and $\tilde{V}_{\bf{n},\text{a}}$ can be generated by a finite sequence of 2-local energy-conserving unitaries in Eq.(\ref{ljdjdj}). Equivalently,   $\tilde{V}_{\bf{n},\text{a}}$  can be implemented with the family of energy-conserving  Hamiltonians $H_{\bf{n},\text{a}}(t)$ defined in Eq.(\ref{ljwj12}). 
\end{theorem}

Therefore, to complete the proof of this result we need to prove  lemma \ref{lem221}. But, first we discuss a modified version of this scheme.  

 

\subsection{A Modified Scheme with Two Ancillary Qubits}\label{App:Two}

In the above scheme we need system-system interactions $R^{(l,l')}_{j,j'}$. It turns out that this interaction can be easily engineered using system-ancilla interactions, provided that we can use a second ancillary qubit, labeled as qubit b, with Hamiltonian $\Delta E |1\rangle\langle 1|_\text{b}$.
This follows from the fact that 
\be
R^{(l,l')}_{j,j'} Z_\text{b}=\frac{1}{2}\big[R^{(l)}_\text{$j$,b}\ , [ R^{(l')}_\text{$j'$,b}\ ,\ Z_\text{b}] \big] \  .
\ee
Therefore, if qubit b is initially in state $|0\rangle_{\text{b}}$, then 
\be\label{wfwf01}
\exp\big(i\theta \frac{1}{2}\big[R^{(l)}_\text{$j$,b}\ , [ R^{(l')}_\text{$j'$,b}\ ,\ Z_\text{b}] \big]\big)|\phi\rangle|0\rangle_\text{b}=\exp\big(i\theta R^{(l,l')}_{j,j'} \big) |\phi\rangle\otimes |0\rangle_\text{b}\ ,
\ee
where $|\phi\rangle$ is an arbitrary state of the rest of systems. 

Based on this observation, we consider interactions between systems $j=1,\cdots n$ and ancillary qibits a and b: 
\begin{align}
R^{(l)}_{\text{a}, j}&\equiv R^{(l)}_{j, \text{a}} \equiv |l-1\rangle\langle l|_{j}\otimes  |1\rangle\langle0|_{\text{a}}+ |l\rangle\langle l-1|_{j}\otimes  |0\rangle\langle1|_\text{a}           \ \ \ \ &&: l=1,\cdots , d-1\ ,\ \ j=1,\cdots , n\\
 R^{(l)}_{\text{b}, j}&\equiv R^{(l)}_{j, \text{b}} \equiv |l-1\rangle\langle l|_{j}\otimes  |1\rangle\langle0|_{\text{b}}+ |l\rangle\langle l-1|_{j}\otimes  |0\rangle\langle1|_\text{b}           \ \ \ \ &&: l=1,\cdots , d-1\ ,\ \ j=1,\cdots , n \ . \end{align}
Then, in this modified  scheme instead of  
2-local  energy-conserving unitaries in Eq.(\ref{ljdjdj}), we consider unitaries 
\be\label{ljdjdj2}
\exp({i\theta Z_{\text{a}}})\ ,\  \exp(i\theta R^{(l)}_{j, \text{a}})\ ,\  \exp(i\theta R^{(l)}_{j, \text{a}})\  ,\    \exp(i\theta R^{(l)}_{j, \text{b}})\ \ \ \  : \theta\in[0, 2\pi) ;\   l,l'=1,\cdots , d-1\ ;\ \ j=1,\cdots , n\ .
\ee
Similarly, in the Hamiltonian picture, instead of Hamiltonians  in Eq.(\ref{ljwj12}), we consider the family of Hamiltonians 
\be\label{kjshd}
H_{{\bf{n}},\text{a}, \text{b}}(t)=\Big(\sum_{j=1}^n H_j + \Delta E |1\rangle\langle1|_{\text{a}}+\Delta E  |1\rangle\langle1|_{\text{b}}\Big)+ g_\text{a}(t)\  Z_\text{a}+g_\text{b}(t)\  Z_\text{b}+\sum_{j=1}^n \sum_{l=1}^{d-1}  g^{(l)}_{j, \text{a}}(t)\ R^{(l)}_{j, \text{a}}  + \sum_{j=1}^n \sum_{l=1}^{d-1}  g^{(l)}_{j, \text{b}}(t)\ R^{(l)}_{j, \text{b}} \ ,
\ee
where $ g_\text{a}$, $g_\text{b}$, $g^{(l)}_{j, \text{a}}$ and  $g^{(l)}_{j, \text{b}}$ are arbitrary real functions. Note that
\be
\forall t\ge0: \ \ \ \ \Big[H_{{\bf{n}},\text{a}, \text{b}}(t) ,  \Big(\sum_{j=1}^n H_j + \Delta E |1\rangle\langle1|_{\text{a}}+ \Delta E |1\rangle\langle1|_{\text{b}}\Big) \Big]=0\ ,
\ee
and therefore the family of unitaries generated by these Hamiltonians are energy-conserving. 

Then, combining the observation in Eq.(\ref{wfwf01}) with theorem \ref{Thm:iejrg}, we conclude that 

\begin{corollary}\label{erh}
Consider an arbitrary  energy-conserving unitary $V_{\bf{n}}$ acting on $(\mathbb{C}^d)^{\otimes n}$  (i.e., a unitary satisfying $[V_{\bf{n}}\ , H_\text{intrinsic}]=0$).  
 There exists  a unitary  $\tilde{V}_{\bf{n},\text{a},\text{b}}$ acting on $(\mathbb{C}^d)^{\otimes n}\otimes \mathbb{C}^2\otimes \mathbb{C}^2$, such that
  \be
\forall |\psi\rangle\in(\mathbb{C}^d)^{\otimes n}:\ \ \  \tilde{V}_{\bf{n},\text{a},\text{b}} (|\psi\rangle |0\rangle_\text{a}|0\rangle_\text{b})=V_{\bf{n}} |\psi\rangle \otimes |0\rangle_\text{a}|0\rangle_\text{b}\ ,
 \ee
and $\tilde{V}_{\bf{n},\text{a},\text{b}}$ can be generated by a finite sequence of 2-local energy-conserving unitaries in Eq.(\ref{ljdjdj2}). Equivalently,   $\tilde{V}_{\bf{n},\text{a},\text{b}}$  can be implemented with the family of energy-conserving  Hamiltonians $H_{\bf{n},\text{a},\text{b}}(t)$ defined in Eq.(\ref{kjshd}). 
\end{corollary}

\newpage

\subsection{Implementing diagonal unitaries with 2-local energy-conserving interactions \\(Proof of lemma \ref{lem221} for the special case of diagonal Hamiltonians)}
In this section we focus on diagonal unitaries, i.e., those that  commute with the Hamiltonians of all systems $j=1,\cdots, n$ and prove the following lemma, which is the special case of lemma \ref{lem221}  for diagonal Hamiltonians.

\begin{lemma}\label{ref:lem23}
For any  Hermitian operator $H_\text{diag}$ on $(\mathbb{C}^{d})^{\otimes n}$ that is diagonal in the basis  $\{\bigotimes_{j=1}^n |r_j\rangle: r_j=0,\cdots, d-1\}$, there exists an operator $\tilde{H}_\text{diag}$ on $(\mathbb{C}^{d})^{\otimes n}\otimes \mathbb{C}^2$, such that  $i \tilde{H}_\text{diag}$ is in the real Lie algebra $\mathfrak{alg}_\mathbb{R}\Big\{i Z_{\text{a}}\ , i R^{(l)}_{j, \text{a}}  , i R^{(l,l')}_{s,s+1}:\  l,l'=1,\cdots , d-1\ ;\ \ j=1,\cdots , n; s=1,\cdots, n-1
\Big\}$ and
\be
\forall |\psi\rangle\in(\mathbb{C}^d)^{\otimes n}:\ \  \ \ \ \ \tilde{H}_\text{diag}(|\psi\rangle\otimes |0\rangle_\text{a})={H}_\text{diag}|\psi\rangle\otimes |0\rangle_\text{a}\ .
\ee   
\end{lemma}
As we explained in Fig.\ref{energies} the main idea is to think of a pair of consecutive levels as a qubit, and apply our qubit results.  

Any diagonal unitary can be written as $e^{i H_\text{diag}}$,  where  
$H_\text{diag}$ is a diagonal Hermitian operator, i.e.,  can be written as
\be
H_\text{diag}=\sum_{s_1,\cdots, s_n=0}^{d-1}  h_{s_1,\cdots, s_n} \ \bigotimes_{j=1}^n |s_j\rangle\langle  s_j|\ , 
\ee
 where $h_{s_1,\cdots, s_n} \in\mathbb{R}$.  
 
We start with the case of $n=1$, i.e.,   a single system with the Hilbert space $\mathbb{C}^d$. Consider the pair of energy eigenstates $|l-1\rangle$ and $|l\rangle$ with energies $(l-1)\times \Delta E $ and $l\times \Delta E $, respectively. Define 
\begin{align}
Z^{(l)}  &\equiv |l-1\rangle\langle l-1|-|l\rangle\langle l|    \ \ \ \ \ \  \ \ \ \ \ \ : l=1,\cdots , d-1\ . 
 \end{align}
Consider the set of operators  
\be\label{basis}
\Big\{I_d , Z^{(l)}\equiv|l-1\rangle\langle l-1|- |l\rangle\langle l| : l=1,\cdots , d-1\Big\}\ ,
\ee  
where $I_d$ is the identity operator on $\mathbb{C}^d$. It can be easily seen that the above $d$ operators form a basis for diagonal operators. 
In particular, 
\be\label{qef}
\text{Span}_\mathbb{R}\big\{|s\rangle\langle  s|\ :  0 \le s\le d-1  \big\}=\text{Span}_\mathbb{R}\big\{I_d , Z^{(l)}\equiv|l-1\rangle\langle l-1|- |l\rangle\langle l| : l=1,\cdots , d-1\big\}\ .
\ee
Next, consider $n$ systems labeled as $j=1,\cdots, n$. The set of operators $\big\{\bigotimes_{j=1}^n |s_j\rangle\langle  s_j|\ :  0 \le s_j\le d-1  \big\}$ spans the space of diagonal operators on $(\mathbb{C}^d)^{\otimes n}$. Clearly, this set can be obtained as the $n-$fold tensor product of the set $\big\{|s\rangle\langle  s|\ :  0 \le s\le d-1  \big\}$, i.e.
\be\label{jldkdf}
\Big\{\bigotimes_{j=1}^n |s_j\rangle\langle  s_j| : s_j=0,\cdots, d-1 \Big\}= \Big\{|s\rangle\langle  s|\ :  0 \le s\le d-1  \Big\}^{\otimes n}\ .
\ee
Next, we consider the $n$-fold tensor product of operators $\big\{I_d , Z^{(l)}\equiv|l-1\rangle\langle l-1|- |l\rangle\langle l| : l=1,\cdots , d-1\big\}$, which appear in the right-hand side of Eq.(\ref{qef}).  For each system $j$ consider the pair of energy eigenstates $|l-1\rangle_j$ and $|l\rangle_j$ with energies $(l-1)\times \Delta E $ and $l\times \Delta E $, respectively. Define  
\begin{align}
Z^{(l)}_j  &\equiv |l-1\rangle\langle l-1|_j-|l\rangle\langle l|_j    \ \ \ \ \ \  \ \ \ \ \ \ : l=1,\cdots , d-1\ ,\ \ j=1,\cdots , n \ .
 \end{align}
Then,  the $n$-fold tensor product of $\big\{I_d , Z^{(l)}\equiv|l-1\rangle\langle l-1|- |l\rangle\langle l| : l=1,\cdots , d-1\big\}$ gives the set of operators 
\be
\Big\{I_d , Z^{(l)} : l=1,\cdots , d-1\Big\}^{\otimes n}= \Big\{I, \prod_{r=1}^t Z^{(l_r)}_{j_r}: \ 1 \le t\le n\ ,\   0 <j_1<j_2<\cdots <j_t<n+1 \ , \ 1 \le l_r\le d-1 
  \Big\} \ ,
\ee
where $I=I^{\otimes n}_d$ is the identity operator on $(\mathbb{C}^d)^{\otimes n}$. Combining this with Eq.(\ref{qef}) and Eq.(\ref{jldkdf}), we find 
\begin{align}
\text{Span}_\mathbb{R}\big\{\bigotimes_{j=1}^n |s_j\rangle\langle  s_j|\ :  0 \le s_j\le d-1  \big\}=\text{Span}_\mathbb{R} \Big\{I, \prod_{r=1}^t Z^{(l_r)}_{j_r}: \ 1 \le t\le n\ ,\   0 <j_1<j_2<\cdots <j_t<n+1 \ , \ 1 \le l_r\le d-1 
  \Big\} \ .
\end{align}
Given that all operators $\prod_{r=1}^t Z^{(l_r)}_{j_r}$ are traceless, we conclude that 
\begin{lemma}\label{lem:diag}
Any  Hermitian operator on $(\mathbb{C}^{d})^{\otimes n}$ that is diagonal in the basis  $\{\bigotimes_{j=1}^n |r_j\rangle: r_j=0,\cdots, d-1\}, $
can be written as
\begin{align}
H_\text{diag}&=\sum_{r_1,\cdots, r_n=0}^{d-1}  h_{r_1,\cdots, r_n} \bigotimes_{j=1}^n |r_j\rangle\langle  r_j| \\
&=   \frac{\Tr(H_\text{diag})}{\Tr(I)}  I + \sum_{t=1}^n \ \ \sum_{\substack{j_1,\cdots j_t\\  0< j_1<j_2<\cdots <j_t< n+1  } }^n \ \ \sum_{l_{1},\cdots , l_{t}=1}^{d-1}  c^{(l_{1},\cdots, l_{t})}_{j_1,\cdots, j_t}\  Z^{(l_1)}_{j_1}Z^{(l_2)}_{j_2} \cdots Z^{(l_{t})}_{j_{t}}\ ,
\end{align}
for a set of real coefficients $h_{r_1,\cdots, r_n}$, and $c^{(l_{1},\cdots, l_{t})}_{j_1,\cdots, j_t}$, where 
$Z^{(l)}_j \equiv |l-1\rangle\langle l-1|_j-|l\rangle\langle l|_j    \ : l=1,\cdots , d-1; j=1,\cdots , n $.
\end{lemma}
Therefore, to generate a general diagonal unitary evolution, up to a global phase,  it suffices to implement all Hamiltonians  
\be
Z^{(l_1)}_{j_1}Z^{(l_2)}_{j_2} \cdots Z^{(l_{t})}_{j_{t}}\ \ :\  \ t=1,\cdots, n\ ; \ \  0< j_1<j_2<\cdots <j_t< n+1\ ; l_1,\cdots, l_j=0,\cdots , d-1\ .
\ee
Next, note that each pair of states  $\{|l-1\rangle_j, |l\rangle_j\}$ can be interpreted as a separate qubit and $Z^{(l)}_{j}$ can be interpreted as the 
Pauli $Z$ operator associated to this qubit. Hence, we can mimic the argument in Appendix \ref{App:B} in the case of qubits. Following this analogy, we define 
\begin{align}
T^{(l)}_{ \text{a}, j}&\equiv \frac{i}{2}[ Z_{\text{a}} , R^{(l)}_{\text{a}, j} ]= i\Big(|l\rangle\langle l-1|_{j}\otimes  |0\rangle\langle1|_\text{a}- |l-1\rangle\langle l|_{j}\otimes  |1\rangle\langle0|_{\text{a}}\Big)      \ \ \ \ &&: l=1,\cdots , d-1\ ,\ \ j=1,\cdots , n\ .
 \end{align}
With this definition we can easily see that

\begin{empheq}{align}\label{lsls}
D^{(l)}_j&\equiv Z^{(l)}_j-Z_\text{a}=i \big[iR_\text{$j$,a}^{(l)} , iT_\text{a, $j$}^{(l)}\big]= \frac{i}{2}\big[i R^{(l)}_{j, \text{a}} , [ i Z_{\text{a}} , i R^{(l)
}_{\text{a}, j} ]\big] \ .
\end{empheq}
Furthermore, rewriting Eq.(\ref{kjf00}), we find that for any distinct $t\ge 2$ systems labeled by $j_1<j_2<\cdots<j_t$, it holds that  
\begin{empheq}{align}
D^{(l_1,\cdots, l_t)}_{j_1,\cdots, j_t}&\equiv  (Z^{(l_1)}_{j_1}-Z_{\text{a}})\   Z^{(l_{2})}_{j_2}...Z^{(l_{t})}_{j_{t}} =  
\begin{cases}\!
  \begin{aligned}[b]
   c_t i\   \big[[\cdots[ [i R^{(l_1,l_2)}_{j_1, j_2}, i R^{(l_2,l_3)}_{j_2j_3} ], i R^{(l_3,l_4)}_{j_3j_4}]\cdots , i R^{(l_{t})}_{j_{t}, \text{a} }], i R^{(l_1)}_{\text{a}, j_{1}}\big]\ \ : t \text{ even} \\ 
   &   \\
       c_t i\  \big[[\cdots[ [i R^{(l_1,l_2)}_{j_1, j_2}, i R^{(l_2,l_3)}_{j_2j_3} ], i R^{(l_3,l_4)}_{j_3j_4}]\cdots, i R^{(l_{t})}_{j_{t}, \text{a}}], i T^{(l_1)}_{\text{a}, j_{1}}\big]\ \ : t\  \text{odd}\
,\\
  \end{aligned} \end{cases}
\end{empheq}
where $c_t=\pm 1 $, depending on $t$.

Next, note that for any $|\psi\rangle\in(\mathbb{C}^d)^{\otimes n}$, it holds that
\begin{align}
 Z_\text{a} (|\psi\rangle|0\rangle_{\text{a}})&=|\psi\rangle|0\rangle_{\text{a}}\ , \\
 D^{(l)}_{j} (|\psi\rangle|0\rangle_{\text{a}})&=(Z^{(l)}_{j} -I)|\psi\rangle\otimes |0\rangle_{\text{a}} \ ,
\\
 D^{(l_1,\cdots, l_t)}_{j_1,\cdots, j_t} (|\psi\rangle|0\rangle_{\text{a}})&=([Z^{(l_1)}_{j_1} \cdots Z^{(l_t)}_{j_t}-Z^{(l_2)}_{j_2} \cdots Z^{(l_t)}_{j_t}]|\psi\rangle)\otimes |0\rangle_{\text{a}} \ , \  \ \ \ \ : t\ge 2\ . 
\end{align}
Considering the linear combinations of the above terms, we find $\forall |\psi\rangle\in(\mathbb{C}^d)^{\otimes n}$,
\begin{align}
\big[D^{(l)}_{j}+Z_\text{a}\big]
  (|\psi\rangle|0\rangle_{\text{a}})&=Z^{(l)}_{j}|\psi\rangle\otimes |0\rangle_{\text{a}} \ , \\
\big[\sum_{r=1}^{t-1} D^{(l_r,\cdots, l_t)}_{j_r,\cdots, j_t}+ D^{(l_t)}_{j_t}+Z_\text{a}\big]
  (|\psi\rangle|0\rangle_{\text{a}})&=(Z^{(l_1)}_{j_1} \cdots Z^{(l_t)}_{j_t} |\psi\rangle)\otimes |0\rangle_{\text{a}} \ ,\ \ \  \  : t\ge 2\ .
\end{align}
Combining this with lemma \ref{lem:diag}, we find that for any  Hermitian operator $H_\text{diag}$ on $(\mathbb{C}^{d})^{\otimes n}$ that is diagonal in the basis  $\{\bigotimes_{j=1}^n |r_j\rangle: r_j=0,\cdots, d-1\}$, there exists an operator $\tilde{H}_\text{diag}$ on $(\mathbb{C}^{d})^{\otimes n}\otimes \mathbb{C}^2$, such that  $i \tilde{H}_\text{diag}$ is in the real Lie algebra $\mathfrak{alg}_\mathbb{R}\Big\{i Z_{\text{a}}\ , i R^{(l)}_{j, \text{a}}  , i R^{(l,l')}_{j,j'}:\  l,l'=1,\cdots , d-1\ ;\ \ j\neq j'=1,\cdots , n
\Big\}$ and
\be
\forall |\psi\rangle\in(\mathbb{C}^d)^{\otimes n}:\ \  \ \ \ \ \tilde{H}_\text{diag}(|\psi\rangle\otimes |0\rangle_\text{a})={H}_\text{diag}|\psi\rangle\otimes |0\rangle_\text{a}\ .
\ee   
Finally, it can be easily shown that the same result remains valid  if  instead of all Hamiltonians $
R^{(l,l')}_{j,j'}:\  l,l'=1,\cdots , d-1\ ;\ \ j\neq j'=1,\cdots , n$, we are restricted to only nearest-neighbor interactions $
R^{(l,l')}_{s,s+1}$, $ l,l'=1,\cdots , d-1\ ; s=1,\cdots, n-1$. The argument  is similar to the argument in Sec.\ref{Sec:Geometry}) for qubits: Combining interactions $R^{(l,l')}_{s,s+1}$ and $Z_s^{(l)}Z_{s+1}^{(l')}$, we can swap the state of qubits defined by   $\{|l-1\rangle_s,|l\rangle_s\}$ in system $j$ and   $\{|l'-1\rangle_{s+1},|l'\rangle_{s+1}\}$ in system $j'$. Furthermore, by combining permutations on nearest-neighbor sites, we can change the order of qubits arbitrarily. Therefore, the additional  restriction to  nearest-neighbor interactions, does not restrict the of Hamiltonians that can be simulated.

This proves lemma \ref{ref:lem23} which is a special case of lemma \ref{lem221} for the case of diagonal Hamiltonians. Next, we prove lemma \ref{lem221} in the general case.

\subsection{All energy-conserving unitaries from diagonal unitaries and 2-local energy-conserving unitaries\\ (Proof of lemma \ref{lem221})}\label{Sec:end}

In this section, we show how a general energy-conserving unitary can be implemented by combining diagonal energy-conserving unitaries and 2-local energy-conserving unitaries. In particular, we  study the Lie algebra of energy-conserving Hamiltonians and show

\begin{lemma}\label{eheheh}
Let $\mathfrak{h}$ be the Lie algebra of energy-conserving skew-Hermitian operators, i.e., those commuting with $H_\text{intrinsic}=\sum_{j=1}^n H_j$. Then, $\mathfrak{h}$ is generated by the set of skew-Hermitian diagonal operators together with operators $i R^{(l,l')}_{j,j+1}:\  l,l'=1,\cdots , d-1\ ;\ \ j=1,\cdots , n-1$, i.e.  
\begin{align}
\mathfrak{h}&\equiv \Big\{A: A+A^\dag=0, [A, \sum_{j=1}^n H_j]=0\Big\}\\ &= \mathfrak{alg}_\mathbb{R}\Big(\big\{i \bigotimes_{j=1}^n |r_j\rangle\langle  r_j|: r_j=0,\cdots, d-1 \big\}\cup \big\{i R^{(l,l')}_{j,j+1}:\  l,l'=1,\cdots , d-1\ ;\ \ j=1,\cdots , n-1\big\}\Big)\ .
\end{align}
\end{lemma}
%

This lemma is a generalization of lemma \ref{Thm5} for the qubit case with U(1) symmetry and can be proven in a similar way.


In the following, we use the notation
\be
|\textbf{r}\rangle\equiv | r_1\rangle\cdots |r_n\rangle=\bigotimes_{j=1}^n |r_j\rangle\ ,\ \ \ \ \ \ \ \ \textbf{r}\in\{0,\cdots, d-1\}^n\ ,
\ee
where $\textbf{r}\equiv r_1\cdots r_n$ and $r_j\in\{0,\cdots, d-1\}$. Then, 
\be
H_\text{intrinsic}=\sum_{j=1}^n  H_j =\sum_{j=1}^n \sum_{r_j=0}^{d-1} (r_j \Delta E)  |r_j\rangle\langle r_j|_j =  \Delta E\sum_{\textbf{r}\in\{0,\cdots, d-1\}^n}  N(\textbf{r})\ |\textbf{r}\rangle\langle \textbf{r}|\ ,
\ee
where
\be
N(\textbf{r})\equiv\sum_{j=1}^n r_j\ .
\ee

\begin{proof}
Any operator $A$ acting on $(\mathbb{C}^d)^{\otimes n}$ that commutes with $H_\text{intrinsic}$ can be written as
\be
A=\sum_{\textbf{r},\textbf{r}': N(\textbf{r})=N(\textbf{r}')} a_{\textbf{r}, \textbf{r}'}\  |\textbf{r}\rangle\langle \textbf{r}'|\ ,
\ee
where the summation is over all $\textbf{r}, \textbf{r}'\in\{0,\cdots, d-1\}^n$  satisfying the condition $N(\textbf{r})=N(\textbf{r}')$. It follows that  $\mathfrak{h}$, the Lie algebra of skew-Hermitian energy-conserving operators, can be written as a linear combinations of 3 sets of operators, namely, 
\begin{align}
\mathcal{D}&\equiv \Big\{i |\textbf{r}\rangle\langle\textbf{r}| : \textbf{r}\in \{0,\cdots, d-1\}^n \Big\} \\
\mathcal{R}&\equiv \Big\{|\textbf{r}\rangle\langle\textbf{r}'|-|\textbf{r}'\rangle\langle\textbf{r}|\ :  \textbf{r}, \textbf{r}'\in \{0,\cdots, d-1\}^n\ ,  N(\textbf{r})=N(\textbf{r}') \Big\}\\ 
\mathcal{I} &\equiv \Big\{  i(|\textbf{r}\rangle\langle\textbf{r}'|+|\textbf{r}'\rangle\langle\textbf{r}|)  \ :  \textbf{r}, \textbf{r}'\in \{0,\cdots, d-1\}^n,  N(\textbf{r})=N(\textbf{r}') \Big\}\ ,
\end{align}
where the constraint $N(\textbf{r})=N(\textbf{r}')$ means that states $|\textbf{r}\rangle$ and $|\textbf{r}'\rangle$ have the same energy. In other words, 
\begin{align}
\mathfrak{h}&\equiv \Big\{A: A+A^\dag=0, [A ,  H_\text{intrinsic}]=0\Big\}=\text{span}_\mathbb{R}(\mathcal{I}\cup \mathcal{D} \cup \mathcal{R}) \ .
\end{align}

For any  distinct pair $\textbf{r}_1, \textbf{r}_2 \in\{0,\cdots, d-1\}^n$,  the following commutation relations hold:
\begin{align}
\Big[ (|\textbf{r}_2\rangle\langle\textbf{r}_1|\mp |\textbf{r}_1\rangle\langle\textbf{r}_2|)\ , |\textbf{r}_1\rangle\langle\textbf{r}_1| \Big]&=|\textbf{r}_2\rangle\langle\textbf{r}_1|\pm |\textbf{r}_1\rangle\langle\textbf{r}_2|\ .\end{align}
This implies that the Lie algebra $\mathfrak{h}$ is generated by $\mathcal{D}$ and $\mathcal{R}$, i.e.
\begin{align}\label{wrg22}
\mathfrak{h}&\equiv \Big\{A: A+A^\dag=0, [A, H_\text{tot}]=0\Big\}=\text{span}_\mathbb{R}(\mathcal{D}\cup \mathcal{I} \cup \mathcal{R}) =\mathfrak{alg}_\mathbb{R}(\mathcal{D}\cup  \mathcal{R}) \ .
\end{align}

Next, we consider the  subset of $\mathcal{R}$ defined by 
\be
\mathcal{R}'=\Big\{ |\textbf{r}\rangle\langle\textbf{r}'|-|\textbf{r}'\rangle\langle\textbf{r}|  \ :  \textbf{r}, \textbf{r}'\in \{0,\cdots, d-1\}^n,  N(\textbf{r})=N(\textbf{r}'), \text{dist}(\textbf{r}, \textbf{r}')=2 \Big\}\ ,
\ee
where for any pair  $\textbf{r}=r_1\cdots r_n$ and $\textbf{r}'=r'_1\cdots r'_n$, we have defined the distance
\be
\text{dist}(\textbf{r}, \textbf{r}')\equiv \sum_{j=1}^n |r_j-r'_j| \ . 
\ee
Note that the two conditions 
\be
N(\textbf{r})=N(\textbf{r}')\  \ \  \ \ \ \text{and}\ \ \  \ \ \ \  \text{dist}(\textbf{r}, \textbf{r}')=2\ ,
\ee
together imply that $r_j=r'_j$ for all systems $j=1,\cdots, n$ except two distinct systems $v$ and $w$, i.e.,  
\be\label{0347}
r_j=r'_j\ \ \ \  :  j\neq v,w \   ; \ \ \ \  r'_v=r_v+1\ \ \ \   \text{and}\ \  r'_w=r_w-1\ . 
\ee

Next, we show that $\mathcal{R}$ can be generated by $\mathcal{R}'$, i.e., 
\be
\mathcal{R}\subset \mathfrak{alg}_\mathbb{R}(\mathcal{R}') ,
\ee
 and therefore 
\begin{align}\label{lkjw20155}
\mathfrak{h}=\mathfrak{alg}_\mathbb{R}(\mathcal{D}\cup \mathcal{R})=\mathfrak{alg}_\mathbb{R}(\mathcal{D}\cup \mathcal{R}') \ .
\end{align}
To see this first note that for any three distinct $\textbf{r}_1, \textbf{r}_2, \textbf{r}_3\in\{0,\cdots, d-1\}^n$,  the following commutation relations hold:
\begin{align}
 \Big[ |\textbf{r}_3\rangle\langle\textbf{r}_2|-|\textbf{r}_2\rangle\langle\textbf{r}_3|  , |\textbf{r}_2\rangle\langle\textbf{r}_1|-|\textbf{r}_1\rangle\langle\textbf{r}_2|  \Big]&= |\textbf{r}_3\rangle\langle\textbf{r}_1|- |\textbf{r}_1\rangle\langle \textbf{r}_3|\  ,
\end{align}
or, equivalently, 
\be\label{arar}
[F(\textbf{r}_3, \textbf{r}_2) , F(\textbf{r}_2, \textbf{r}_1)]=F(\textbf{r}_3, \textbf{r}_1)\ ,
\ee
where we have defined 
\be
F(\textbf{r}', \textbf{r})\equiv |\textbf{r}'\rangle\langle\textbf{r}|-|\textbf{r}\rangle\langle\textbf{r}'|\ .
\ee

Then, consider a pair of 
 $\textbf{r}_\text{in}, \textbf{r}_\text{fin}\in\{0,1,\cdots , d-1\}^n$ satisfying $N(\textbf{r}_\text{in})=N(\textbf{r}_\text{fin})$, which means $|\textbf{r}_\text{in}\rangle$ and $|\textbf{r}_\text{out}\rangle$ have the same energy.  It can be easily seen that  any such pair  can be converted to each other through  a sequence of transitions 
 \be\label{lkjww}
 \textbf{r}_\text{in}=\textbf{r}_1 \longrightarrow  \textbf{r}_2 \longrightarrow   \cdots  \cdots \ \longrightarrow \textbf{r}_m=\textbf{r}_\text{fin}\ ,
 \ee
where any consecutive pairs $\textbf{r}_t $ and $\textbf{r}_{t+1} $, satisfy 
\be
N(\textbf{r}_t)=N(\textbf{r}_{t+1}) \ \ \ \  \text{and}\ \ \ \ \  \text{dist}(\textbf{r}_t , \textbf{r}_{t+1})=2\ \ \ , 1\le t\le m\ .
\ee
 This means that at each step in Eq.(\ref{lkjww}), energy $\Delta E$ is transferred from one system to another. Combining this with  
 Eq.(\ref{arar}), we find
\be
F(\textbf{r}_\text{out},\textbf{r}_\text{in})=[F(\textbf{r}_m, \textbf{r}_{m-1}),[ \cdots [F(\textbf{r}_3, \textbf{r}_2) ,[F(\textbf{r}_3, \textbf{r}_2) ,  [F(\textbf{r}_3, \textbf{r}_2) , F(\textbf{r}_2, \textbf{r}_1)]]]\cdots]]\ .
\ee 
Furthermore, because at each step $N(\textbf{r}_t)=N(\textbf{r}_{t+1})$ and  $\text{dist}(\textbf{r}_t , \textbf{r}_{t+1})=2$, then $F(\textbf{r}_{t+1} , \textbf{r}_{t})\in \mathcal{R}'$. This proves that $\mathcal{R}\subset \mathfrak{alg}_\mathbb{R}(\mathcal{R}')
$, and therefore implies Eq.(\ref{lkjw20155}).
 
Next, we prove that 
 \be
\mathcal{R}'\subset\mathfrak{alg}_\mathbb{R}\Big(\mathcal{D}\cup \big\{i R^{(l,l')}_{j,j'}:\  l,l'=1,\cdots , d-1\ ;\ \ j\neq j'=1,\cdots , n\big\}\Big)\ .
\ee
That is we show that for any  $\textbf{r}, \textbf{r}'\in \{0,\cdots, d-1\}^n$, satisfying   $N(\textbf{r})=N(\textbf{r}')$ and $\text{dist}(\textbf{r}, \textbf{r}')=2$, it holds that
\be
 |\textbf{r}\rangle\langle\textbf{r}'|-|\textbf{r}'\rangle\langle\textbf{r}| \in   \mathfrak{alg}_\mathbb{R}\Big(\mathcal{D}\cup \big\{i R^{(l,l')}_{j,j'}:\  l,l'=1,\cdots , d-1\ ;\ \ j\neq j'=1,\cdots , n\big\}\Big)\ .
\ee 
 To see this note that, as we have seen in Eq.(\ref{0347}), for any pair $\textbf{r}=r_1\cdots r_n$ and $\textbf{r}'=r'_1\cdots r'_n$, the two conditions  $N(\textbf{r})=N(\textbf{r}')$ and $\text{dist}(\textbf{r}, \textbf{r}')=2$ together imply
\be
r_j=r'_j\ \ \ \  :  j\neq v,w \   ; \ \ \ \  r'_v=r_v+1\ \ \ \   \text{and}\ \  r'_w=r_w-1\ . 
\ee

It follows that 
  \be
|\textbf{r}'\rangle=R_{v,w}^{(r_v+1,r_w)}|\textbf{r}\rangle=\frac{1}{2}\Big(|r_v+1\rangle\langle r_v|_v\otimes |r_w-1\rangle\langle r_w|_{w}+|r_v\rangle\langle r_v+1|_v\otimes |r_w\rangle\langle r_w-1|_{w}\Big)|\textbf{r}\rangle  \  .
\ee
This means that
\be
F(\textbf{r}',\textbf{r})= |\textbf{r}'\rangle\langle\textbf{r}|-|\textbf{r}\rangle\langle\textbf{r}'|= \big[ i|\textbf{r}\rangle\langle\textbf{r}|\ ,\ i R_{v,w}^{(r_v+1,r_w)} \big]\ .
\ee
Since $\mathcal{R}'\equiv\big\{ |\textbf{r}\rangle\langle\textbf{r}'|-|\textbf{r}'\rangle\langle\textbf{r}|  \ :  \textbf{r}, \textbf{r}'\in \{0,\cdots, d-1\}^n,  N(\textbf{r})=N(\textbf{r}'), \text{dist}(\textbf{r}, \textbf{r}')=2 \big\}$, we conclude that $\mathcal{R}'\subset\mathfrak{alg}_\mathbb{R}\big(\mathcal{D}\cup \big\{i R^{(l,l')}_{j,j'}:\  l,l'=1,\cdots , d-1\ ;\ \ j\neq j'=1,\cdots , n\big\}\big)$. Combining this with Eq.(\ref{lkjw20155}), we find
\bes\label{oqoqo}
\begin{align}
\mathfrak{h}&\equiv \Big\{A: A+A^\dag=0, [A ,  H_\text{intrinsic}]=0\Big\}\\ &=\text{span}_\mathbb{R}(\mathcal{D}\cup \mathcal{I} \cup \mathcal{R})\\ &=\mathfrak{alg}_\mathbb{R}(\mathcal{D}\cup  \mathcal{R})\\  &=\mathfrak{alg}_\mathbb{R}(\mathcal{D}\cup  \mathcal{R}')\\  &=\mathfrak{alg}_\mathbb{R}  \big(\mathcal{D}\cup \big\{i R^{(l,l')}_{j,j'}:\  l,l'=1,\cdots , d-1\ ;\ \ j\neq j'=1,\cdots , n\big\}\big)\\ &= \mathfrak{alg}_\mathbb{R}\Big(\big\{i \bigotimes_{j=1}^n |r_j\rangle\langle  r_j|: r_j=0,\cdots, d-1 \big\}\cup \big\{i R^{(l,l')}_{j,j'}:\  l,l'=1,\cdots , d-1\ ;\ \ j\neq j'=1,\cdots , n\big\}\Big)   \ . 
\end{align}
\ees
This means that all energy-conserving unitaries can be implemented using diagonal Hamiltonians together with interactions $ \big\{i R^{(l,l')}_{j,j'}:\  l,l'=1,\cdots , d-1\ ;\ \ j\neq j'=1,\cdots , n\big\}$.

Finally, we can easily see that the above conclusion  remains valid if instead of all pairwise interactions $ \big\{R^{(l,l')}_{j,j'}:\  l,l'=1,\cdots , d-1\ ;\ \ j\neq j'=1,\cdots , n\big\}$, one only considers  interactions $ \big\{R^{(l,l')}_{j,j+1}:\  l,l'=1,\cdots , d-1\ ;\ \ j=1,\cdots , n-1\big\}$
on nearest-neighbor systems. Here, we sketch the argument:  Using diagonal Hamiltonians 
together with interactions $ \big\{i R^{(l,l')}_{j,j+1}:\  l,l'=1,\cdots , d-1\ ;\ \ j=1,\cdots , n-1\big\}$, one can implement swap unitaries on nearest-neighbor systems $j$ and $j+1$ , i.e., the unitaries that exchange the state of the two nearest-neighbor systems.  To see this first note that for any pair of neighbor systems $j$ and $j+1$ the swap unitary $S_{j,j+1}$ that exchanges the state of  systems $j$ and $j+1$, can be implemented using Hamiltonians 
\begin{align}
&|r\rangle\langle r|_j \otimes |r'\rangle\langle r'|_{j+1} :\ \ \ \ \  r,r'=0,1,\cdots, d-1\ \  \ \ \ \text{and}\ \ \ \\ &R^{(l,l')}_{j,j+1}=|l-1\rangle\langle l|_j \otimes |l'\rangle\langle l'-1|_{j+1}+|l\rangle\langle l-1|_j \otimes |l'-1\rangle\langle l'|_{j+1} : \ \ \ \   \  l,l'=1,\cdots , d-1\ .
\end{align}
In other words, $S_{j,j+1}$ is in the Lie group associated to the Lie algebra 
\be
\mathfrak{alg}_\mathbb{R}\Big\{i |r\rangle\langle r|_j \otimes |r'\rangle\langle r'|_{j+1} : r,r'=0,1,\cdots, d-1,\ \ ; i R^{(l,l')}_{j,j+1} : l,l'=1,\cdots , d-1\Big\}\ .
\ee
This follows, for instance, using the above result in Eq.(\ref{oqoqo}) in the special case of $n=2$, and the fact that  $S_{j,j+1}$ is an energy-conserving unitary.  Since this holds for all $j=1,\cdots, n-1$, and since swaps of nearest-neighbor systems generate all possible permutations of the systems, we conclude that all permutations are in the Lie group associated to the Lie algebra  
\be\label{yryr}
\mathfrak{alg}_\mathbb{R}  \big(\mathcal{D}\cup \big\{i R^{(l,l')}_{j,j+1}:\  l,l'=1,\cdots , d-1\ ;\ \ j=1,\cdots , n-1\big\}\big)  \ .
\ee
But, since the Lie algebra is closed under the adjoint action of the Lie group, it follows that  the above Lie algebra is closed under all permutations of $n$ systems. This, in particular, implies that for any pair of distinct systems $j$ and $j'$, operator $i R^{(l,l')}_{j,j'}$ is in the  Lie algebra in Eq.(\ref{yryr}), for all $j\neq j'=1,\cdots, n$ and $l,l'=1,\cdots, n-1$. We conclude that
\begin{align}
&\mathfrak{alg}_\mathbb{R}  \big(\mathcal{D}\cup \big\{i R^{(l,l')}_{j,j+1}:\  l,l'=1,\cdots , d-1\ ;\ \ j=1,\cdots , n-1\big\}\big)\\ &= \mathfrak{alg}_\mathbb{R}  \big(\mathcal{D}\cup \big\{i R^{(l,l')}_{j,j'}:\  l,l'=1,\cdots , d-1\ ;\ \ j\neq j'=1,\cdots , n\big\}\big)=\mathfrak{h} \ .
\end{align}

This implies lemma \ref{eheheh}, and therefore completes the proof of lemma  \ref{lem221}.

\end{proof}

\newpage

\bibliography{Ref_2020}

\end{document}